\documentclass[11pt,letterpaper]{article}

\usepackage{amsmath,amsthm,nicefrac}
\usepackage{amsfonts, amstext}

\usepackage{typearea}
\typearea{14}
\usepackage[font={small,it}]{caption}

\usepackage{setspace}

\usepackage{enumitem}
\usepackage[breaklinks]{hyperref}
\hypersetup{colorlinks=true,%
            citebordercolor={.6 .6 .6},linkbordercolor={.6 .6 .6},%
citecolor=blue,urlcolor=black,linkcolor=blue,pagecolor=black}

\usepackage[nameinlink]{cleveref}
\Crefname{algocf}{Algorithm}{Algorithms}
\crefname{algocfline}{line}{lines}
\Crefname{invariant}{Invariant}{Invariants}
\Crefname{claim}{Claim}{Claims}
\Crefname{subclaim}{Subclaim}{Subclaims}

\usepackage{epsfig}
\usepackage{amsthm,amssymb}
\usepackage{amsthm,amssymb}

\usepackage{mathrsfs}
\usepackage{xspace}
\usepackage{soul}
\usepackage{latexsym}
\usepackage{bbm}

\usepackage{framed}
\renewenvironment{leftbar}[1][\hsize]
{%
    \MakeFramed{\hsize#1\advance\hsize-\width\FrameRestore}%
}
{\endMakeFramed}

\usepackage[dvipsnames]{xcolor}
\definecolor{DarkGray}{rgb}{0.66, 0.66, 0.66}
\definecolor{DarkPowderBlue}{rgb}{0.0, 0.2, 0.6}
\definecolor{fluorescentyellow}{rgb}{0.8, 1.0, 0.0}

\usepackage[ruled,vlined,linesnumbered,algonl]{algorithm2e}
\SetEndCharOfAlgoLine{}
\SetKwComment{Comment}{\footnotesize$\triangleright$\ }{}

\SetCommentSty{mycommfont}

\usepackage{thmtools,thm-restate}

\allowdisplaybreaks

\newcommand{\rqq}{\ell_q}

\newcommand{\alert}[1]{{\color{red}#1}}

\newcounter{note}[section]
\renewcommand{\thenote}{\thesection.\arabic{note}}
\newcommand{\agnote}[1]{\refstepcounter{note}$\ll${\bf Anupam~\thenote:}
  {\sf \color{blue} #1}$\gg$\marginpar{\tiny\bf AG~\thenote}}

\sethlcolor{fluorescentyellow}
\newcommand{\myhl}[1]{\hspace*{-\fboxsep}\colorbox{yellow}{\parbox{5.5in}{%
    #1%
    }}}

\newcommand{\initOneLiners}{%
    \setlength{\itemsep}{0pt}
    \setlength{\parsep }{0pt}
    \setlength{\topsep }{0pt}
}
\newenvironment{OneLiners}[1][\ensuremath{\bullet}]
    {\begin{list}
        {#1}
        {\initOneLiners}}
    {\end{list}}

\newtheorem{theorem}{Theorem}[section]
\newtheorem{lemma}[theorem]{Lemma}
\newtheorem{claim}[theorem]{Claim}

\newtheorem{corollary}[theorem]{Corollary}

\newtheorem{invariant}{Invariant}
\newtheorem{subclaim}[theorem]{Subclaim}

\theoremstyle{definition}
\newtheorem{defn}[theorem]{Definition}

\theoremstyle{remark}
\newtheorem{remark}[theorem]{Remark}

\makeatletter 
\renewcommand{\theinvariant}{(I\@arabic\c@invariant)}
\makeatother

\newcommand{\kser}{{\textsf{$k$-Server}}\xspace}
\newcommand{\ksertw}{{\textsf{$k$-ServerTW}}\xspace}

\newcommand{\wpage}{{\textsf{Weighted Paging}}\xspace}

\newcommand{\sse}{\subseteq}
\newcommand{\la}{\lambda}
\newcommand{\ga}{\gamma}
\newcommand{\level}{{\textsf{level}}}
\newcommand{\loc}{{\mathsf{loc}}}
\newcommand{\leafreq}{{\mathsf{EarliestLeafReq}}}
\newcommand{\inh}{{\mathsf{inh}}}

\newcommand{\cost}{{\textsf{cost}}}
\newcommand{\lcost}{{\textsf{logcost}}}

\newcommand{\rs}{{\cR^s}}
\newcommand{\ZZ}{\mathbb{Z}}

\newcommand{\rt}{\mathbbm{r}}
\newcommand{\rootvtx}{\rt}

\newcommand{\lca}{\operatorname{lca}}
\newcommand{\poly}{\operatorname{poly}}

\renewcommand{\emptyset}{\varnothing}

\newcommand{\floor}[1]{\lfloor#1\rfloor}

\newcommand{\fD}{\mathfrak{D}}
\newcommand{\fL}{\mathfrak{L}}
\newcommand{\fM}{\mathfrak{M}}
\newcommand{\fMtw}{\mathfrak{M}_{TW}}
\newcommand{\fR}{\mathfrak{R}}
\newcommand{\fT}{\mathfrak{T}}
\newcommand{\chF}{{\cF^{ch}}}

\newcommand{\cA}{{\mathcal {A}}}

\newcommand{\cF}{{\cal F}}

\newcommand{\cL}{\mathcal{L}}

\newcommand{\cR}{\mathcal{R}}

\newcommand{\cT}{{\mathcal{T}}}

\newcommand{\suppl}{{\mathsf{supply}}}
\newcommand{\DU}{\textsc{FullUpdate}\xspace}
\newcommand{\DUZ}{\textsc{SimpleUpdate}\xspace}

\newcommand{\FL}{{\textsc{FindLeaves}}\xspace}
\newcommand{\BT}{{\textsc{BuildTree}}\xspace}
\newcommand{\BW}{{\textsc{BuildWitness}}\xspace}

\newcommand{\RL}{\textsf{ReqLoc}}
\DeclareMathOperator{\RI}{\textsf{ReqInt}}

\newcommand{\y}[1]{y^{#1}}

\newcommand{\yv}{\y{v}}
\newcommand{\yu}{\y{u}}

\newcommand{\cons}{{\fL}}%
\newcommand{\constr}{{\cal C}}

\newcommand{\sib}{{\sf activesib}}

\newcommand{\ist}{{i^\star}}
\newcommand{\ust}{{u^\star}}
\newcommand{\vstar}{{v^\star}}
\newcommand{\tst}{{\tau^\star}}
\newcommand{\Sst}{{S^\star}}

\newcommand{\awake}{{\sf Awake}}

\newcommand{\unmarked}{{\sf prev}}
\newcommand{\prev}{\unmarked}
\newcommand{\fll}{{\sf fill}}
\newcommand{\loss}{{\sf loss}}
\newcommand{\nf}{\nicefrac}

\newcommand{\btau}{\pmb{\tau}}
\newcommand{\Diff}{D}

\newcommand{\junk}[1]{}
\newcommand{\eat}[1]{}

\newif\ifhideproofs

\ifhideproofs
\usepackage{environ}
\NewEnviron{hide}{}

\fi

\newenvironment{subproof}[1][\proofname]{%
  \begin{proof}[#1]%
}{%
  \end{proof}%
}

\def\Mbound{\frac{5H \la^H k}{4 \gamma} + 1}

\begin{document}
\title{A Hitting Set Relaxation for $k$-Server \\ and an Extension to
  Time-Windows}

\author{
{Anupam Gupta\thanks{Computer Science Department, Carnegie Mellon University, Pittsburgh, PA. Email: {\tt anupamg@cs.cmu.edu}.}}
\and
{Amit Kumar\thanks{Department of Computer Science and Engineering, IIT Delhi, New Delhi, India. Email: {\tt amitk@cse.iitd.ac.in}.}}
\and
{Debmalya Panigrahi\thanks{Department of Computer Science, Duke University, Durham, NC. Email: {\tt debmalya@cs.duke.edu}.}}
}

\maketitle

\begin{abstract}
  We study the $k$-server problem with time-windows. In this problem,
  each request $i$ arrives at some point $v_i$ of an $n$-point metric space 
  at time $b_i$ and comes
  with a deadline $e_i$. One of the $k$ servers must be moved to $v_i$
  at some time in the interval $[b_i, e_i]$ to satisfy this
  request. We give an online algorithm for this problem with a
  competitive ratio of $\poly\log (n,\Delta)$, where $\Delta$
  is the aspect ratio of the metric space. Prior to our work, the
  best competitive ratio known for this problem was
  $O(k \poly \log(n))$ given by Azar et al. (STOC 2017).

  Our algorithm is based on a new covering linear program
  relaxation for $k$-server on HSTs. This LP naturally corresponds to 
  the min-cost flow formulation of $k$-server, and easily extends to the
  case of time-windows. We give an online algorithm for obtaining a 
  feasible fractional solution for this LP, and a primal dual analysis
  framework for accounting the cost of the solution. Together, they yield
  a new $k$-server algorithm with poly-logarithmic competitive ratio, 
  and extend to the time-windows case as well. Our principal technical 
  contribution lies
  in thinking of the covering LP as yielding a {\em truncated}
  covering LP at each internal node of the tree, which allows us to keep
    account of server movements across subtrees.
  We hope that this LP relaxation and the algorithm/analysis will 
  be a useful tool for addressing $k$-server and related problems.

\end{abstract}



\section{Introduction}
\label{sec:introduction}

The \kser problem, originally proposed by Manasse, McGeoch, 
and Sleator~\cite{ManasseMS90}, is perhaps the most well-studied
problem in online algorithms. Given an $n$-point
metric space and an online sequence of requests at various locations, the goal
is to coordinate $k$ servers so that each request is served by moving
a server to the corresponding location. The objective of the algorithm
is to minimize the total distance moved by the servers (i.e., the movement cost). 
It has been known for 
more than two decades that the best deterministic competitive ratio for
this problem is
between $k$~\cite{ManasseMS90} and $2k-1$~\cite{KoutsoupiasP95}, although 
determining the exact constant remains open. For randomized algorithms,
even obtaining a tight asymptotic bound is still open, although
there has been tremendous progress in the last decade culminating in a 
poly-logarithmic competitive
ratio~\cite{BBMN11,BubeckCLLM18,BuchbinderGMN19}. 

We focus on the \emph{$k$-server with time-windows}
(\ksertw) problem, where each request arrives at a location in the metric space at 
some time $b$ with a deadline $e \ge b$. 
The algorithm must satisfy
the request by moving a server to that location at any point during
this time interval $[b, e]$. (If $e=b$ for every request, this reduces
to \kser.) The techniques used to solve the standard \kser
problem seem to break down in the case of time-windows. Nonetheless,
an $O(k \poly \log n)$-competitive deterministic algorithm was given for the case
where the underlying metric space is a tree~\cite{AzarGGP17}; this gives an
$O(k \poly \log n)$-competitive \emph{randomized} algorithm for arbitrary metric
spaces using metric embedding results. 

For the special case of \ksertw on an unweighted star,
\cite{AzarGGP17} obtained competitive ratios of $O(k)$ and $O(\log k)$ using
deterministic and randomized algorithms respectively. The deterministic 
competitive ratio of $O(k)$ extended to weighted stars as well (which is same as \wpage),
but a randomized (poly)-logarithmic bound already turned out
to be more challenging; a bound of $\poly\log(n)$ was obtained 
only recently~\cite{Gupta0P20}. 
This raises the natural question: can we obtain a poly-logarithmic
competitive ratio for the \ksertw problem on general metric spaces?
The technical
gap between \wpage and \kser is substantial and bridging this gap for
randomized algorithms was the preeminent challenge in online
algorithms for some time. Moreover, the approaches eventually used to bridge 
this gap do not seem to extend to time-windows, so we have to devise a new 
algorithm for \kser as well in solving \ksertw.
We successfully answer this question. 

\begin{theorem}[Randomized Algorithm]
  \label{thm:main-tw}
  There is an $O(\poly\log (n\Delta))$-competitive randomized algorithm for
  \ksertw on any $n$-point metric space with aspect ratio~$\Delta$.
\end{theorem}

\Cref{thm:main-tw} follows from our main technical result
\Cref{thm:frac-tw} below. Indeed, since
any $n$-point metric space can be probabilistically approximated using
$\lambda$-HSTs with height $H = O(\log_\lambda \Delta)$ and expected
stretch $O(\lambda \log_\lambda n)$~\cite{FakcharoenpholRT04}, we can set
$\lambda = O(\log \Delta)$ and use the rounding algorithm
from~\cite{BBMN11,BubeckCLLM18} to complete the reduction.

\begin{theorem}[Fractional Algorithm for HSTs]
  \label{thm:frac-tw}
  Fix $\delta' \leq \nf1{n^2}$. There is an
  $O(\poly (H, \la, \log n))$-competitive fractional algorithm for \ksertw
  using $\frac{k}{1-\delta'}$ servers such that for any instance on a $\la$-HST with height $H$ and $\lambda \geq 10H$, and for
  each request interval $R = [b,e]$ at some leaf $\ell$ in this instance, there is a
  time in this interval at which the number of servers at $\ell$ is at
  least~$1$.
\end{theorem}

Apart from the result itself, a key contribution of our paper is an
approach to solve a new covering linear program for \kser.
Previous results in \kser (e.g.,~\cite{BubeckCLLM18}) used a very
different LP relaxation, and it remains unclear how to extend that
relaxation to the case of time-windows. The covering LP in this paper
is easy to describe and flexible. It is quite natural, following from
the min-cost LP formulation for \kser (see \S\ref{sec:min-cost-relation}).  We hope that this relaxation,
and indeed our online algorithm and  accounting framework for obtaining a feasible solution
will be useful for other related problems.

\subsection{Our Techniques}
\label{sec:techniques}

The basis of our approach is a restatement of \kser (and thence
\ksertw) as a covering LP without box constraints. This LP has
variables $x(v,t)$ that try to capture the event that a server leaves
the subtree rooted at $v$ at some time $t$. There are several complications
with this LP: apart from having an exponential number
of constraints, it is too unstructured to directly tell us how to
move servers. E.g., the variable for a node may increase but that for its
parent or child edges may not. Or the online LP solver may increase
variables for timesteps in the past, which then need to be translated to
server movements at the present timestep.

Our principal technical contribution is to view this new LP as
yielding ``truncated'' LPs, one for each internal node $v$ of the
tree. This ``local'' LP for $v$ restricts the original LP to inequalities
and variables corresponding to the subtree below $v$. This truncation
is contingent on prior decisions taken by the algorithm, 
and so the constraints obtained may not be implied by those
for the original LP. However, we show how the primal---and just as
importantly---the dual solutions to local LPs can be composed to give
primal/dual solutions to the original LP. These are then crucial for
our accounting.

The algorithm for
\kser  proceeds as follows. Suppose a request comes at some leaf
$\ell$, and suppose $\ell$ has less than $1-\delta'$ amounts of server
at it (else we deem it satisfied):
\begin{enumerate}[leftmargin=2mm]
\item Consider a vertex $v_i$ on the \emph{backbone} (i.e., the path
  $\ell = v_0, v_1, \ldots, v_H = \rootvtx$ from leaf $\ell$ to the
  root $\rootvtx$). If $v_i$ has off-backbone children whose descendant
  leaves contain non-trivial amounts of server, we
  move servers from these descendants to $\ell$ until the total
  server movement has cost roughly some small quantity $\gamma$. Since
  the cost of server movement grows exponentially up the tree, and the
  movement cost is roughly the same for each $v_i$, more server mass is moved from
  closer locations. Since there are $H$ levels in the HST,
  the total movement cost is roughly $H \gamma$. This concludes one
  ``round'' of server movement. 
  This server movement is now repeated over multiple rounds until $\ell$
  has $1-\delta'$ amount of server at it. (This can be thought of as a
  discretization of a continuous process.)

\item To account for server movement at node $v_i$, we raise both
  primal and dual variables of the local LP at $v_i$. The primal
  increase tells us which children of $v_i$ to move the servers
  from. The dual increase allows us to account for the server
  movement. Indeed, we ensure that the total dual increase for the
  local LP at each $v_i$---and hence by our composition operations,
  the dual increase for the global LP---is also approximately $\gamma$
  in each round. Moreover, we show this dual scaled down by
  $\beta \approx O(\log n)$ is feasible. This means that the
  $O(H\gamma)$ cost of server movement in each round can be
  approximately charged to this increase of the global LP dual, giving us
   $H\beta = O(H \log n)$-competitiveness.

\item The choice of dual variables to raise for the local LP at node
  $v$ is dictated by the corresponding dual variables for the children
  of $v$. Each constraint in the local LP at $v$ is
  composed from the local constraints at some of its children. It is
  possible that there are several constraints at $v$ that are
  composed using the same constraint at a child $u$ of $v$. We
  maintain the invariant that the total dual values of the former is
  bounded by the dual value of the latter. Now, we can only raise
  those dual variables at $v$ where there is some slack in this
  invariant condition.

\end{enumerate}

Finally, to extend our results to \ksertw, we say that a request
$(\ell, I=[b,q])$ becomes critical (at time $q$) if the amount of
server mass at $\ell$ at any time during $I$ was at most
$1 - \delta'$. We proceed as above to move server mass to
$\ell$. However, after servicing $\ell$, we also service active
request intervals at nearby leaves: we service these {\em piggybacked
  requests} according to (a variation of) the earliest deadline rule while
ensuring that the total cost incurred remains bounded by (a factor
times) the cost incurred to service $\ell$. We use ideas
from~\cite{AzarGGP17} (for the case of $k=1$) to find this tour, but
we need a new dual-fitting-based analysis of this algorithm. Moreover,
new technical insights are needed to fit this dual-fitting analysis
(which works only for $k=1$) with the rest of our analytical
framework. Indeed, the power of our LP relaxation for \kser lies in
the ease with which it extends to $\ksertw$.

\subsection{Roadmap}

In~\S\ref{sec:lp}, we describe the covering LP relaxation for both
\kser and \ksertw. In~\S\ref{sec:trunclp} we define the notion of
``truncated'' constraints used to define local LPs at the internal
nodes of the HST, and show how constraints for the children's local
LPs can be composed to get constraints for the parent LP. We then give
the algorithm and analysis for the \kser problem
in~\S\ref{sec:algodesc} and~\S\ref{sec:analysis} respectively:
although we could have directly described the algorithm for \ksertw,
it is easier to understand and build intuition for the algorithm for
\kser first, and then see the extension to the case of
time-windows. This extension appears in~\S\ref{sec:main-windows}: the
algorithm is similar to that in~\S\ref{sec:algodesc}, the principal
addition being the issue of \emph{piggybacked} requests. We give the
analysis in \S\ref{sec:tw-analysis}: many of the ideas
in~\S\ref{sec:analysis} extend easily, but again new ideas are needed
to account for the piggybacked requests. We conclude with some open
problems in~\S\ref{sec:open}.

\subsection{Related Work}
\label{sec:related-work}

The \kser problem is %
arguably the most prominent problem in online algorithms. Early work
focused on deterministic algorithms~\cite{FiatRR94,KoutsoupiasP95},
and on combinatorial randomized algorithms~\cite{Grove91,BartalG00}.
\kser has also been studied
for special metric spaces, such as lines, (weighted) stars, trees:
e.g.,~\cite{chrobak1991new,chrobak1991optimal,FiatKLMSY91,McGeochS91,AchlioptasCN00,BansalBN12,Seiden01,CoteMP08,CsabaL06,BansalBN12a,BansalBN10}.
\cite{BorodinE98} gives more background on the
\kser problem.
Works obtaining poly-logarithmic competitive ratio are more recent,
starting with %
\cite{BansalBMN15}, and more recently, by
\cite{BubeckCLLM18} and~\cite{Lee18}; this resulted in the first
$\poly\log k$-competitive algorithm.
(\cite{BuchbinderGMN19} gives an alternate projection-based
perspective on \cite{BubeckCLLM18}.)  A new LP
relaxation was introduced by \cite{BubeckCLLM18}, who then use a
mirror descent strategy with a multi-level entropy regularizer to
obtain the online dynamics. However, it is unclear how to extend their
LP when there are time-windows, even for the case of star metrics.
Our competitive ratio for \kser on HSTs is $\poly\log (n\Delta)$
as against just $\poly \log (k)$ in their work, but this weaker bound
is in exchange for a more flexible algorithm/analysis that extends to time-windows.

Online algorithms where requests can be served within some time-window
(or more generally, with delay penalties) have recently been given for
matching~\cite{EmekKW16, AshlagiACCGKMWW17, AzarCK17},
TSP~\cite{AzarV16}, set cover~\cite{AzarCKT20}, multi-level aggregation~\cite{BienkowskiBBCDFJSTV16,BuchbinderFNT17,AzarT19},
$1$-server~\cite{AzarGGP17,AzarT19}, 
network design~\cite{AzarT20}, etc. %
The work closest to ours is that of~\cite{AzarGGP17} who show
$O(k \log^3 n)$-competitiveness for \kser with general delay
functions, and leave open the problem of getting poly-logarithmic
competitiveness. Another related work is 
\cite{Gupta0P20} who show $O(\log k \log n)$-competitiveness for \wpage, which is the same as \kser with delays for weighted star
metrics. This work also used a hitting-set LP: this was based on two different kinds of extensions of the
request intervals and was very tailored to the star metric, and is
unclear how to extend it even to $2$-level trees. Our new LP relaxation is
more natural, being implied by the min-cost flow relaxation for \kser,
and extends to time-windows.

Algorithms for the online set cover problem were first given
by~\cite{AAABN03}: this led to the general primal-dual approach for
covering linear programs (and sparse set-cover
instances)~\cite{BN-MOR}, and to sparse CIPs~\cite{GN12-mor}. Our
algorithm also uses a similar primal-dual approach for the local LPs
defined at each node of the tree; we also need to crucially use the
sparsity properties of the corresponding set-cover-like constraints.


\section{A Covering LP Relaxation}
\label{sec:lp}



For the rest of the paper,
we consider the \kser problem %
on hierarchically well-separated trees (HSTs) with $n$ leaves, rooted
at node $\rt$ and having height $H$. (The standard extension to general metrics
via tree embeddings was outlined in \S\ref{sec:introduction}.)
Define the {\em level} of a node as its combinatorial height, with the leaves at level $0$, and the root at level $H$.  For a
non-root node $v$, the length of the edge $(v, p(v))$ going to its
parent $p(v)$ is $c_v := \lambda^{\text{level}(v)}$.
So leaf edges have length $1$, and edges between the root and its
children have length $\lambda^{H-1}$. We assume that $\lambda \geq
10H$. 
For a vertex $v$, let $\chi_{v}$ be its children, $T_v$ be the subtree
rooted at $v$, and $L_v$ be the leaves in this subtree. Let $n_v := |T_v|$.
For a subset $A$ of  nodes of  a tree $T$, let $T^A$ denote the
minimal subtree of $T$ containing the root node and set $A$, i.e., the subtree 
consisting of all nodes in $A$ and their ancestors. 

\paragraph{Request Times and Timesteps.}
Let the request sequence be $\fR := r_1, r_2, \ldots$. For \kser, each request $r_i \in \fR$ is a tuple $(\ell_{q_i}, q_i)$
for some leaf $\ell_{q_i}$ and  
distinct {\em request time} $q_i \in \ZZ_+$, such that $q_{i-1} < q_i$
for all $i$.
In \ksertw
each request $r_i$ is a tuple $(\ell_{e_i}, I_i=[b_i, e_i])$ 
for a leaf $\ell_i$ and (request) interval $I_i=[b_i,e_i]$ with
\emph{arrival/start time} $b_i$ and \emph{end time} $e_i$. The algorithm sees this request $r_i$ at time $b_i$; again
$b_{i-1} < b_i$ for all $i$.
A solution must ensure that a server visit $\ell_i$ during interval
$I_i$.
The set of all starting and ending times of intervals are called {\em
  request times}; we assume these are distinct
integers.
\footnote{\kser (without
  time-windows) can be modeled by time-intervals of length $1$,
  where each $e_i = b_i+1$.} 

Between any two \emph{request times} $q$ and $q+1$, we define a large
collection of \emph{timesteps} (denoted by $\tau$ or $t$)---these
timesteps take on values $\{q + i\eta\}$ for some small value
$\eta \in (0,1)$. (Each request arrival time is also a timestep). We
use $\fT$ to denote the set of timesteps.  Our fractional
algorithm moves a small amount of server to the request location 
$r_q$ at some of the timesteps $t \in [q, q+1)$.
Given a timestep $\tau$, let $\floor{\tau}$ refer to 
the request time $q$ such that $\tau \in [q,q+1)$. 

\subsection{The Covering LP Relaxation}

We first give a covering LP relaxation for \kser, and then generalize
it to \ksertw. 
Consider an instance of \kser specified by an HST and a request sequence $r_1, r_2, \ldots$. 
Our LP relaxation $\fM$
has variables
$x(v,t)$ for every non-root node $v$ and timestep $t$, where 
$x(v,t)$ indicates the amount of server traversing the edge from $v$ to its parent
$p(v)$ at timestep $t$. %
The objective function
is
\begin{gather}
  \sum_{v \neq \rt} \sum_t c_v \; x(v,t).  \label{eq:obj-fn}
\end{gather}
There are exponentially many constraints. Let $A$ be a subset
of leaves. Let $\btau := \{\tau_u\}_{u \in T^A}$ be a set of timesteps for each node in $T^A$, i.e., nodes in $A$ and their ancestors.\footnote{We use boldface $\btau$ to denote a vector of timesteps, and $\tau_u$ to be the value of this vector for a vertex $u$.} These timesteps
must satisfy two conditions: (i) each (leaf) $\ell \in A$ has a
request at time $\floor{\tau_\ell}$, and (ii) for each internal node
$u \in T^A$, $\tau_u = \max_{\ell \in A \cap T_u} \tau_\ell$;
i.e., $\tau_u$ is the latest timestep assigned to a leaf in $u$'s
subtree by $\btau$.  For the tuple $(A, \btau)$, the LP relaxation contains the
constraint $\varphi_{A,\btau}$:
\begin{align}
\label{eq:basic-cons}
  \sum_{v \in T^A, v \neq \rt} x(v, (\tau_v, \tau_{p(v)}]) \geq |A|-k.
\end{align}
Define $x(v,I) := \sum_{t \in I} x(v,t)$ for any interval $I$. We now prove
validity of these constraints. (In \S\ref{sec:min-cost-relation} we show  these
constraints are implied by the usual min-cost flow formulation for
\kser, giving another proof of validity.)

\begin{figure}
    \centering
    \includegraphics[width=2.5in]{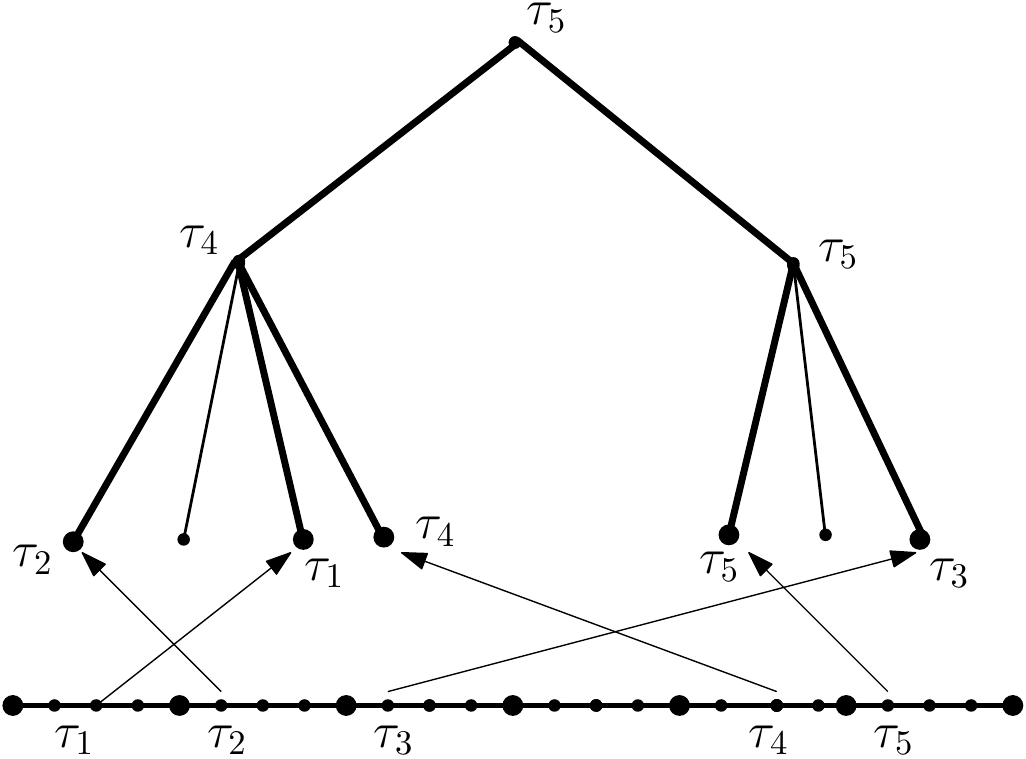}
    \caption{\small{Example of a  tuple $(A, \pmb{\tau})$: the set $A$ is given by the leaves in bold, the tree $T^A$ by the bold edges, and $\pmb{\tau}$ is shown against each vertex in $T^A$. Bold circles on the timeline denote request arrival times, the dots denote timesteps. Each shown timestep has the corresponding request arriving at the shown leaf at the arrival time (bold dot) preceding it.}}
    \label{fig:lp}
\end{figure}

\begin{claim}
  \label{cl:lp}
  The linear program $\fM$ is a valid relaxation for the \kser problem.
\end{claim}
\begin{proof}
  Consider a solution to the $k$-server instance that ensures that for a
  request at a leaf $\ell$ at time $q$, there is a server at $\ell$ at
  time $q$. We assume that this solution has the {\em eagerness}
  property---if leaves $\ell$ and $\ell'$, requested at times $q$ and
  $q'$ respectively, are two consecutive locations visited by a
  server, %
  the server moves from $q$ to $q'$ at timestep $q+\eta$ (which is less than  $q'$). 
  
  Now for a constraint of the form~(\ref{eq:basic-cons}),
  let $A_1, A_2, \ldots, A_k$ be the subsets of $A$ that are served by
  the different servers (some of these sets may be empty). Define
  $x_i(v,t) \gets 1$ if server $i$ crosses the edge $(v,p(v))$
  \emph{from $v$ to $p(v)$ (i.e., upwards)} at timestep $t$, and $0$ otherwise. We show that
  \[ \sum_{v \in T(A_i), v \neq \rt} x_i(v, (\tau_v,
    \tau_{p(v)}]) \geq |A_i|-1. \]
  Defining
  $x(v,t) := \sum_i x_i(v,t)$ by summing over all $i$ gives~(\ref{eq:basic-cons}).
For any server $i$ and set $A_i$, define $E'$ to be the edges
  $(v,p(v))$ for which $x_i(v, (\tau_v, \tau_{p(v)}]) = 1$. If
  $|E'| < |A_i|-1$, then deleting the edges in $E'$ from the tree %
  leaves a connected component $C$ with at least two vertices from $A_i$.
  Server $i$ serves at least two leaf
  vertices $C \cap A_i$, say $v,w$, %
  requested at times $q_v = \floor{\tau_v}, q_w=\floor{\tau_w}$
  respectively. Say $q_v < q_w$, and let $u$ be the least common
  ancestor of $v,w$. Notice that $\tau_u \geq \tau_v$, and if
  the path from $v$ to $u$ is labeled $v_0=v, v_1, \ldots, v_h = u$,
  then the intervals $(\tau_{v_i}, \tau_{v_{i+1}}]$ partition
  $(\tau_v, \tau_u]$. Since the server is at $v$ at timestep $\tau_v$
  (by the construction above) and
  is at $w$ at time $q_w \leq \tau_w$, there must be an edge $(v_i, v_{i+1})$
  such that it crosses this edge upwards during
  $(\tau_{v_i}, \tau_{v_{i+1}}]$. Then this edge should be in $E'$, a
  contradiction. 
\end{proof}


\begin{remark}
  \label{rem:few-timesteps}
We could have replaced the constraint
  \eqref{eq:basic-cons} by its simpler version 
  involving $x_i(v, (q_v, q_{p(v)}])$, where $q_v := \floor{\tau_v}$:
  that would be valid and sufficient. However,  since our algorithm works at the level of
  timesteps, it is convenient to use~\eqref{eq:basic-cons}.
\end{remark}

\paragraph{Extension to Time-Windows.} We now extend these ideas to
\ksertw. In constraint~\eqref{eq:basic-cons} for a pair $(A,\btau)$,
the timesteps for ancestors of (a leaf in) $A$ could be inferred from
the values assigned by $\tau$ to $A$. We now generalize this by (i)
allowing $A$ to contain non-leaf nodes, as long as they are
independent (in terms of the ancestor-descendant relationship), and
(ii) the timestep assigned to an internal node is {\em at least} that
of each of its descendants in $A$. Formally, consider a tuple 
$(A, f, \btau)$, where $A$ is a subset of tree nodes such that no two of them
have an ancestor-descendant relationship, the function $f: A \to  \fR$
maps each node $v \in A$ to a 
request $(\ell_v, [b_v, e_v])$ given by a leaf $\ell_v \in T_v$ and an
interval $[b_v, e_v]$ at $\ell_v$, and the assignment $\btau$ maps
each node $u \in T^A$ to a timestep 
$\tau_u$ satisfying the following two (monotonicity)
properties:
\begin{OneLiners}
\item[(a)] For each node $v \in T^A$, $\tau_v \geq \max_{u \in A \cap T_v} e_u.$
\item[(b)] If $v_1, v_2$ are two  nodes in $T^A$ with $v_1$ being the ancestor of $v_2$, then $\tau_{v_1} \geq \tau_{v_2}.$
\end{OneLiners}
Given such a tuple
$(A, f,\btau)$, 
we define
the constraint $\varphi_{A,f,\btau}$ \footnote{The condition $v \neq
  \rt$ in the first summation is invoked only when $A = \{\rt\}$, in which case the LHS is empty.} 
\begin{align}
  \label{eq:basic-consnew}
  \sum_{v \in A, v \neq \rt} x(v, (b_v, \tau_{p(v)}]) + \sum_{v \in T^A \setminus A, v \neq \rt} x(v, (\tau_v, \tau_{p(v)}]) \geq |A|-k.
\end{align}
Note the differences with constraint~\eqref{eq:basic-cons}: the
LHS for a node $v \in A$ has a longer interval starting from $b_v$ instead of from $\tau_v$. %
Also,~\eqref{eq:basic-consnew} does not use the timesteps
$\{\tau_v\}_{v \in A}$: these will be useful later in defining the
truncated constraints. In the special case of \kser where $e_v = b_v
+1$, the above constraint is similar to~\eqref{eq:basic-cons}, though
the terms for nodes in $A$ differ slightly.  
The objective function is the same as~(\ref{eq:obj-fn}).
We denote this LP by $\fMtw$. 

\begin{figure}
    \centering
    \includegraphics[width=3.5in]{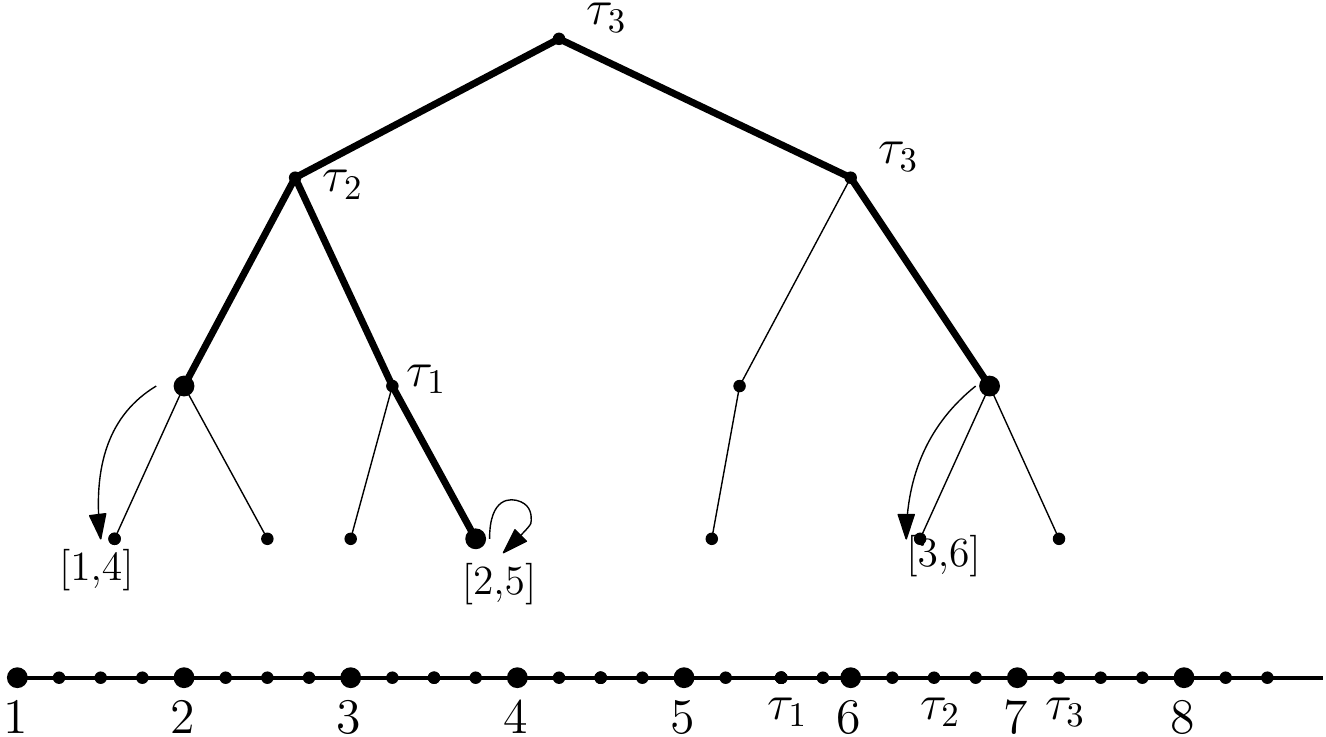}
    \caption{\small{Example of a tuple $(A,f, \pmb{\tau})$ for \ksertw: the set $A$ is given by the bold nodes, the tree $T^A$ by the bold edges, and $\pmb{\tau}$ is shown against each internal vertex in $T^A$. Arrows indicate the mapping $f$ to a leaf request interval.}}
    \label{fig:lpTW}
\end{figure}

\begin{restatable}{claim}{twrelax}
  \label{cl:lp-tw}
  The linear program $\fMtw$ is a valid relaxation for \ksertw. 
\end{restatable}

\begin{proof}
Consider a solution to the instance that
    ensures a server moves only when a request becomes critical
    (although at this time, it can serve several outstanding
    requests).  Now for a constraint of the
    form~(\ref{eq:basic-consnew}), let $L$ denote the set of leaves
    corresponding to the nodes in $A$.  let $L_1, L_2, \ldots, L_k$ be
    the subsets of $L$ that are served by the different servers (some
    of these sets may be empty), and let $A_i$ be the subset of $A$
    corresponding to $L_i$. Define $x_i(v,t) \gets 1$ if server $i$
    crosses the edge $(v,p(v))$ at time $t$, and
    $x(v,t) := \sum_i x_i(v,t)$. We show that
    $$
     \sum_{v \in A, v \neq \rt} x_i(v, (b_v, \tau_{p(v)}]) + \sum_{v \in T^A \setminus A, v \neq \rt} x_i(v, (\tau_v, \tau_{p(v)}]) \geq |A_i|-1;$$
recall that $b_v, v \in A,$ is the starting time of the request interval given by $f(v)$. 
  Summing the above inequality over all $i$ gives~(\ref{eq:basic-consnew}). 
  
  For sake of brevity, let $I_v$ denote the interval $(b_v, \tau_{p(v)}]$ or $(\tau_v, \tau_{p(v)}]$ depending on whether $v \in A$. 
  Define $E'$ as the set of edges
  $(v,p(v))$ for which $x_i(v,I_v) \geq 1$. We need to show that $|E| \geq |A_i| - 1$. Suppose not.  
  Then deleting the edges in $E'$ from the tree $T$
  leaves a connected component with at least two vertices from $A_i$.
  
  Call this component $C$, and let $u, v$ be two distinct vertices in $A_i \cap C$.  
  Let $f(u)$ and $f(v)$ be $(\ell_u, R_u=(b_u,e_u]), (\ell_v, R_v=(b_v,e_v])$ respectively. 
  Let $w$ be the lca of $\ell_{u}$ and $\ell_{v}$. Note that $w$ is also the lca of $u$ and $v$.
  Suppose server $i$ satisfies $R_{u}$ before satisfying $R_{v}$. We claim that the server $i$ reaches $w$ at some time during $(b_{u}, \tau_{w}]$. To see this, we consider two cases: 
  \begin{itemize}
      \item Server $i$ visits $u$ at time $e_u$ when the request $R_{u}$ becomes critical: Since it reaches $v$ by time $e_v$, 
       it must have visited $w$ during $(e_{u}, e_{v}] \subseteq (b_u, \tau_{w}] $. 
      \item Server $i$ visits $u$ before $R_u$ becomes critical. In this case, it would have visited $u$ strictly after $b_{u}$ ((because all start and end times of requests are distinct). Since it reaches $w$ at or before $e_{v} \leq \tau_{w}$, the desired statement holds in this case as well. 
  \end{itemize}
  
  Let the sequence of nodes in $B$ from $u$ to $w$ be
  $v_0 = u, v_1, \ldots, v_h = w.$ Note that all the edges $(v_i, p(v_i)), i < h, $ lie below $w$ and so are not in $E'$. Observe that the intervals $I_{v_i}, i=0, \ldots, h-1$,  
  partition $(b_u, \tau_{w}]$. As outlined in the two cases above, the server $i$ leaves $u$ strictly after $b_u$ and reaches $w$ by time $\tau_{w}$. Therefore, there must be an edge $(v_i, p(v_{i})), i < h$, such that it crosses this edge during $I_{v_i}$. Then this edge should be in $E'$, a contradiction.
\end{proof}

\section{The Local LPs: Truncation and Composition}
\label{sec:trunclp}
We maintain a collection of \emph{local} LPs $\fL^v$, one for each
internal vertex $v$ of the tree. While the constraints of local LPs
for the non-root nodes are not necessarily valid for the original \kser instance,
those in the local LP $\fL^\rootvtx$ 
are implied by constraints of $\fM$ or $\fMtw$. This
gives us a handle on the optimal cost. The constraints in the local
LP at a node are related to those in its children's local LPs, 
allowing us to relate their primal/dual solutions, and %
their costs. %

To define the local LPs, we need some notation. 
Our (fractional) algorithm $\cA$ moves server mass around over
timesteps. In the local LPs, we define constraints based on the state
of our algorithm $\cA$. Let $k_{v,t}$ be the server mass that $\cA$ has in $v$'s
subtree $T_v$ at timestep $t$ (when $v$ is a leaf, this is the amount of server mass at $v$ at timestep $t$). We choose three non-negative parameters $\delta, \delta', \gamma$.
The first two help define lower and upper bounds on the amount of
(fractional) servers at any leaf, and $\gamma$ denotes the granularity at which movement of server mass happens. We ensure $\delta'
\gg \delta \gg \gamma$, and set $\delta' = \frac{1}{n^2}, \delta = \frac{1}{10n^3}, \gamma = \frac{1}{n^4}$.

\begin{defn}[Active and Saturated Leaves]
  Given an algorithm $\cA$, a leaf $\ell$ is \emph{active} if it has at
  least $\delta$ amount of server (and \emph{inactive} otherwise). The
  leaf is \emph{saturated} if $\ell$ has more than $1-\delta'$
  amount of server (and \emph{unsaturated} otherwise).
\end{defn}

The server mass at each location
should ideally lie in the interval $[\delta, 1-\delta']$, but since we move servers in
discrete steps,
we maintain the following (slightly weaker) invariant:

\begin{leftbar}
  \begin{invariant}
    \label{invar:active}
    The server mass at each leaf lies in the interval
    $[\nicefrac\delta2, 1-\nicefrac{\delta'}2].$
  \end{invariant}
\end{leftbar}

Constraints of $\fL^v$ are defined using \emph{truncations} of
the constraints $\varphi_{A,\btau}$. For a node $v$ and subset of
nodes $A$ in $T$, let the subtree $T_v^A$ be the minimal subtree of
$T_v$ containing $v$ and all the nodes in $A \cap T_v$. 
\begin{defn}[Truncated Constraints]
Consider a node $v$, a subset $A$ of leaves in $T$ and a set $\btau := \{\tau_u\}_{u \in T^A_v}$ of timesteps satisfying the conditions: (i) each (leaf) $\ell \in A$ has a request at time $\floor{\tau_\ell}$, and (ii) for each internal node $u \in T^A_v, \tau_u = \max_{\ell \in A} \cap T_v$. The {\em truncated constraint}  $\varphi_{A,\btau,v}$ is defined as: 
\begin{align}
\label{eq:implied1}
 \sum_{u: u \neq v, u \in T^A_v} \yv(u, (\tau_u,
    \tau_{p(u)}]) \geq |A \cap T_v| - k_{v, \tau_v} -2\delta(n-n_v) ;
\end{align}
recall that $k_{v, \tau_v}$ is the amount of server mass  in $T_v$ \emph{at the
end of} timestep $\tau_v$. 
We say that the truncated constraint $\varphi_{A,\btau,v}$ {\em ends at} $\tau_v$.
\end{defn}
The truncated constraint $\varphi_{A,\btau,v}$ can be thought of as truncating an actual LP constraint of the form~\eqref{eq:basic-cons} for the nodes in $T_v^A$ only. 
One subtle  difference is the last term that weakens the constraint slightly; 
we will see in \Cref{cor:comp} that this weakening is crucial.
The truncated constraint $\varphi_{A,f,\btau,v}$ in case of $\ksertw$ is defined analogously: 
given a node $v$, a tuple $(A, f, \btau)$ satisfying the conditions stated above~\eqref{eq:basic-consnew} with the restriction that $A$ lies in $T_v$ and $\btau$ is defined for nodes in $T_v^A$ only, 
the truncated constraint $\varphi_{A,f,\btau,v}$ (ending at $\tau_v$) is defined as (see~\Cref{def:twtrunc} for a formal definition): 
\begin{align}
 \label{eq:implied1new}
  \sum_{u \in A \cap T_v, u \neq v} \yv(u, (b_u, \tau_{p(u)}]) + \sum_{u \in T^A_v \setminus A, u \neq v} \yv(u, (\tau_u, \tau_{p(u)}]) \geq |A \cap T_v| - k_{v, \tau_v} -2\delta(n-n_v)
\end{align}

A few remarks about the truncation: first, this truncated
constraint uses \emph{local variables} $\yv$ that are ``private'' for
the node $v$ instead of the global variables $x$. In fact, we can
think of $x$ as denoting variables $\y{\rootvtx}$  local
to the root, and therefore $\varphi_{A,\btau,\rootvtx} = \varphi_{A,\btau}$
(or $\varphi_{A,f,\btau,\rootvtx} = \varphi_{A,f,\btau}$). Second, a
truncated constraint is \emph{not necessarily implied} by the LP
relaxation $\fM$ (or $\fMtw$) even when we replace $\yv$ by
$x$, since a generic algorithm is not constrained to 
maintain $k_{v,\tau_v}$ servers in subtree $T_v$ after timestep $\tau_v$.
But, at the root (i.e., when $v = \rootvtx$), we always have
$k_{v, \tau_v} = k$  and the last term is $0$, so replacing 
$y^{\rootvtx}$ by $x$ in its constraints
gives us constraints  of the form~\eqref{eq:basic-cons} 
from the actual LP.

\begin{defn}[$\bot$-constraints]
  \label{def:bot}
  A truncated constraint where
  $|A|=1$ is called a \emph{$\bot$-constraint}.
\end{defn}
Such $\bot$-constraints play a special role when a subtree has only one
active leaf, namely the requested leaf. 
In the case of \kser, if $|A|= 1$ then the
constraint \eqref{eq:implied1} has no terms on the LHS but a positive
RHS,
so it can never be satisfied. Nevertheless,
such constraints will be useful when forming new constraints by composition.

\paragraph{Composing Truncated Constraints.} 
The next concept is that of {\em constraint composition}:  a truncated constraint $\varphi_{A,\btau,v}$ can be obtained from the corresponding truncated constraints for the children of $v$. 
Consider a subset $X$ of $v$'s
children. For $u \in X$, let $C_u := \varphi_{A(u), \btau(u) ,u }$ be a
constraint in $\cons^u$ ending at $\tau_u := \tau(u)_u$, given by some linear inequality
$\langle a^{C_u}, \y{u} \rangle \geq b^{C_u}$. Then defining $A :=
\cup_{u \in X} A(u)$ and $\tau : T^A \to \fT$ obtained
by extending maps
$\btau(u)$ %
and setting $\tau_v = \max_{u \in X} \tau_u$,  the constraint
$\varphi_{A,\btau,v}$ is written as:
\footnote{The vector $a^{C_u}$ has
  one coordinate for every node in $T^A_u$, whereas $y^v$ has one
  coordinate for each node in $T^A_v \supseteq T^A_u$. We
   define the inner product $\langle a^{C_u}, \yv \rangle$ by adding
  extra coordinates (set to $0$) in the vector $a^{C_u}$. }
\begin{align}
  \label{eq:compose1}
  \sum_{u \in X} \bigg(\yv(u, (\tau_u, \tau_v]) + \langle a^{C_u}, \yv \rangle
  \bigg) \geq \sum_{u \in X} b^{C_u} - \Big( k_{v,\tau_v} - \sum_{u \in X}
  k_{u,\tau_u} \Big) + 2 \delta \Big(n_v - \sum_{u \in X} n_u \Big).
\end{align}
The constraints $\varphi_{A(u),\btau(u),u}$ used their local
variables $\yu$, whereas this new constraint uses $\yv$. Every
constraint in $\cons^v$ can be obtained this way, and so the constraints of
$\fL^{\rootvtx}$ (which are implied by $\fM$) can be obtained by
recursively composing truncated constraints
for its children's local LPs. In case of \ksertw, the composition operation holds for the
constraints $\varphi_{A,f,\btau,v}$: a minor change is that the terms in LHS involving a vertex $u \in A$ have $\yv(u,
(b_u, \tau_v])$, where $b_u$ is the starting time of the request
corresponding to $f(u)$. (We see the details later in~\eqref{eq:consnew}.)

\subsection{Constraints in Terms of Local Changes}

The local constraints~(\ref{eq:implied1}) and the composition rule
(\ref{eq:compose1}) are written in terms of $k_{u, \tau_u}$, the amount
of server that our algorithm $\cA$ places at various locations and
times. It will be more convenient to rewrite them in terms of
server movements in $\cA$.

\begin{defn}[$g,r,\Diff$]
  For a vertex $v$ and timestep $t$, let the \emph{give} $g(v,t)$ and
  the \emph{receive} $r(v,t)$ denote the total
  (fractional) server movement \emph{out of}  and \emph{into} the subtree $T_v$ on the edge
  $(v,p(v))$ at timestep $t$. 
  For interval
  $I$, let $g(v,I) := \sum_{t \in I} g(v,t)$ and define $r(v,I)$
  similarly, and define the ``difference'' $\Diff(v,I) := g(v,I) - r(v,I)$.
\end{defn}

Restating the composition rule in terms of the quantities $\Diff$ defined above shows the utility of the extra term on the RHS of the truncated constraint.
\begin{restatable}{lemma}{compositionrule}
\label{cor:comp}
Consider a vertex $v$,  a timestep $\tau$ and a subset $X$ of children of $v$ such that at timestep $\tau$ all active leaves in $T_v$ are descendants of the nodes in $X$. For each $u \in X$, consider a truncated constraint $C_u := \varphi_{A(u), \btau(u) ,u }$  given by some linear inequality
$\langle a^{C_u}, \y{u} \rangle \geq b^{C_u}$.
Define $(A, \btau)$ as in~\eqref{eq:compose1} with $\tau := \tau_v$, and assume \Cref{invar:active} holds. 
Then the truncated constraint $\varphi_{A,\btau,v}$ from~\eqref{eq:compose1}
implies the inequality~\footnote{When $y \geq 0$, a constraint $\langle a, y \rangle
  \geq b$ is said to \emph{imply} a constraint  $\langle a', y
  \rangle \geq b'$ if $a \leq a'$ and $b \geq b'.$}:
  \begin{gather}
    \label{eq:compose2}
    \sum_{u \in X} \bigg(\yv(u, (\tau_u, \tau_v]) + \langle a^{C_u}, \yv \rangle
    \bigg) \geq \sum_{u \in X} \bigg( \Diff(u,
    (\tau_u,\tau_v]) + b^{C_u} \bigg) + 
    \big(n_v - \sum_{u \in X} n_u\big) 
    \delta,
  \end{gather}
  We call this the {\em composition rule}.
  An analogous statement   holds for a  tuple $(A, f, \btau)$ for a
  vertex $v$ in the case of \ksertw, except that $\tau_u$ is
  replaced by $b_u$ for every vertex $u \in A$ on the LHS (see~\eqref{eq:consnew}).
\end{restatable}

\begin{proof}
Note that
    $k_{v,\tau_v} 
    = \sum_{u \in X} k_{u,\tau_v} + \sum_{w \not \in X} k_{w,\tau_v}
    = \sum_{u \in X} k_{u, \tau_u} - \sum_{u \in X} \Diff(u, (\tau_u, \tau_v]) + \sum_{w \not \in X} k_{w,\tau_v}.$
    in~\eqref{eq:compose1} gives
\begin{align*}
  \sum_{u \in X} \bigg(\yv(u, (\tau_u, \tau_v]) + \langle a^{C_u}, \yv \rangle
  \bigg) \geq \sum_{u \in X} \Big( \Diff(u, (\tau_u, \tau_v]) + b^{C_u} \Big) 
  - \sum_{w \not \in X} k_{w,\tau_v}
   + 2 \delta \Big(n_v - \sum_{u \in X} n_u \Big).
\end{align*}
Finally, since all active leaves in $T_v$ at timestep $\tau_v$ are descendants of $X$,
\Cref{invar:active} implies that
$\sum_{w \not \in X} k_{w,\tau_v} \le \delta  \sum_{w \notin X} n_w \leq
\delta \left( n_v - \sum_{u \in X} n_u \right)$. This is where
  the weakening in~(\ref{eq:implied1}) is useful.
\end{proof}

\subsection{Timesteps and Constraint Sets}

Recall that $\fT$ is the set of all timesteps. For each vertex $v$ we
define a subset $\cR(v) \sse \fT$ of \emph{relevant} timesteps, such
that
the local LP $\fL^v$ contains a non-empty set of constraints
$\cons^v(\tau)$ for each $\tau \in \cR(v)$. \alert{Should we say what the
  variables are in this LP?}
Each constraint in
$\cons^v(\tau)$ is of the form $\varphi_{A,\btau,v}$ for a tuple
$(A,\btau)$ ending at $\tau$.
Overloading notation, let
$\cons^v := \bigcup_{\tau \in \cR(v)} \cons^v(\tau)$ denote the set of all
constraints in the local LP at $v$. %
The objective function of this local LP is
$\sum_{u\in T_v, \tau} c_u \; \yv(u,\tau)$. \alert{What does $\tau$
  sum over?}

The timesteps in $\cR(v)$ are partitioned into $\cR^s(v)$ and
$\cR^{ns}(v)$, the \emph{solitary} and \emph{non-solitary} timesteps
for $v$. The decision whether a timestep belongs to $\cR(v)$ is made
by our algorithm. and is encoded by adding $\tau$ to either $\cR^s(v)$
or $\cR^{ns}(v)$. For each timestep $\tau \in \cR^s(v)$, the algorithm
creates a constraint set $\cons^v(\tau)$ consisting of a single
$\bot$-constraint (recall \Cref{def:bot}); for each timestep
$\tau \in \cR^{ns}(v)$ it creates a constraint set $\cons^v(\tau)$
containing only
non-$\bot$-constraints obtained
by composing constraints from $\cons^w(\tau_w)$ for some
children $w$ of $v$ and timesteps $\tau_w \in \cR(w)$, where
$\tau_w \leq \tau$.

For each $\tau$,
a constraint $C \in \cons^v(\tau)$ corresponds 
to a dual variable
$z_C$, which is raised only at timestep $\tau$. We ensure the
following invariant. %

\begin{leftbar}
  \begin{invariant}
    \label{invar:dual-raised}
    At the end of each timestep $\tau \in \cR^{ns}(v)$, the objective function value of the dual variables corresponding
    to constraints in $\cons^v(\tau)$ equals $\gamma$. I.e., if a
    generic constraint $C$ is given by
    $\langle a^C \cdot \yv \rangle \geq b^C$, then
    \begin{align}
      \label{eq:dualinvariant}
      \sum_{C \in \cons^v(\tau)} b^C \cdot z_C = \gamma \qquad \forall
      \tau \in \cR^{ns}(v). \tag{I2}
    \end{align}
    Furthermore, $b^C > 0$ for all $C \in \fL^v(\tau)$ and $\tau \in \cR(v)$.
  \end{invariant}
\end{leftbar}
No dual variables $z_C$ are defined for $\bot$-constraints, and
(the first statement of) \Cref{invar:dual-raised} does not apply to timesteps $\tau \in \cR^s(v)$.
In the following sections, we show how to maintain a dual solution that is feasible for $\fD^v$ (the dual LP for $\fL^v$) when scaled
down by some factor $\beta = \poly \log (n\lambda)$.

\paragraph{Awake Timesteps.} For a vertex $v$, we maintain a subset
$\awake(v)$ of \emph{awake} timesteps. The set $\awake(v)$ has the
property that it contains all the solitary timesteps, i.e., $\cR^s(v)$,
and some 
non-solitary ones. Hence
$\cR^s(v) \sse \awake(v) \sse \cR^{s} \cup \cR^{ns}(v) = \cR(v)$.
Whenever we add a timestep to $\cR(v)$, we initially add it to
$\awake(v)$; some of the non-solitary ones subsequently get removed.
A timestep $\tau$ is awake for vertex $v$ at some moment in the
algorithm if it belongs to $\awake(v)$ at that moment. For any vertex
$v$, define
\begin{gather}
  \prev(v, \tau) := \arg\max\{\tau' \in \awake(v) \mid \tau' \leq \tau \} \label{eq:prev}
\end{gather}
Note that as the set $\awake(v)$ evolves over time, so does the identity
of $\prev(v,\tau)$. We show in~\Cref{cl:unmarkeddefined} that
$\prev$ is well-defined for all relevant $(v, \tau)$
pairs. \alert{Motivate this better?}

\paragraph{Starting configuration.} At the beginning of the algorithm,
assume that the root has $2k$ ``dummy'' leaves as children, each of
which has server mass $\nicefrac{1}{2}$ at time $q=0$. All other leaves of the tree
have mass $\delta/2$. (This ensures \Cref{invar:active} holds.) No
requests arrive at any dummy leaf $v$; moreover, we add a
$\bot$-constraint $\varphi_{A, \btau, v}$, where $A = \{v\}$ and
$\tau_v =0$. \alert{Should we say why?} Assuming this starting
configuration only changes the cost of our solution by at most an
additive term of $O(k\Delta)$, where $\Delta$ is the aspect ratio of the
metric space.


\section{Algorithm for \texorpdfstring{\kser}{k-server}}
\label{sec:algodesc}

We now describe our algorithm for \kser. At request time $q$, the
request arrives at a leaf $\rqq$. The \emph{main procedure} calls \emph{local
  update} procedures for each ancestor of $\rqq$. Each such local update
possibly moves servers to $\rqq$, and also adds constraints to the
local LPs and raises the primal/dual values to account for this
movement. We use $\RL(\tau)$ to denote the  location of request with deadline at time $\floor{\tau}$, i.e., $\ell_{\floor{\tau}}$.

\subsection{The Main Procedure}
\label{sec:main-basic}

In the main procedure of \Cref{algo:main}, let the
\emph{backbone} be the leaf-root path $\rqq = v_0, v_1, \ldots, v_H = \rt$. %
We move servers to $\rqq$
from other leaves until it is saturated: this server movement happens
in small discrete increments over several timesteps. Each
iteration of the \textbf{while} loop in line~\eqref{l:while}
corresponds to a distinct timestep $\tau$. %
Let %
$\sib(v, \tau)$ be the siblings $v'$ of
$v$ %
with active leaves in
their subtrees $T_{v'}$ (at timestep $\tau$). 
Let $i_0$ be the smallest index with non-empty $\sib(v_{i_0}, \tau)$. 
The procedure $\DUZ$ adds a $\bot$-constraint to each of the sets $\fL^{v_i}(\tau)$
for $i = 0, \ldots, i_0$. For $i > i_0$, the procedure $\DU$ adds
(non-$\bot$) constraints to $\fL^{v_i}(\tau)$. If $\sib(v_i, \tau)$ is
non-empty, it also transfers some servers from the subtrees below
$\sib(v_i, \tau)$ to $\rqq$. 
\medskip

\begin{algorithm}[H]
  \caption{Main Procedure}
  \label{algo:main}
  \ForEach{$q = 1,2, \ldots$}{ get request $r_q$; let the path from
    $r_{q}$ to the root be $\rqq = v_0, v_1, \ldots, v_H =
    \rt$.\; 
    $\tau \gets q+\eta$, the first timestep after $q$. \;
    \While{$k_{v_0, \tau} \leq 1 - \delta'$ \label{l:while} }{
      let $i_0 \gets$ smallest index such that $\sib(v_{i_0},
      \tau) \neq \emptyset$.       \label{l:eq}  \\
        \lFor{$i=0,\ldots, i_0$}{call $\DUZ(v_i, \tau)$.
        }
        \lFor{$i=i_0+1, \ldots, H$\label{l:ffor1}}{call
          $\DU(v_i,\tau)$.\label{l:call-dual} 
        }
        $\tau \gets \tau+\eta$. \tcp*[f]{move to the next
          timestep} \label{l:ffor3}
      }
    }
\end{algorithm}

\subsection{The Simple Update Procedure}
\label{sec:simple-local-basic}

This procedure adds timestep $\tau$ to both $\cR^s(v)$ and $\awake(v)$, and  
creates a $\bot$-constraint in the  LP $\fL^v$.

\begin{algorithm}
  \caption{$\DUZ(v, \tau$)}
  \label{algo:simpledual}
  let $v_0 \gets \RL(\tau)$. \;
  add timestep $\tau$ to the event set
  $\cR^s(v)$ and to $\awake(v)$. \label{l:addp} \tcp*[f]{``solitary'' 
timestep for $v$} \;
  $\cons^{v}(\tau) \gets$ the $\bot$-constraint
  $\varphi_{A,\btau,v}$, where $A = \{v_0\}$ and \label{l:addp1}
  $\tau_w = \tau$ for nodes $w$ on the $v_0$-$v$ path. 
\end{algorithm}

\subsection{The Full Update Procedure}
\label{sec:full-local-basic}

The $\DU(v,  \tau$) procedure is called for backbone nodes $v$
that are above $v_{i_0}$ (using the notation of~\Cref{algo:main}).
It has two objectives. First, it transfers servers to the requested leaf
node  $v_0$ from the subtrees of the off-backbone children of $v$,
incurring a total cost of at most $\gamma$. Second, it defines the
constraints $\cons^{v}(\tau)$ and runs a primal-dual update on these
constraints until the total dual value raised is exactly
$\gamma$. This dual increase is at least the server transfer cost, which we
use to bound the algorithm's cost. We now explain the steps of
\Cref{algo:dual} in more detail. (The notions of
slack and depleted constraints are in \Cref{def:slack}.)

\begin{algorithm}
  \caption{$\DU(v, \tau$)}
  \label{algo:dual}
  let $h \gets \level(v)-1$ and $u_0 \in \chi_{v}$ be child containing
  the current request $v_0 := \RL(\tau)$. \;
  let $U \gets \{u_0\} \cup \sib(u_0,\tau)$; say $U = \{u_0, u_1,
  \ldots, u_\ell\} $, $L_U \leftarrow $ active leaves below $U \setminus \{u_0\}$. \;
  add timestep $\tau$ to event set $\cR^{ns}(v)$ and to  $\awake(v)$. \tcp*[f]{``non-solitary''
    timestep for $v$}\;
  set timer $s \gets 0$. \;
  \Repeat{the dual objective corresponding to constraints in $\cons^v(\tau)$ becomes $\gamma$.}{
    \For{$u \in U$}{
      let $\tau_u \gets \prev(u,\tau)$ and $I_u = (\tau_u,\tau]$.
          \label{l:Iu} \;
      let $C_u$ be a slack constraint in $\cons^u(\tau_u)$. \tcp*[f]{slack constraint exists since $\prev(u,\tau)$ is awake}
      \label{l:chooseC}
    }
    let $\sigma \gets (C_{u_0}, C_{u_1}, \ldots, C_{u_\ell})$ be the resulting tuple of
    constraints. \;
    add new constraint $C(v,\sigma, \tau)$ to constraint set $\cons^v(\tau)$.\;

    \While{all constraints $C_{u_j}$ in $\sigma$ are slack
      \textbf{and} dual objective for $\cons^v(\tau)$ less than $\gamma$}{ \label{l:dual}
      increase timer $s$ at uniform rate. \;
      increase $z_{C(v,\sigma, \tau)}$ at the same rate as $s$. \label{l:draise1}\;
      for all $u \in U$, define $S_u := I_u \cap \left( \cR^{ns}(u) \cup \{\tau_u+\eta\} \right).$ \label{l:su} \;
      increase $\yv(u, t)$ for $u \in U, t \in S_u$
      according to 
      $\frac{d \yv(u,t)}{d s} = \frac{\yv(u,t)}{\lambda^{h}} +
        \frac{\gamma}{Mn \cdot \lambda^h}$.     \label{l:draise2}
      \;
      \underline{\textbf{transfer}} server mass from $ T_u$ into $v_0$ at rate $\frac{d
        \yv(u,I_u)}{d s} + \frac{
          b^{C_u}}{\lambda^h}$ using the leaves in $L_U \cap T_u$, for each $u \in U \setminus \{u_0\}$  \label{l:dtr} \;
    }
    \ForEach{constraint $C_{u_j}$ that is depleted}{
      \label{l:mark2}
      \lIf{\textit{all} the constraints in $\cons^{u_j}(\tau_{u_j})$
        are depleted}{remove $\tau_{u_j}$ from $\awake(u_j)$.     }
    }
  }
\end{algorithm}

Consider a call to $\DU(v, \tau)$ with $u_0$ being the child of $v$ on
the path to the request $v_0$ (See~\Cref{fig:du}). Each iteration of the {\bf repeat} loop
adds a constraint $C$ to $\cons^v(\tau)$ and raises the dual variable
$z_C$ corresponding to it. For each node $u$ in
$U := \{u_0\} \cup \sib(u_0, \tau)$, define $\tau_u :=\prev(u, \tau)$
to be the most recent timestep 
currently in
$\awake(u)$. This timestep $\tau_u$ may move backwards over the
iterations as nodes are removed from $\awake(u)$ in 
line~\eqref{l:mark2}. One exception is the node $u_0$: we will show that $\tau_{u_0}$ stays equal to $\tau$ for the entire run of
$\DU$. Indeed, we add $\tau$ to $\awake(u_0)$ during
$\DUZ(u_0, \tau)$ or $\DU(u_0, \tau)$ before calling $\DU(v, \tau)$,
and \Cref{cl:u0} shows that $\tau$ stays awake in $\cR(u_0)$ during
$\DU(v, \tau$).

\begin{enumerate}[leftmargin=5mm]
\item %
  We add a constraint $C(v,\sigma,\tau)$ to
  $\cons^v(\tau)$ %
  by taking one constraint $C_u \in \cons^u(\tau_u)$ for each
  $u \in U$ and setting $\sigma := (C_1, \ldots, C_{|U|})$.  (The
  choice of constraint from $\cons^u(\tau_u)$ is described in
  item~\ref{item:DU3} below.) Each %
  $C_u$ has form $\varphi_{A(u), \btau(u), u}$ ending at $\tau_u :=
  \btau(u)_u$ for some tuple
  $(A(u),\btau(u))$. The new constraint $C(v,\sigma, \tau)$ is the composition
  $\varphi_{A,\btau,v}$ as in (\ref{eq:compose1}),
  where $I_u := (\tau_u, \tau]$.
  Since $U$ contains all the children of $v$ whose subtrees contain
  active leaves at $\tau$, the set $A = \cup_u A(u)$ and the $\btau$ obtained by
  extending the $\btau(u)$ functions 
  both satisfy the conditions of~\Cref{cor:comp}, which shows that
  $\varphi_{A(u), \btau(u), u}$ implies:
  \begin{align}
    \label{eq:cons}
    \underbrace{\sum_{u \in U} \left( \yv(u,
    I_u) + a^{C_u} \cdot \yv \right)}_{a^{C(v,\sigma,\tau)} \cdot \yv} \geq 
    \underbrace{\sum_{u \in U} \left( \Diff(u, I_u) + b^{C_u} \right)  + 
    (n_v - \sum_{u \in U} n_u) \delta}_{\leq b^{C(v,\sigma, \tau)}}.
  \end{align}

\item Having added constraint $C(v,\sigma,\tau)$, we raise the
  new dual variable $z_{C(v, \sigma, \tau)}$ at a constant rate in
  line~\eqref{l:draise1}, and the primal variables $\yv(u, t)$ 
  for each $u \in U$ and any $t$ in some index set $S_u$ 
  using %
  an exponential update rule in
  line~\eqref{l:draise2}. The index set $S_u$ consists of all
  timesteps in $I_u \cap \cR^{ns}(u)$ and the first timestep of
  $I_u$---which is $\tau_u + \eta$ if $I_u$ is
  non-empty.\footnote{This timestep may not belong to
    $\cR(u)$, but all other timesteps in $S_u$ lie in $\cR(u)$; see
    also~\Cref{fig:du}.} We will soon show that $S_u$ is not too large, yet captures all the ``necessary'' variables that should be raised (see~\Cref{fig:du}).
  Moreover, we transfer servers from active leaves in $T_u$ into
  $\RL(q)$ in line~\eqref{l:dtr}. This transfer is done arbitrarily,
  i.e., we move servers out of any of the leaf nodes that were active
  at the beginning of this procedure. Our definition of
  $\sib(u_0, \tau)$ means that $T_u$ has at least one active leaf and
  hence at least $\delta$ servers to begin with. Since we move at most
  $\gamma \ll \delta$ amounts of server, we maintain
  \Cref{invar:active}, as shown in \Cref{cl:rec}. The case of $u_0$ is
  special: since
  $\tau_{u_0} = \tau$, the interval $I_{u_0}$ is empty so no variables
  $\yv(u_0, t)$ are raised.

  Somewhat unusually for an online primal-dual algorithm, both the
  primal and dual variables are used to account for our algorithm's
  cost, and not for actual algorithmic decisions (i.e., the server
  movements). This allows us to increase primal variables from the
  past, even though the corresponding server movements are always
  executed at the current timestep.
\end{enumerate}

To describe the stopping condition for this process, we need to
explain the relationships between these local LPs, and define the
notions of \emph{slack} and \emph{depleted} constraints.  We use the
fact that we have an almost-feasible dual solution
$\{z_C\}_{C \in \cons^u(\tau_u)}$ for each $u \in U$. This in turn
corresponds to an increase in primal values for variables
$\yu(u',\tau')$ in $\fL^u$. It will suffice for our proof to ensure
that when we raise $z_{C(v, \sigma, \tau)}$, we constrain it as
follows:
\begin{leftbar}
  \begin{invariant}
    \label{invar:duals-match}
    For every $u \in \chi_{v}, t \in \cR^{ns}(u),$ and every constraint
    $C \in \cons^u(t)$ (which by definition of $\cR^{ns}(u)$ is not a $\bot$-constraint):
    \begin{align}
      \label{eq:bound}
      \left( 1 + \frac{1}{H} \right) z_C \geq \sum_{\tau' \geq t} \sum_{\sigma:
      C \in \sigma} z_{C(v, \sigma, \tau')}. \tag{I3}
    \end{align}
  \end{invariant}
\end{leftbar}

\begin{defn}[Slack and Depleted Local Constraints]
\label{def:slack}
  A non-$\bot$ constraint $C \in \cons^u$ is \emph{slack}
  if~(\ref{eq:bound}) is satisfied with a strict inequality, else it
  is \emph{depleted}. By convention,
  $\bot$-constraints are always slack.
\end{defn}

We can now explain the remainder of the local update.

\begin{enumerate}[leftmargin = 5mm]
  \setcounter{enumi}{2}
\item \label{item:DU3} %
  The choice of the constraint in line~\eqref{l:chooseC} is now easy:
  $C_u$ is chosen to be any slack constraint in $\cons^u(\tau_u)$. If
  $\tau_u \in \cR^s(u)$, this is the unique $\bot$-constraint in
  $\cons^u(\tau_u)$.

  The primal-dual update in the \textbf{while} loop proceeds as long
  as all constraints $C_u$ in $\sigma$ are slack: once a constraint
  becomes tight, some other slack constraint $C_u \in \cons^u(\tau_u)$
  is chosen to be in $\sigma$. If there are no more slack constraints
  in $\cons^u(\tau_u)$, the timestep $\tau_u$ is removed from the
  awake set (in line~\eqref{l:mark2}).  In the next iteration,
  $\tau_u$ gets redefined to be the most recent awake timestep before
  $\tau$ (in line~\eqref{l:Iu}). 
  \Cref{cl:unmarkeddefined} shows that there is always an awake
  timestep on the timeline of every vertex.

\item The dual objective corresponding to constraints in
  $\cons^v(\tau)$ %
  is $\sum_{C \in \cons^v(\tau)} b^C\, z_C$, where
  $C$ is of the form %
  $\langle a^C, \yv\rangle \geq b^C$. The local update process ends
  when the increase in this dual objective due to raising variables
  $\{ z_C \mid C \in \cons^v(\tau)\}$ equals $\gamma$.
  
\end{enumerate}

\begin{figure}
    \centering
    \includegraphics[width=4.5in]{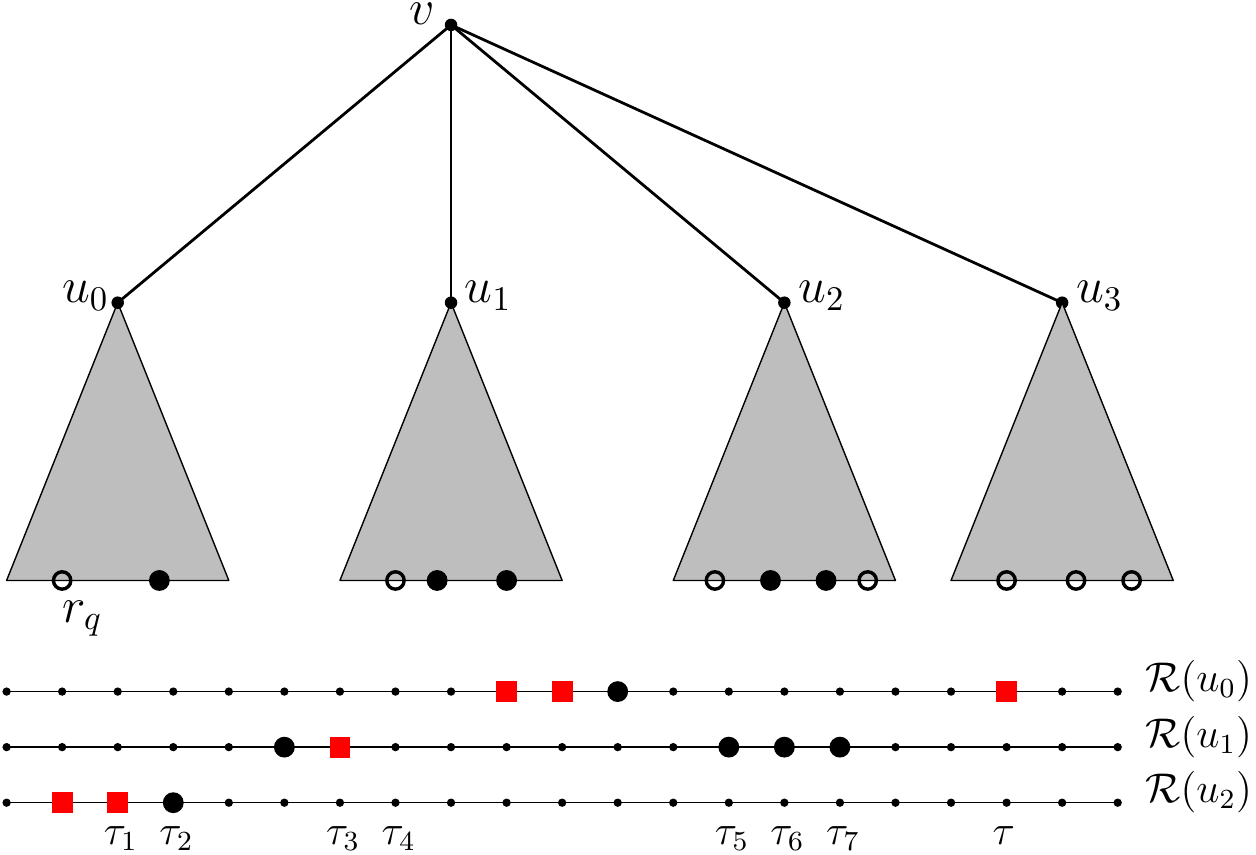}
    \caption{Illustration of $\DU(v, \tau)$: leaves with filled dots
      are active and open dots are inactive, so
      $\sib(u_0, \tau) = \{u_1, u_2\}$. The bold red squares or black
      circles denote timesteps in
      $\cR(u) = \cR^s(u) \cup \cR^{ns}(u)$, with the red squares being
      awake at timestep $\tau$. Hence,
      $I_{u_0} = S_{u_0} = \emptyset$,
      $I_{u_1} = (\tau_3, \tau], S_{u_1} = \{\tau_4, \tau_5, \tau_6,
      \tau_7\}, I_{u_2} = (\tau_1, \tau], S_{u_2} = \{\tau_2\}.$
      Timesteps in $\cR^s(u)$ always remain awake. }
    \label{fig:du}
\end{figure}
For a constraint $C \in \cons^u(t)$, the variable $z_C$ is only
raised in the call $\DU(u,t)$. Subsequently, only the right side
of~(\ref{eq:bound}) can be raised. Hence, once a constraint $C$
becomes depleted, it stays depleted. %
It is worth discussing the special case when $\sib(u_0, \tau)$ is
empty, so that $U = \{u_0\}$. In this case, no server transfer can
happen, and the constraint $C(v,\sigma, \tau)$ is same as a slack
constraint of $\fL^{u_0}(\tau)$, but with an additive term of
$(n_v - n_{u_0}) \delta$ on the RHS, as in~\eqref{eq:cons}. We still
raise the dual variable $z_{C(v, \sigma, \tau)}$, and  prove
that the dual objective value rises by $\gamma$.

There is a parameter $M$ in line~\eqref{l:draise2} that specifies the
rate of change of $\yv$. This value $M$ should be an upper bound on
the size of the index set $S_u$ over all calls to \DU, and over all
$u \in U$. 
\Cref{cor:2} gives a bound of
$M \leq \Mbound$, independent of the trivial bound $M \leq T$, where
$T$ is the length of the 
input sequence.


\section{Analysis Details}
\label{sec:analysis}

The proof rests on two lemmas: the first (proved in
\S\ref{sec:proving-invariants}) bounds the movement cost in
terms of the increase in dual value, and the second (proved in
\S\ref{sec:appr-dual-feas}) shows
near-feasibility of the dual solutions.

\begin{lemma}[Server Movement]
  \label{lem:DUcost}
  The total movement cost during an execution of the
  procedure $\DU$ is at most $2\gamma$, and the objective value
  of the dual $\fD^v$ increases by exactly $\gamma$.
\end{lemma}
\begin{lemma}[Dual Feasibility]
  \label{lem:dual-feasible}
  For each vertex $v$, the dual solution to $\fL^v$ is
  feasible if scaled down by a factor of $\beta$, where $\beta = O(\log \frac{nMk}{\gamma}) = O(H
  \log (n\lambda))$. 
\end{lemma}

\begin{theorem}[Competitiveness for $k$-server]
  \label{thm:main-simple}
  Given any instance of the $k$-server problem on a $\la$-HST with
  height $H \leq \la/10$, \Cref{algo:main} ensures that each request
  location $\rqq$ is saturated at some timestep in $[q,q+1)$. The total
  cost of (fractional) server movement is
  $O(\beta H) = O(H^2 \log(n\la))$ times the cost of the optimal
  solution.
\end{theorem}

\begin{proof}
  All the server movement happens within calls to \DU. By
  \Cref{lem:DUcost}, each iteration of the \textbf{while} loop of
  line~\eqref{l:while} in \Cref{algo:main} incurs a total movement cost of $O(H \gamma)$
  over at most $H$ vertices on the backbone. Moreover, the call
  $\DU(\rootvtx, \tau)$ corresponding to the root vertex $\rootvtx$ increases the value of the dual solution to
  the LP $\fL^\rootvtx$ by $\gamma$. This means the total movement
  cost is at most $O(H)$ times the dual solution value.  Since all
  constraints of $\fL^\rootvtx$ are implied by the relaxation $\fM$,
  any feasible dual solution gives a lower-bound on the optimal
  solution to $\fM$. %
 By \Cref{lem:dual-feasible}, 
 the dual solution is feasible when scaled down by $\beta$, 
  and so the
  (fractional) algorithm is
  $O(\beta H)=O(H^2 \log (n \lambda))$-competitive. 
\end{proof}
As mentioned in the introduction, using $\lambda$-HSTs with $\lambda =
O(\log \Delta)$ allows us to extend this result to general metrics
with a further loss of $O(\log^2 \Delta)$.

\subsection{Bounds on Server Transfer and Dual Increase}
\label{sec:proving-invariants}

The dual
increase of $\gamma$ claimed by \Cref{lem:DUcost} 
will follow from the proof of
\Cref{invar:dual-raised}. The upper bound on the server movement will
follow from a new invariant, which we state below. Then in \S\ref{sec:proving-invariants} we
show both invariants are indeed maintained throughout the algorithm.

We first define the
notion of the ``lost'' dual increase.  Consider a call
$\DU(v,  \tau)$. Let $u$ be $v$'s child such that request location
$v_0$ lies in $T_u$. We say that $u$ is $v$'s \emph{principal} child at
timestep $\tau$. %
We prove (in~\Cref{cl:u0}) that
$\tau \in \cR(u)$ remains in the awake set and hence
$\tau_u = \tau$ throughout this procedure call.  The dual
update raises $z_{C(v,\sigma,\tau)}$ in line~\eqref{l:draise1} and
transfers servers from subtrees $T_{u'}$ for $u' \in \sib(u, \tau)$
into subtree $T_u$ in line~\eqref{l:dtr}. This transfer has two
components, which we consider separately. The first is the
{\em local component} $\frac{d\yv(u,t)}{ds}$, and the second is the {\em
  inherited component} $\frac{b^{C_u}}{\la^h}$. In a sense, the
inherited component matches the dual increase corresponding to the
term $\sum_{u' \in \sib(u, \tau)} b^{C_{u'}}$ on the RHS
of~\eqref{eq:cons}. The only term without a
corresponding server transfer is $b^{C_u}$ itself, where $C_u \in \cons^u(\tau)$
is the constraint in $\sigma$ corresponding to the principal child $u$. Motivated by this, we give the following definition.

\begin{defn}[Loss]
  For vertex $u$ with parent $v$, consider a timestep
  $\tau \in \cR^{ns}(v)$ such that $\tau \in \cR(u)$ as well.  If $\tau \in \cR^s(u)$,
  define $\loss(u, \tau) := 0$. Else  $\tau \in \cR^{ns}(u)$, in which case.
  \begin{gather}
    \label{eq:lossdef}
    \loss(u, \tau) := \sum_{C \in \cons^u(\tau)} \; \sum_{C(v, \sigma, \tau): C \in \sigma}
    b^C \; z_{C(v,\sigma,\tau)} ~~.
  \end{gather}
\end{defn}

\begin{leftbar}
  \begin{invariant}
    \label{invar:movement}
    For node $v$ and timestep $\tau \in \cR^{ns}(v)$, let $u$ be $v$'s
    principal child at timestep $\tau$.  The server mass entering subtree
    $T_{u}$ during the procedure $\DU(v, \tau)$ is at most
    \begin{gather}
      \frac{\gamma - \loss(u, \tau)}{\la^{\level(u)}}. \tag{I4}
    \end{gather}
    Moreover, timestep $\tau \in \cR(u)$ stays
    awake during the call 
    $\DU(v, \tau)$. 
  \end{invariant}
\end{leftbar}

Multiplying the amount of transfer by the cost of this transfer, we
get that the total movement cost is at most $O(\gamma)$.
\Cref{invar:dual-raised,invar:movement} prove \Cref{lem:DUcost}. We
now show these invariants hold over the course of the algorithm.

\subsubsection{Proving \Cref{invar:dual-raised,invar:movement}}

To prove these invariants, we define a total order on pairs
$(v, \tau)$ with $\tau \in \cR(v)$ as follows:
\begin{gather*}
  \text{define: } (v_1, \tau_1) \prec (v_2, \tau_2) \text{ if } \tau_1 < \tau_2, \text{
    or if $\tau_1 = \tau_2$ and $v_1$ is a descendant of $v_2$}. 
\end{gather*}
Since calls to $\DU$ are made in this order, we also prove the
invariants by induction on this ordering: Assuming both invariants hold
for all pairs $(v, \tau) \prec (\vstar, \tst)$, we prove them for the
pair $(\vstar, \tst)$.  The base case is easy to settle: at $q=0$, we
only have $\bot$-constraints at the dummy leaf nodes. The only
non-trivial statement among~\Cref{invar:dual-raised,invar:movement}
for these nodes is to check that $b^C > 0$ for any such
$\bot$-constraint $C$ at a dummy leaf $v$. Note that $b^C = 1 -
k_{v,0} - 2 \delta(n-1) = \frac{1}{2} - 2 \delta(n-1) >
0$. \alert{DOuble-check this.}

We start off with some supporting claims before proving the inductive step 
\Cref{invar:dual-raised,invar:movement}. First, we show that the
notion of $\prev$ timestep in the \DU procedure is well-defined.  
\begin{claim}
  \label{cl:unmarkeddefined}
  Let $u$ be any non-root vertex. Then the first timestep in $\cR(u)$
  corresponds to a $\bot$-constraint. Therefore, for any timestep
  $\tau$ such that $T_u$ has an active leaf at timestep $\tau$,
  $\prev(u, \tau)$ is well-defined.
\end{claim}
\begin{proof}
  If $u$ is any of the dummy leaf nodes, then this follows by
    construction, the first timestep has a $\bot$-constraint.  Else,
  let $q$ be the first time when a request arrives below $u$. Let
  $\tau_f$ be the first timestep after $q$. In the first iteration of
  the {\bf while} loop in~\Cref{algo:main} (corresponding to timestep
  $\tau_f$), we would call $\DUZ(u, \tau_f)$ because there are no
  active leaves below $u$ at this timestep. Hence we would add a
  $\bot$-constraint at timestep $\tau_f$, proving the first part of
  the claim.  To show the second part, let $\tau$ be a timestep such
  that $T_u$ has an active leaf below it at timestep $\tau$. This
  means that $\tau \geq \tau_f$. %
  Since $\fL^u(\tau_f)$ is a
  $\bot$-constraint, $\tau_f$ is awake, and so 
  $\prev(u, \tau)$ is  well-defined.
\end{proof}

Next, we define $\fll(u,\tau)$ to to be the set of timesteps that
\emph{load the constraints in $\fL^u(\tau)$}. Formally, we have
\begin{defn}[\fll]
  Given a node $u$ and its parent $v$, timestep $\tau \in \cR^{ns}(u)$,
  and constraint $C \in \cons^u(\tau)$, define $\fll(C)$ to be the
  timesteps $\tau'$ such that some constraint $C' \in \cons^v(\tau')$
  appears on the RHS of inequality~(\ref{eq:bound}) corresponding
  to $C$. All these timesteps $\tau'$ must be after $\tau$.
  Extending this, let
  \begin{align}
    \label{eq:filldef}
    \fll(u,\tau) := \bigcup_{C \in \cons^u(\tau)} \fll(C).
  \end{align}
\end{defn}
In other words, $\fll(u,\tau)$ is the set of timesteps
$\tau' \in \cR^{ns}(v)$ such that when we called $\DU(v,\tau'$), the
node $u$ was either the $v$'s principal child at timestep $\tau'$ or
else belonged to the active sibling set, and moreover
$\prev(u, \tau') = \tau$. The following lemma shows part of their
structure. Recall that $(\vstar, \tst)$ denotes the current pair in
the inductive step.

\begin{claim}[Structure of \fll\  times]
  \label{cl:fill}
  Fix a node $u$ with parent $v$, and a timestep $\tau \in \cR^{ns}(u)$
  such that $(v, \tau) \prec (\vstar, \tst)$.  Then for any
  $\tau' \in \fll(u, \tau)$, either (a)~$\tau'=\tau$, and
  $u$ is the principal child of $v$ at timestep $\tau'$, or else
  (b)~$\tau' > \tau$, and $u$ is \emph{not} $v$'s principal child at
  timestep $\tau'$.
\end{claim}
\begin{proof}
  Suppose $\tau = \tau'$. Since we call \DU\ only for ancestors of the
  requested node $v_0$, and $\tau \in \cR^{ns}(u)$, so $v_0$ belongs to
  $T_u$ (and hence $u$ is the principal child of $v$ at timestep $\tau$). Else suppose $\tau' > \tau$, and suppose $u$ is indeed $v$'s
  principal child at this timestep. Then during the call
  $\DU(v, \tau')$, we have $\prev(u, \tau') = \tau'$
  throughout the execution of $\DU(v,\tau')$ (by the second statement in~\Cref{invar:movement}), and
  hence $\tau' \notin \fll(u, \tau)$, giving a contradiction.
\end{proof}

We now give an upper bound on the server mass
entering a subtree at any timestep $\tau < \tst$.

\begin{claim}
  \label{cl:serverupperbound}
  Let $\tau \in \cR(u)$, $\tau < \tst$. The server mass
  entering $T_u$ at timestep $\tau$ is at most
   \[ \left( 1 + \frac{1}{\lambda-1} \right) \frac{\gamma}{\lambda^{\level(u)}} -
   \frac{\loss(u, \tau)}{\lambda^{\level(u)}}.  \]
\end{claim}

\begin{proof}
  Since $\tau < \tst$, we can apply the induction hypothesis to all
  pairs $(v, \tau)$ where $v$ is an ancestor of $u$.  Servers enter
  $u$ at timestep $\tau$ because of $\DU(w,\tau)$ for some ancestor $w$
  of $u$. When $w$ is the parent of $u$, \Cref{invar:movement} shows
  this quantity is at most $ \frac{\gamma - \loss(u, \tau)}{\la^h}, $
  where $h = \level(u)$. For any other ancestor $w$ of $v$,
  \Cref{invar:movement} implies a weaker upper bound of
  $\frac{\gamma}{\la^{h+k}}$, where $\level(w) = h+k+1$.  Simplifying
  the resulting geometric sum
  $ \frac{\gamma - \loss(u, \tau)}{\la^h} + \sum_{h' \geq h+1}
  \frac{\gamma}{\la^{h'}}$ completes the proof.
\end{proof}

Next, we give a lower bound on the amount of server \emph{moving out}
of some subtree $T_w$. Such transfers out of $w$ takes place in
line~\eqref{l:dtr} with $w$ being either the node $u$ referred to on
this line, or
a descendant of such a node. Moreover,  the server movement out of $T_w$ at timestep
$\tau$ is denoted $g(w, \tau)$, which is non-zero only for those timesteps $\tau$ when
$w$ is not on the corresponding backbone. We split this transfer
amount into two:
\begin{OneLiners}
\item[(i)] $g^\loc(w,\tau)$: the \emph{local component} of the transfer, i.e.,
  due to the increase in $\yv$ variables.
\item[(ii)] $g^\inh(w,\tau)$: the \emph{inherited component} of
  the transfer, i.e., due to the $b^{C_u}$ term.
\end{OneLiners}

\begin{lemma}
  \label{lem:rel1}
 Let $u$ be a non-principal child of $\vstar$ at timestep $\tst$, and $I := (\tau_1, \tst]$ for some timestep $\tau_1 < \tst$. 
  Let
  $S$ be the timesteps in $\cR^{ns}(u) \cap (\tau_1, \tst]$
  that have been removed from $\awake(u)$  
  by the moment when $\DU(\vstar,\tst)$ is called. %
  Then
  \begin{gather*}
    g^\inh(u, (\tau_1,\tst]) \geq \left( 1 + \frac{1}{H} \right) |S|\,
    \frac{\gamma}{\lambda^{\level(u)}} - \sum_{\tau \in S} \frac{\loss(u,
      \tau)}{\lambda^{\level(u)}}.
  \end{gather*}
\end{lemma}

\begin{proof}
  Consider timesteps $\tau\in S$ and $\tau' \geq \tau$ such that
  $\tau' \in \fll(u, \tau)$. (We use the term phase here to denote a range of values of
  the timer $s$.) Consider the phase during $\DU(\vstar, \tau')$
  when $\tau_u' := \prev(u, \tau')$ equals $\tau$: since
  $\tau' \in \fll(u, \tau)$, we know that there will be such a
  phase.  %
  Whenever
  we raise the timer $s$ by a small $\varepsilon$ amount during this phase,
  we raise some dual variable $z_{C(\vstar,\sigma,\tau')}$ by the
  same amount, where $\sigma$ contains a constraint $C$ from
  $\cons^u(\tau)$. Thus we contribute $\varepsilon$ to the LHS
  of~\eqref{eq:bound} for constraint $C$.  For such a constraint
  $C$, let $[s_1(\tau',C), s_2(\tau',C)]$ be the range of the timer
  $s$ during which we raise a dual variable of the form
  $z_{C(\vstar,\sigma,\tau')}$ such that $C \in \sigma$.

  The timestep $\tau$ was removed from $\awake(u)$ by
  line~\eqref{l:mark2} because~\eqref{eq:bound} became tight for all
  constraints $C \in \cons^u(\tau)$, so:
  \begin{align}
   \left( 1 + \frac{1}{H} \right) \sum_{C \in \cons^u(\tau)} b^C z_C &
     = \sum_{C \in \cons^u(\tau)} b^C \sum_{C(\vstar, \sigma, \tau'): C \in \sigma} z_{C(\vstar,\sigma,\tau')} \notag \\
    \intertext{Now the definition of $\loss(u, \tau)$  allow us to split the expression on
    the RHS as follows:}
   & = \loss(u, \tau) + \sum_{C \in \cons^u(\tau)} b^C \sum_{C(\vstar,
     \sigma, \tau'): C \in \sigma, \tau' > \tau} z_{C(\vstar,\sigma,\tau')} \notag \\
   & = \loss(u, \tau) + \sum_{\tau' \in \fll(u, \tau), \tau' > \tau} \sum_{C \in \cons^u(\tau)} b^C \Big(s_2(\tau',C)-s_1(\tau',C)\Big). \label{eq:calc1} 
  \end{align}

  We now bound the second expression on the RHS in another way.
  For a timestep $\tau' \in \fll(u, \tau)$ with
  $\tau' > \tau$, consider the phase when timer $s$ lies
  in the range $[s_1(\tau',C),s_2(\tau',C)]$ for a constraint
  $C \in \cons^u(\tau)$. Since $\tau' > \tau$, \Cref{cl:fill} implies
  that $u$ is not the principal child of $\vstar$ at timestep $\tau'$, so
  raising $s$ by $\varepsilon$ units during this phase means that
  line~\eqref{l:dtr} moves $\varepsilon\cdot \frac{b^C }{\la^h}$ servers
  \emph{out} of $T_u$, where $h := \level(u)$. Hence the increase in $g^\inh(u, I)$ due to
  transfers corresponding to timestep $\tau \in S$ is at least
  \begin{align*}
  \sum_{\tau' \in \fll(u, \tau), \tau' > \tau} \sum_{C \in \cons^u(\tau)} \frac{b^C (s_2(\tau',C)-s_1(\tau',C))}{\la^h} 
  & ~~\stackrel{\text{by~\eqref{eq:calc1}}}{=}~~
  \left( 1 + \frac{1}{H} \right) \sum_{C \in \cons^u(\tau)} \frac{b^C z_C}{\la^h} - \frac{\loss(u, \tau)}{\la^h} \\
  &= \left( 1 + \frac{1}{H} \right) \frac{\gamma}{\la^h} - \frac{\loss(u, \tau)}{\la^h}. 
  \end{align*}
  The final equality above uses $(u,\tau) \prec (\vstar, \tst)$,
  because $\tau$ had been removed from $\awake(u)$ %
  before the call to $\DU(\vstar, \tst)$, which means we can use the
  induction hypothesis \Cref{invar:dual-raised} for timestep $\tau
  \in \cR^{ns}(u)$. Finally, summing
  over all timesteps in $S$ completes the proof.
\end{proof}

\begin{corollary}
  \label{cor:1}
  Let $u$ be a non-principal child of $\vstar$ at timestep $\tst$, and
  $I := (\tau_1, \tst]$. Consider the moment when $\DU(\vstar, \tst)$
  is called. If none of the timesteps in $I \cap \cR(u)$ belong to
  $\awake(u)$, then
  \begin{OneLiners}
  \item[(i)] $g^\inh(u,I) \geq r(u,I)$,
  \item[(ii)] $0 \leq g^\loc(u,I) \leq \Diff(u,I)$, and
  \item[(iii)] $g^\loc(u, I) \geq \y{\vstar}(u,I)$.
  \end{OneLiners}
  Finally, $b^C > 0$ for any constraint $C \in \fL^\vstar(\tst)$ of the form $a^C \cdot \yv  \geq b^C$. 
\end{corollary}

\begin{proof}
  Since timesteps in $\cR^s(u)$ always stay awake,
  $I \cap \cR(u) = I \cap \cR^{ns}(u)$; call this set $S$. Since $u$
  is a non-principal child at timestep $\tst$, we have
  $\tst \not\in \cR^{ns}(u)$.  This means $\tau < \tst$ for any
  $\tau \in S$, and so \Cref{cl:serverupperbound} gives an \emph{upper
    bound} on the server movement \emph{into} $u$ at timestep $\tau$,
  and \Cref{lem:rel1} gives a \emph{lower bound} on the server
  movement \emph{out of} $u$. Combining the two,
  \begin{align}
    g^\inh(u,I) - r(u,I) \geq 
    \left( \frac1H - \frac{1}{\lambda-1} \right) \frac{\gamma
    |S|}{\lambda^h} \geq \frac{4}{5H}\cdot \frac{\gamma
    |S|}{\lambda^h} \geq 0,  \label{eq:g-minus-r}
  \end{align}
  since $\la \geq 10H$ and $H \geq 2$, which proves~(i). To prove~(ii),
  \[ g^\loc(u,I) = (g-g^\inh)(u,I) \stackrel{\text{by~(i)}}{\leq} (g-r)(u,I)
    \stackrel{\text{by defn.}}{=} \Diff(u,I). \]
To prove~(iii), whenever we raised $\y{\vstar}(u, \tau')$ for some
  timestep $\tau'$, we 
  raised $g^\loc(u, \tau'')$ for some $\tau'' \geq \tau'$) with the same
  rate. 
  Both timesteps $\tau', \tau''$ appear before $\tst$, because
  we consider the moment when we call $\DU(\vstar, \tst)$. Since
  interval $I$ ends at $\tst$, it must contain either only $\tau''$ or
  both $\tau', \tau''$, giving us that
  $g^\loc(u, I) \geq \y{\vstar}(u,I)$.

  \agnote{Stopping here.}
  We now prove the final statement. If $C \in \fL^\vstar(\tst)$ is a
  $\bot$-constraint $C$ added by $\DUZ(\vstar, \tst)$. 
  $b^C = 1 - k_{\vstar, \tst} - 2 \delta(n-n_\vstar)$ (using~\eqref{eq:implied1}). Since $k_{\vstar,\tst} \leq 1 - \delta'$ (otherwise the {\bf while} loop in \Cref{algo:main} would have terminated), we see that $b^C \geq \delta' - 2 \delta n > 0$. 
  The other case is when $C$ is of the form $C(\vstar,\sigma,\tst) $
  as in~\eqref{eq:cons}. By the induction hypothesis (\Cref{invar:dual-raised}), $b^{C_u} > 0$ and $\Diff(u, I_u) \geq 0$ by~(ii) above. Since $n_\vstar > \sum_{u \in U} n_u$, it follows that $b^{C} > 0$.  
\end{proof}

Having proved all the supporting claims,
we start off with proving that the second statement in~\Cref{invar:dual-raised} holds at $(\vstar, \tst)$. 
\begin{claim}[Principal Node Awake]
  \label{cl:u0}
  Suppose we call $\DU(\vstar,  \tst)$. If $u$ is the principal child of $\vstar$  at timestep
  $\tst$, this call does not remove the timestep $\tst$ from 
  $\awake(u)$.
\end{claim}

\begin{proof}
  At the beginning of the call to $\DU(\vstar, \tst)$, the timestep
  $\tst$ has just been added to $\cR(u)$ (and to
  $\awake(u)$) in the call to $\DU(u, \tst)$ or to $\DUZ(u, \tst)$, and
  cannot yet be removed from $\awake(u)$. So we start with
  $\tau_{u} = \tst$. For a contradiction, if we remove $\tst$ from
  $\awake(u)$ in line~\eqref{l:mark2}, then all the constraints in
  $\cons^{u}(\tst)$ must have become depleted. For each such
  constraint $C \in \cons^{u}(\tst)$, the contributions to the RHS
  in~\eqref{eq:bound} during this procedure come only from the
  newly-added constraints
  $C(\vstar,\sigma,\tst) \in \cons^\vstar(\tst)$. So if all
  constraints in $\cons^{u}(\tst)$ become depleted, the total dual
  objective raised during this procedure is at least
  $$ \sum_{C \in \cons^{u}(\tst)}\; \sum_{C(\vstar,\sigma,\tst) \in
    \cons^\vstar(\tst): C \in \sigma} b^{C(\vstar,\sigma,\tst)} \;
  z_{C(\vstar,\sigma,\tst)} \geq \left(1 + \nicefrac1H \right) \sum_{C \in
    \cons^{u}(\tst)} b^C\,z_C, $$ where we use that
  $b^{C(\vstar,\sigma,\tst)} \geq b^C$ (because in~\eqref{eq:cons}, $b^{C_u} \geq 0$ by the induction hypothesis (\Cref{invar:dual-raised}) and  $\Diff(u,I_u) \geq 0$  by~\Cref{cor:1}),  and that each
  constraint in $\cons^{u}(\tst)$ satisfies~(\ref{eq:bound}) at
  equality. The induction hypothesis \Cref{invar:dual-raised}
  applied to $(u, \tst)$ implies that
  $\sum_{C \in \cons^{u}(\tst)} b^C\,z_C = \gamma$, so the RHS above
  is $(1+\nicefrac1H) \gamma$. So the total dual increase
  during $\DU(\vstar,\tst)$, which is at least the LHS above, is
  strictly more than $\gamma$, contradicting the stopping condition of $\DU(\vstar, \tst)$.
\end{proof}

Next, we prove the remainder of the inductive step, namely that
\Cref{invar:dual-raised,invar:movement} are satisfied with respect
to $(\vstar, \tst)$ as well. 

\begin{claim}[Inductive Step: Active Siblings Exist]
  \label{cl:dualtr}
  Consider the call $\DU(\vstar, \tst)$, and let $u_0$ be 
  the principal child of $\vstar$ at this timestep. Suppose
  $\sib(u_0, \tst) \neq \emptyset$. Then the dual objective value
  corresponding to the constraints in $\cons^\vstar(\tst)$ equals
  $\gamma$; i.e.,
  \[ \sum_{C \in \cons^\vstar(\tst)} z_C \, b^C = \gamma. \] Moreover, 
  the server mass entering $T_{u_0}$ going to the requested node in this
  call is at most
  \[  \frac{\gamma - \loss(u_0, \tst)}{\lambda^{\level(u)}}. \]
\end{claim}

\begin{proof}
  Let $U' := \sib(u_0, \tst)$ be the non-principal children of
  $\vstar$ at timestep $\tst$; let $U := \{u_0\} \cup U'$ as in \DU.
  The identity of the timesteps $\tau_u$ and intervals $I_u$ change
  over the course of the call, so we need notation to track them
  carefully.  Let
  $I_u(s')$ be the set $I_u$ when the timer value is $s'$; similarly, 
  let $\Diff_{s}(u, I_u(s'))$ be the value of
  $\Diff(u,I_u(s'))$ when the timer value is $s$, and  
  $\y{\vstar}_{s}(u, I_u(s'))$ is defined similarly.

  For $u \in U'$, \Cref{cor:1}(ii,iii) implies that for
  any interval $I_u(s)$, 
  \begin{gather}
    \Diff_{0}(u, I_u(s)) \geq \y{\vstar}_{0}(u,I_u(s)). \label{eq:delta-y}
  \end{gather}
  Since the timestep $\tst$ stays awake for the principal child $u_0$ (due
  to \Cref{cl:u0}), the interval $I_{u_0}(s)$ equals $(\tst, \tst]$,
  which is empty, for all values of the timer $s$.
  
  The dual increase is \emph{at most} $\ga$ due to the stopping
  criterion for \DU, so we need to show this quantity reaches $\ga$.
  Indeed, suppose we raise the timer from $s$ to $s+ds$ when
  considering some constraint $C_s(v, \sigma, \tau)$---the subscript
  indicates the constraint considered at that value of timer $s$. The dual objective increases by 
   $b^{C_s(v, \sigma, \tau)} \, ds$. We now use the definition of $b^{C_s(v, \sigma, \tau)}$
  from~(\ref{eq:cons}), substitute
  $(n_v - \sum_{u \in U} n_u) \geq 1$, and use that all $b^{C_u}$ terms in the
  summation are non-negative (by~\Cref{invar:dual-raised}) to drop these terms. This gives the first inequality below (recall that $I_{u_0}(s)$ stays empty):
  \begin{align}
    \label{eq:transy}
    b^{C_s(v, \sigma, \tau)} &\geq \sum_{u \in U'} \Diff_{s}(u, I_u(s)) +
                             \delta \geq \sum_{u \in U'}
                             \Diff_{0}(u, I_u(s)) + \delta 
                            \geq \sum_{u \in U'} \y{\vstar}_{0}(u,I_u(s))  + \delta.
  \end{align}
  The second inequality above uses that $\Diff_s \geq \Diff_0$ for
  non-principal children, and the third uses~(\ref{eq:delta-y}).
  Let
  \[ Y(s):= \sum_{\tau'} \sum_{u \in U'} \Big(\y{\vstar}_s(u, \tau') -
    \y{\vstar}_0(u, \tau')~\Big) \] to be the total increase in
the   $\y{\vstar}$ variables during $\DU(\vstar, \tst)$ until the timer
  reaches $s$. This is also the total amount of server transferred to the
  requested node due to the \emph{local component} of transfer in
  line~\eqref{l:dtr} until this moment.

  \begin{subclaim}
  \label{subcl:Y}
    $Y(s) < \gamma$. %
  \end{subclaim}
  \begin{subproof}
    Suppose not, and let $s^\star$ be the smallest value of the timer such that
    $Y(s^\star) = \gamma$. Note that $Y(s)$ is a continuous non-decreasing function of $s$. 
    For any $s \in [0,s^\star)$, we get
    $Y(s) < Y(0) + \gamma$, where $Y(0)=0$ by definition. Since the
    intervals $I_u(s') \sse I_u(s)$ for $s' \leq s$, all the increases
    in the $\y{\vstar}$ variables during $[0,s]$ correspond to timesteps in $I_u(s)$. Thus
    for any $s < s^\star$,
    \begin{gather}
      Y(0) + \gamma > Y(s) \quad \implies \quad \sum_{u \in U'}
      \y{\vstar}_0(u, I_u(s)) + \gamma > \sum_{u \in U'}
      \y{\vstar}_s(u, I_u(s)). \label{eq:YY}
    \end{gather}
    The dual increase during
    $[s, s+ds]$ is
    \begin{align*}
      b^{C_s(v, \sigma, \tau)} \, ds
      &\stackrel{\text{by~(\ref{eq:transy},\ref{eq:YY})}}{>} \Big( \sum_{u \in U'}
                                       \y{\vstar}_s(u, I_u(s)) +
                                       \delta - \gamma \Big)\, ds \\
      &= \Big(\la^h \,dY(s) - \frac{\gamma}{Mn } \sum_{u \in U'} |S_u| %
      \, ds\Big) + (\delta-\gamma)\, ds > \la^h \,
        dY(s) \geq dY(s).
    \end{align*}
    The second line uses (a) the update rule in
    line~(\ref{l:draise2}) with $dY(s)$ denoting $Y(s+ds) -
    Y(s)$, (b) that $M \geq |S_u|$ %
    and $|U'| \leq
    n$, so the
    second expression is bounded by $\gamma$, and (c) that
    $\delta > 2\gamma$. Integrating over $[0,s^\star]$, the total dual
    increase is strictly more than
    $Y(s^\star) = \gamma$, which contradicts the stopping condition of $\DU$.
  \end{subproof}
  Combining \Cref{subcl:Y} (and specifically its implication \eqref{eq:YY}) with \eqref{eq:delta-y} implies that for all
  values $s$ of the timer:
  \begin{align}
      \label{eq:timery1}
      \sum_{u \in U'} \y{\vstar}_s(u, I_u(s)) < \sum_{u \in U'} \Diff_{0}(u, I_u(s)) + \gamma.  
  \end{align}
  Therefore, the increase in dual objective during $[s, s+ds]$ is at least
  \begin{align*}
        b^{C_s(v, \sigma, \tau)}\, ds &\stackrel{\eqref{eq:cons}}{\geq} \bigg(\sum_{u \in U'} \Big( \Diff_{0}(u, I_u(s)) +
                                                                     b^{C_{u,s}} \Big)  + \delta  + b^{C_{u_0}} \bigg)\, ds  \\
    &\stackrel{\eqref{eq:timery1}}{>} \bigg(\sum_{u \in U'} \Big( \y{\vstar}_s(u, I_u(s)) + 
      b^{C_{u,s}} \Big)  + (\delta-\gamma)  + b^{C_{u_0}} \bigg)\, ds \\
      &  \geq \sum_{u \in U'} \Big( \la^h\, d\y{\vstar}_s(u,I_u(s)) -
        \frac{\gamma}{Mn} |S_u|\,  ds %
        + b^{C_{u,s}}\, ds \Big) +
        \gamma\, ds + b^{C_{u_0}}\, ds  \\
      & \geq  \sum_{u \in U'} \Big( \la^h \, d\y{\vstar}_s(u,I_u(s)) + b^{C_{u,s}} \Big) \,ds + b^{C_{u_0}} \,ds \\
      & = \la^h [\text{amount of server transferred in $[s,s+ds]$}] + b^{C_{u_0}} \, ds
  \end{align*}
  Here $C_{u,s}$ is the constraint corresponding to $u \in U'$ when
  the timer equals $s$. The third inequality above follows from the
  update rule in line~\eqref{l:draise2}, and that $\delta \geq
  2\gamma$.  The last equality follows
  from line~\eqref{l:dtr}. Integrating over the entire range of the
  timer $s$, we see that the total dual objective increase is at least
  $\la^h [\text{total server transfer}] + \loss(u_0, \tst)$. Since the
  total dual increase is at most $\gamma$, the total server transfer
  is at most $\frac{\gamma - \loss(u_0, \tst)}{\la^h}$. This proves
  the second part of~\Cref{cl:dualtr}.

  We now prove that the \DU process does not stop until the dual
  increase is $\gamma$. %
  For each $u \in U'$, the subtree $T_u$ contains at least one active leaf
  and hence at least $\delta$ servers when \DU is called. Since the
  total server transfer is at most $\gamma \ll \delta$, we do not run
  out of servers. It follows that until the dual objective reaches
  $\gamma$, we keep raising $\y{\vstar}_s(u, I_u(s))$ for some non-empty
  interval $I_u(s)$ for each $u \in U'$, and this also raises the dual
  objective as above.
\end{proof}

It remains to consider the general case when $\sib(u_0, \tst)$ may be empty. 
\begin{claim}[Inductive Step: General Case]
  \label{cl:0}
  At the end of any call $\DU(\vstar, \tst)$, the total dual objective
  raised during the call equals $\gamma$.
\end{claim}
\begin{proof}

  If $\sib(u_0, \tst)$ is non-empty, this follows
  from~\Cref{cl:dualtr}. So  assume that $\sib(u_0, \tst)$ is empty.
  In this case, there are no $\yv(u,t)$ variables to raise because the interval $I_{u_0}$ is empty. 
  As we raise $s$, we also raise $z_{C(\vstar,\sigma,\tst)}$ in
  line~(\ref{l:draise1}).
   Since we do not make all the
  constraints in $\cons^{u_0}(\tst)$ depleted (\Cref{cl:0}), the total
  dual increase must reach $\gamma$, because $b^{C(\vstar, \sigma, \tst)} > 0$ by~\Cref{cor:1}. 
\end{proof}

This completes the proof of the induction hypothesis for the pair
$(\vstar, \tst)$. Before we show dual feasibility, we give an upper bound on the parameter $M$.

\begin{corollary}[Bound on $M$]
  \label{cor:2}
  For node $u$ and timestep $\tau$, let $\tau_u :=
  \prev(u,\tau)$. %
  There are at most $\Mbound$ timesteps in
  $(\tau_u, \tau] \cap \cR^{ns}(u)$. So  we can set $M$ to  $\frac{5H \la^H k}{4 \gamma}+1$.
\end{corollary}
\begin{proof}
  Let $I := (\tau_u, \tau]$. By the choice of $\tau_u$, none of the
  timesteps in $S:= I \cap \cR^{ns}(u)$ belong to $\awake(u)$.  The proof
  of~\Cref{cor:1}, and specifically (\ref{eq:g-minus-r}), shows that
  $g^\inh(u,I)-r(u,I) \geq\frac{4|S| \gamma}{5 H \la^h}$.  This
  difference cannot be more than the total number of servers, so
  $|S| \leq \frac{5H \la^H k}{4 \gamma}$. Since the set $|S_u|$
  defined in line~\ref{l:su} in \DU is at most $|S| +1$ (because of
  the first timestep of $I_u$), the desired result follows.
\end{proof}

\subsection{Approximate Dual Feasibility}
\label{sec:appr-dual-feas}

For $\beta \geq 1$, a dual solution $z$ is \emph{$\beta$-feasible} if
$z/\beta$ satisfies satisfies the dual constraints.  We now show that
the dual variables raised during the calls to $\DU(v, \tau)$ for
various timesteps $\tau$ remain $\beta$-feasible for 
$\beta = O(\ln \frac{nMk}{\gamma})$. First we show~\Cref{invar:active}, and also give bounds on
variables $\yv(u,t)$.

\begin{claim}[Proof of Invariant~\ref{invar:active}]
  \label{cl:rec}
  For any timestep $\tau$ and
  leaf $v$, the server amount $k_{v, \tau}$ remains in the
  range $[\nf\delta2, 1-\nf{\delta'}2]$. 
\end{claim}
\begin{proof}
Recall that $\gamma \leq 4\delta\ \leq 4\delta'$.
  \Cref{lem:DUcost} proves that the total server mass entering the
  request location in any timestep is at most $2\gamma$. Since the
  request location must have less than $1-\delta'$ at the start of the
  timestep, $k_{v, \tau}$ remains at most
  $1 - \delta' + 2 \gamma \leq 1-\nf{\delta'}2$. Similarly, we move server
  mass from a leaf only when it is active, i.e., has at least $\delta$
  server mass. Hence, $k_{v, \tau}$ remains at least
  $\delta - 2\gamma \geq \nf{\delta}2$.
\end{proof}

\begin{claim}[Bound on $\yv$ Values]
  \label{cl:uppery}
  For any vertex $v$, any child $u$ of $v$, and timestep $\tau$, the
  variable $\yv(u,\tau) \leq 4\gamma M + k$.
\end{claim}
\begin{proof}
  For a contradiction, consider a call $\DU(v, \tau')$ during which we
  are about to raise $\yv(u,\tau)$ beyond $4\gamma M + k$.  Any
  previous increases to $\yv(u,\tau)$ happen during calls
  $\DU(v, \tau'')$ for some $\tau'' \in [\tau, \tau']$.  Moreover,
  whenever we raise $\yv(u,\tau)$ by some amount, we move out
  at least the same
  amount of server mass from the subtree $T_u$. Hence,
  at least $4\gamma M + k$
  amount of server mass has been moved out of $T_u$ in the interval
  $[\tau,\tau']$. Since we have a non-negative amount of server in
  $T_u$ at all times, we must have moved in at least $4\gamma M$
  amounts of server into $T_u$ during the same interval. All this
  movement happens at timesteps in $\cR(u)$. Moreover, for each
  individual timestep $\tau'' \in \cR(u)$, we bring at most $2 \gamma$
  servers into $T_u$, so there must be at least $2M$ timesteps in
  $\cR(u) \cap [\tau,\tau']$. Finally, since we are raising
  $\yv(u,\tau)$ at timestep $\tau'$, the interval $I_u$
  (defined in line~\eqref{l:Iu}) at timestep $\tau'$ must contain
   $[\tau, \tau']$, which means $|I_u \cap \cR^{ns}(u)| > M$ (because no timestep in $\cR^s(u)$ can lie in $[\tau,\tau']$). This
  contradicts the definition of $M$.
\end{proof}

\begin{claim}
\label{cl:two}
Let $t$ be any timestep \alert{typo: in $\cR(u)$}, and $v$ be the parent of
$u$. Define $t_1$ to be the last timestep in $\cR(u) \cap [0,t]$,
and $t_2$ to be the next timestep, i.e., $t_1 + \eta$. Let $C$ be a
constraint in $\fL^v$ containing the variable $\yv(u,t)$ on the
LHS. Then $C$ contains at least one of $\yv(u,t_1)$ and
$\yv(u,t_2)$. Moreover, whenever we raise $z(C)$ in
line~\eqref{l:draise1} of the \DU procedure, we also raise either
$\yv(u,t_1)$ or $\yv(u, t_2)$ according to line~\eqref{l:draise2}.
\end{claim}
\begin{proof}
  Suppose $\yv(u,t)$ appears in a constraint  $\fL^v(\tau)$. Define
  $I_u = (\tau_u, \tau]$ as in line~\eqref{l:Iu}. It follows that $t
  \in I_u$, and so $\tau_u < t$. Therefore, $\tau_u \in \cR(u) \cap [0,t]$, so either $t_1 > \tau_u$ and hence belongs to $I_u$, or else $t_1 = \tau_u$ in which case $t_2 \in I_u$.
It follows that the index set $S_u$ contains either $t_1$ or $t_2$. This implies the second statement in the claim. 
\end{proof}
We now show the approximate dual feasibility. Recall that the
constraints added to $\cons^v(\tau)$ are of the form
$C(v,\sigma,\tau)$ given in~(\ref{eq:cons}),
and we raise the corresponding dual variable $z_{C(v,\sigma,\tau)}$
only during the procedure $\DU(v, \tau$) and never again.

\begin{lemma}[Approximate Dual Feasibility]
  \label{lem:dual}
  For a node $v$ at height $h+1$, the
  dual variables $z_C$ are $\beta_{h}$-feasible for the dual program
  $\fD^v$, where
  $\beta_h = \left( 1+ \nicefrac{1}{H} \right)^h O(\ln n + \ln M
  + \ln (k/\gamma)). $
\end{lemma}

\begin{proof}
  We prove the claim by induction on the height of $v$. For a leaf
  node, this follows vacuously, since the primal/dual programs are
  empty. Suppose the claim is true for all nodes of height at most $h$.
  For a node $v$ at height $h+1 > 0$ with children $\chi_{v}$, the
  variables in $\cons^v$ are of two types: (i)~$\yv(u,t)$ for some timestep $t$
  and child $u \in \chi_{v}$, and (ii)~$\yv(u',t)$ for some timestep $t$ and
  non-child descendant $u' \in T_v \setminus \chi_{v}$. We consider
  these cases separately:
  \begin{enumerate}
  \item[I.] Suppose the dual constraint corresponds to variable $\yv(u,t)$
    for some child $u \in \chi_{v}$. Let $\cons'$ be the set of constraints in
    $\cons^v$ containing $\yv(u,t)$ on the LHS. The dual constraint is:
    \begin{gather}
      \sum_{C \in \cons'} z_C \leq c_u = \lambda^{h}. \label{eq:dual-cons}
    \end{gather}
    Let $t_1, t_2$ be as in the statement of~\Cref{cl:two}. 
    When we raise $z_C$ for a constraint $C \in \cons'$ in
    line~(\ref{l:draise1}) at unit rate, we raise either $\yv(u,\tau_1)$ or $\yv(u, t_2)$ at the
    rate given by line~(\ref{l:draise2}). Therefore, if we raise the
    LHS of the dual constraint~(\ref{eq:dual-cons}) for a total of
    $\Gamma$ units of the timer, we would have raised one of the two variables, say $\yv(u, \tau_1)$, for at least $\Gamma/2$ units of the timer. Therefore, 
    the value of  $\yv(u,\tau_1)$ variable due to this exponential update is at least 
    \[ \frac{\gamma}{Mn} ( e^{\Gamma/2\la^h} - 1). \]
    By~\Cref{cl:uppery}, this is at most $4\gamma M + k$, %
    so we get
    $$ \Gamma = \la^h \cdot O \left( \ln n + \ln M + \ln (k/\gamma) \right)
    = \beta_0 c_u, $$
    hence showing that~\eqref{eq:dual-cons} is satisfied up to $\beta_0$ factor. 
    
  \item[II.] Suppose the dual constraint corresponds to some variable
    $\yv(u',\tau)$ with $u' \in T_u$, and $u \in \chi_{v}$. Suppose
    $u'$ is a node at height $h' < h$. Now let $\cons'$ be the
    constraints in $\cons^u$ (the LP for the child $u$) which
    contain $\y{u}(u',\tau)$. By the induction hypothesis:
    \begin{align}
      \label{eq:ind}
      \sum_{C \in \cons'} z_C \leq \beta_{h-1} \; c_{u'}.
    \end{align}
    Let $\cons''$ denote the set of constraints in $\cons^v$ (the LP
    for the parent $v$) which contain $\yv(u',\tau)$.  Each constraint
    $C(v, \sigma, \tau)$ in this set $\cons''$ has the coordinate
    $\sigma_u$ corresponding to the child $u$ being a constraint in
    $\cons'$, which implies:
    \begin{gather}
      \sum_{C(v, \sigma, \tau) \in \cons''} z_{C(v,\sigma,\tau)} =
      \sum_{C \in \cons'} \sum_{C(v, \sigma, \tau) \in \cons'': \sigma_u
        = C} z_{C(v,\sigma,\tau)} \leq (1 + \nicefrac1H)\; \sum_{C
        \in \cons'} z_C,
    \end{gather}
    where the last inequality uses \Cref{invar:duals-match}. Now the
    induction hypothesis~\eqref{eq:ind} and the fact that
    $\beta_{ h} = (1+\nicefrac1H)\,\beta_{h-1}$ completes the
    proof. \qedhere
  \end{enumerate}
\end{proof}

\Cref{lem:dual} means that the dual solution for $\cons^\rootvtx$ is
$\beta_H$-feasible, where $\beta_H = O(\ln \frac{nMk}{\gamma})$. This
proves \Cref{lem:dual-feasible} and completes the proof of our
fractional $k$-server algorithm.


\section{Algorithm for \texorpdfstring{\ksertw}{k-server-TW}}
\label{sec:main-windows}

In this section, we describe the online algorithm for \ksertw. The
structure of the algorithm remains  similar to that for
\kser. Again, we have a main procedure (\Cref{algo:mainnew}) which considers the backbone
consisting of the path from the requested leaf node to the root node.
It calls a suitable subroutine for each node on this backbone to add
local LP constraints and/or transfer servers to $v_0$.  We say that a request interval 
$R_q=[b,q]$ at a leaf node $\rqq$ becomes \emph{critical} (at time $q$)
if it has deadline $q$, and it has not been served until time $q$, i.e.,
if $k_{\rqq, t} < 1-2\delta'$ for all timesteps $t \in [b,q)$: for technical reasons  we allow a gap of up to $2 \delta'$ instead of $\delta'$. In case this node becomes critical at $q$, the algorithm ensures that  $\ell_q$ receives at least $1-\delta'$ amount of server at time $q$. This ensures that we move at least $\delta'$ amount of server mass when a request becomes critical. The parameters $\delta, \delta'$ remain unchanged, but we set $\gamma$ to $\frac{1}{n^4 \Delta}$. We extend
the definition of $\RL$ from \S\ref{sec:algodesc} in the natural way:
\begin{gather*}
  \RL(\tau) = \text{location of request with deadline at time
    $\floor{\tau}$, and }\\
  \RI(\tau) = \text{request interval with deadline at time $\floor{\tau}$}.
\end{gather*}

\begin{algorithm}
  \caption{Main Procedure for Time-Windows}
  \label{algo:mainnew}
  \ForEach{$q = 1,2, \ldots$}{
    \If{$\RI(q)$ exists and is critical }{
      let the path from
      $\rqq := \RL(q)$ to the root be $\rqq = v_0, v_1, \ldots, v_H =
      \rt$.\; %
      \myhl{let $Z_q, \{F_{v,q} \mid v \in Z_q\} \gets \BT(q)$} \;
      $\tau \gets$ $q+\eta$, the first timestep after $q$ \;
      \While{$k_{v_0, \tau} \leq 1 - \delta'$ \label{l:whilenew} }{
        let $i_0 \gets$ smallest index such that $\sib(v_{i_0},
        \tau) \neq \emptyset$.       \label{l:eqnew}  \;
        \lFor{$i=0,\ldots, i_0$}{call $\DUZ(v_i, \tau,
          \la^i \cdot \nicefrac{\gamma}{\la^{i_0}})$.
          \label{l:duzcall}
        }
        \lFor{$i=i_0+1,\ldots, H$}{call $\DU(v_i, \tau)$.  \label{l:call-dualnew}
        }
        $\tau \gets \tau+\eta$. \tcp*[f]{create a new
          timestep} \label{l:ffor3new}
      }
      \myhl{      serve requests at leaves in $\{F_{v_i,q} \mid v_i \in Z_q\}$ using
        server mass at $v_0.$ \label{l:serverest1} }
    }
  }
\end{algorithm}

Here are the main
differences with respect to~\Cref{algo:main}:

\begin{itemize}

\item[(i)] When we service a critical request at a leaf $\ell_q$, we would like to also serve active requests at nearby nodes. 
 The procedure $\BT(q)$ returns a set of backbone nodes
  $Z_q \sse \{v_0,\ldots, v_H\}$, and a tree $F_{v_i,q}$ rooted at each node
  $v_i \in Z_q$. In line~\eqref{l:serverest1}, we service all the
  outstanding requests at the leaf nodes of these subtrees
  $\{F_{v_i,q} \mid v_i \in Z_q\}$ using the server at $v_0$. (These are
  called \emph{piggybacked} requests.)

\item[(ii)] For a node $v_i$ with $i \leq i_0$, the previous \DUZ
  procedure in \S\ref{sec:simple-local-basic} would define the set
  $\cons^{v_i}(\tau)$ in the local LP $\fL^{v_i}$ to contain just one
  $\bot$-constraint. For the case of time-windows, we give a new \DUZ
  procedure in \S\ref{sec:duz-new}, which defines a richer set of
  constraints %
  based on a \emph{charging forest} $\chF(v_i)$. This procedure also
  raises some local dual variables; this dual increase was not
  previously needed in the case of the $\bot$-constraint. Finally, the
  procedure constructs the tree $F_{v_i,q}$ rooted at $v_i$ which is
  used for piggybacking requests. Although this construction of the
  charging tree is based on ideas used by~\cite{AzarGGP17} for the single-server
  case, we need a new dual-fitting analysis in keeping with our
  analysis framework.
  
\item[(iii)] We need a finer control over the amount of dual raised in
  the call $\DUZ$ %
  in line~\eqref{l:duzcall}. Fix a call to $\DUZ(v_i,\tau, \xi)$;
  hence $i \leq i_0$ at this timestep. To prove dual feasibility, we want the increase in
  the dual objective function value to match the cost (with respect to
  vertex $v_i$) of the server movement into $v_i$ during this
  iteration of the {\bf while} loop. This server mass entering $v_i$
  is dominated by the server mass transferred to the request location
  $v_0$ by $\DU(v_{i_0+1}, \tau)$, which is roughly
  $\nicefrac{\gamma}{\la^{i_0}}$. The cost of transferring this server
  mass to $v_i$ from its parent is
  $\la^i \cdot \nicefrac{\gamma}{\la^{i_0}}$. We pass this value as an
  argument $\xi$ to \DUZ in line~\eqref{l:duzcall}, indicating the
  extent to which we should raise dual variables in this
  procedure.

  Moreover, we need to remember these values: for each node $v$ and
  timestep $\tau \in  \cR^{ns}(v)$, we maintain a quantity
  $\Gamma(v, \tau)$, which denotes the total dual objective value
  raised for the constraints in $\fL^v(\tau)$.  If these constraints were
  added by $\DUZ(v, \tau, \xi)$, we define it as $\xi$; and finally,
  if they were added by $\DU(v, \tau)$ procedure, this stays equal to
  the usual amount $\gamma$ (as in the algorithm for $\kser$). In case %
  $\tau \in \cR^s(v)$, this quantity is undefined.

\end{itemize}

We first explain \BT and \BW in \S\ref{sec:build-tree}, which build
the set $Z_q$ and the trees to satisfy the piggybacked requests, and
the charging forest. Then we describe the modified local update
procedures in \S\ref{sec:duz-new} and
\S\ref{sec:local-updatenew}: the main changes are to \DUZ, but small changes also
appear in \DU.

\subsection{The \BT procedure}
\label{sec:build-tree}
\begin{figure}
    \centering
    \includegraphics[width=6in]{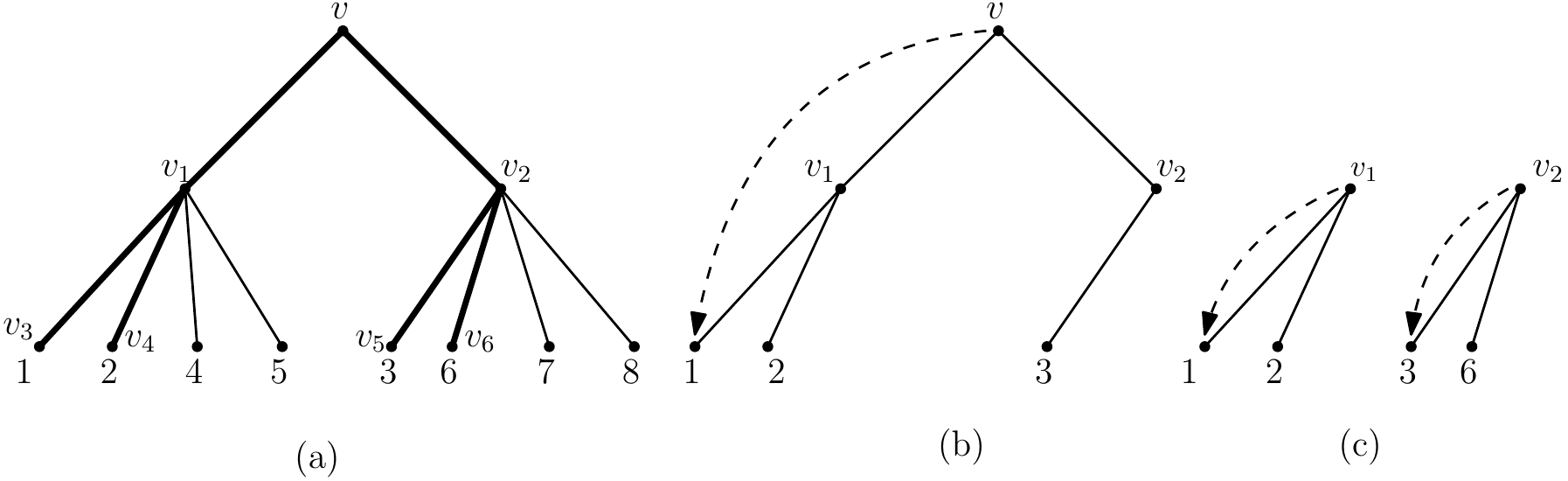}
    \caption{Example of \BT procedure when processing $v$: (a) tree
      rooted at $v$, with $c_v =5, c_{v_1} = c_{v_2} = 2$, all leaves
      have cost 1. For each leaf, the earliest deadline of an active
      request is shown.  $\FL(q,v)$ returns the subtree in (b) with
      $S=\{v_1, v_2\}$. $\FL(q,v_1)$ and $\FL(q,v_2)$ return trees in
      (c) with $S$ being $\{v_3, v_4\}$ and $\{v_5, v_6\}$
      respectively. The dashed arrows indicate the associated leaf
      requests. The heavier edges in (a) indicate the tree returned by
      $\BT(q,v)$. The nodes
      $(w,q)$ for $w \in \{v,v_1, v_2, v_3, v_4, v_5, v_6\}$ get added to
      $\chF(v)$.}
    \label{fig:BTproc}
\end{figure}

To find the piggybacked requests, the main procedure calls the \BT
procedure (\Cref{alg:btree}). This procedure first obtains an estimate
$\cost(q)$ of the cost incurred to satisfy the critical request at
time $q$, and defines $Z_q$ to be the first
$\floor{\log_\lambda 2\lambda \cost(q)}$ nodes on the backbone. The
estimate $\cost(q)$ is the minimum cost of moving servers to $\RL(q)$
so that it has $1-\delta'$ amount of server mass while ensuring that
all leaf nodes have at least $\delta-\gamma$ server mass. Since our
algorithm moves servers from active leaf nodes only, and $\DU$
procedure never moves more than $\gamma$ amount of server in one
function call (see \Cref{cl:dualtrnew}), $\cost(q)$ is a lower bound on the cost incurred by the algorithm to move server mass to $v_0$.
For
each node $v$ in $Z_q$, $\BT$ then finds a tree $F_{v,q}$ of cost at most
$H^2 \cdot c_v$.

Given a node $v \in Z_q$, the tree $F_{v,q}$ is built by calling the
sub-procedure \FL (\Cref{alg:findedge}) on nodes at various levels,
starting with node $v$ itself. (See \Cref{fig:BTproc}.) When called for a node $w$, \FL returns
a subtree $G$ of cost at most $H c_w$ by adding paths from $w$ to
some set of leaves. Specifically, it sorts the leaves in increasing
order of deadlines of the current requests (i.e., in \emph{Earliest
  Deadline First} order). It then adds paths from $w$ to these leaves one by
one until either (a) all leaves with current requests have been
connected, or (b)~the union of these paths contains some level with
cost at least $c_w$. In the latter case, \BT calls $\FL$ for the set
$S$ of nodes at this ``tight'' level. (If $\FL(q,w)$ returns a set of
nodes $S$, nodes in $S$ are said to be \emph{spawned} by $w$, and
necessarily lie at some level lower than $w$.) A simple induction
shows that the total cost of calls to $\FL(q,w)$ for nodes $w$ at any
level cost at most $H c_v$, and hence the tree $F_{v,q}$ returned by
$\BT(q)$ costs at most $H^2 c_v$.

\begin{algorithm}[H]
  \caption{$\BT(q)$}
  \label{alg:btree}
  $\cost(q) \gets$ min-cost to increase server at
  $\RL(q)$ to $1-\delta'$, $\lcost(q) \gets \floor{\log_\la
   (2\la\, \cost(q))}$\;
  $Z_q \gets \{ v_0, v_1, \ldots, v_{\lcost(q)} \}$ \tcp*[f]{$v_0 =
    \RL(q)$ is the request location, $v_1, v_2, \ldots$ are its
    ancestors.} \;
  \ForEach{$v \in Z_q$}{
    initialize a queue $Q \leftarrow \{v\}$, subtree $F_{v,q} \leftarrow \emptyset$. \;
    \While{$Q \neq \emptyset$}{  
      $w \gets$ dequeue$(Q)$. \;
      $(G, S) \leftarrow \FL(q,w)$; we say that nodes of $S$ are \emph{spawned} by $w$ at time $q$.\;
      $F_{v,q} \leftarrow F_{v,q} \cup G$. \;
      \lForEach{$u \in S$}{enqueue$(Q,u)$.}
      $\BW(q,w,v)$. \;  
    }
  }
  \Return set $Z_q$ and subtrees $\{F_{v,q}\}$. 
\end{algorithm}

\begin{algorithm}
  \caption{$\FL(q, w)$}
  \label{alg:findedge}
  $\ell_1, \ell_2, \ldots \gets$ leaves of $T_w$ in increasing order of
  deadline of outstanding requests at them. \;
  Initialize $G \leftarrow \emptyset$. \;
  \For{$i=1, 2, \ldots$}{
    add the $\ell_i$-$w$ path in $T_w$ to the subtree $G$. \;
    \If{cost of vertices in $G$ at some level $\ell$ (strictly
      below $w$'s level) is at
      least $c_w$}{ \label{l:stop}
      \Return $(G,  S)$, where $S$ is the set of vertices at level $\ell$ in $G$.
    }
  }
  \Return $(G, \emptyset)$.
\end{algorithm}

For each node $w$ that is either the original node $v$ or else is
spawned during \FL, 
the algorithm calls the procedure $\BW(q,w,v)$ to construct the
charging tree: we describe this next.

\subsubsection{\BW and the Charging Forest}
\label{sec:charging-forest}

\begin{algorithm}[H]
  \caption{$\BW(q,w,v)$}
  \label{alg:bwitness}
  \textbf{add} node $a := (w,q)$ to $\chF(v)$ \;
  \textbf{let} $(\ell,I) \leftarrow \leafreq(w,q).$ %
  \label{l:bwleaf} \;
  \textbf{let} $q' \gets \arg\max\{ q'' \mid q'' < q , (w,q'') \in \chF(v) \}$. \;
  \If{$q' \in I$}{
    \ForEach{node $w'$ spawned by $w$ in \BT$(q',v)$}{
      make $(w',q')$ a child of $(w,q)$ in the witness forest $\chF(v)$.
    }
  }
\end{algorithm}

\begin{figure}
    \centering
    \includegraphics[width=5.5in]{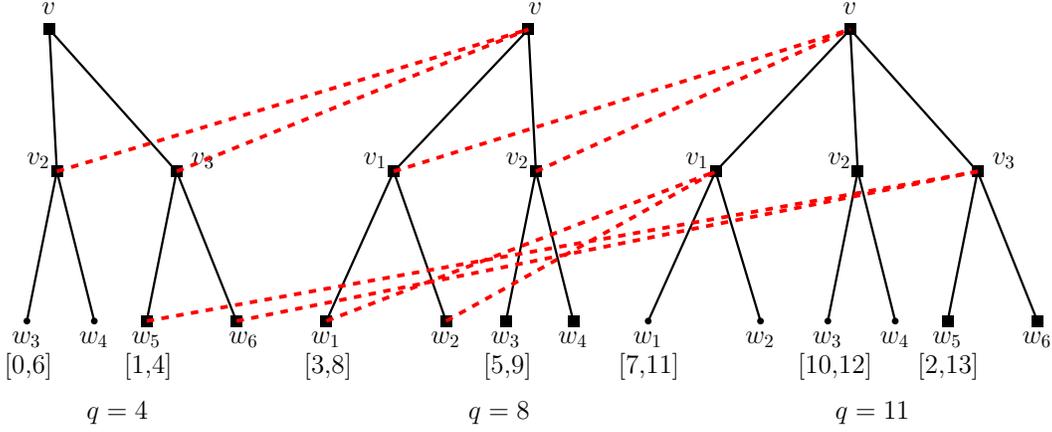}
    \caption{The trees $F_{v,q}$ returned by $\BT(q)$ for node $v$ at times $q=4,8,11$ are shown
      in black/solid. For each, the bold squares are nodes that lie in
      $S$. These are also nodes of the charging forest $\chF(v)$,
      whose edges are shown in red/dashed. The request intervals for only
      the relevant leaf nodes are shown: each bold square is
      associated with the leaf request below it with the earliest end
      time. For example, the interval associated with $v_3$ at time
      $q=11$ is $[2,13]$: since this interval is active at time $q=4$,
      we have edges from $(v_3, q=11)$ to nodes $(w_5, q=4), (w_6, q=4)$
      spawned by $v_3$ at time $q=4$. }
    \label{fig:Charge}
\end{figure}
Each node $v$ maintains a charging forest $\chF(v)$, which we use to
build a lower bound on the value of the optimal solution for servicing
the outstanding requests below $v$, assuming there is just one
available server. The construction here is inspired by the analysis
of~\cite{AzarGGP17}. We use this charging forest to add constraints to
$\fL^v$ (during \DUZ procedure) and to build a corresponding dual
solution. We need one more piece of notation: for node $w$ and time
$q$, let $\ell$ be the leaf below $w$ such that the active request at
$\ell$ has the earliest deadline after $q$. (In case no active request lies
below $w$ at time $q$, this is undefined). Let $I$ be the
corresponding request interval at $\ell$. We use $\leafreq(w,q)$ to
denote the pair $(\ell, I)$.

The procedure $\BW(q,w,v)$ adds a new vertex called $(w,q)$ to the charging forest $\chF(v)$.
To add edges, let
$q' < q$ be the largest time such that $\chF(v)$ contains a vertex
of the form $(w,q')$. Let $(\ell, I)$ denote $\leafreq(w,q)$.
If time
$q'$ also belongs to $I$, we add in an edge from $(w,q)$ to $(w',q')$ for every node $w'$ that was spawned by $w$ in
the call to $\BT(v,q')$. (See~\Cref{fig:Charge}.)

Here's the intuition behind this construction: at time $q'$, there
were outstanding leaf requests below each of the nodes $w'$ which were
spawned by $w$. The reason that interval $I$ was not serviced at time
$q'$ (i.e., the leaf $\ell$ was not part of the tree returned by
$\BT(q',v)$) was because the intervals chosen in that tree were preferred over
$I$, and the total cost of servicing them was already too high. This
allows us to infer a lower bound.

\subsection{Reminder: Truncated Constraints}
\label{sec:reminder-of-trunc}

We now describe the procedures $\DUZ$ and $\DU$ in detail; both of these procedures will add (truncated) constraints of the form $\varphi_{A,f,\btau,v}$ to the local LP for a node $v$ as defined in~\eqref{eq:implied1new}. For sake of completeness, we formally define this notion here:
\begin{defn}[Truncated Constraints]
\label{def:twtrunc}
Consider  a node $v$, a subset $A$ of nodes in $T_v$ (where no two them have an ancestor-descendant relationship), a function $f: A \to \mathfrak{R}$ 
mapping each node $u \in A$ to a request $(\ell_u, [b_u,e_u])$ at some leaf $\ell_u$ below $u$, and an assignment $\tau_u$ of timesteps to each $u \in T^A_v$. 
The timesteps $\btau$ must satisfy the following two (monotonicity)
properties: (a) For each node $u \in T^A_v$, $\tau_u \geq \max_{w \in A \cap T_u} e_w$; (b) If $v_1, v_2$ are two  nodes in $T^A_v$ with $v_1$ being the ancestor of $v_2$, then $\tau_{v_1} \geq \tau_{v_2}.$
Given such a tuple
$(A, f,\btau, v)$, 
the truncated constraint $\varphi_{A,f,\btau, v}$ (ending at timestep $\tau_v$) is defined as follows:
\begin{align*}
  \sum_{u \in A \cap T_v, u \neq v} \yv(u, (b_u, \tau_{p(u)}]) + \sum_{u \in T^A_v \setminus A, u \neq v} \yv(u, (\tau_u, \tau_{p(u)}]) \geq |A \cap T_v| - k_{v, \tau_v} -2\delta(n-n_v).
\end{align*}
\end{defn}

\subsection{The Simple Update Procedure}
\label{sec:duz-new}

The \DUZ procedure is called with parameters: node $v_i$, timestep $\tau$ with
$\floor{\tau} = q$, and target dual increase $\xi$.  In
the case without time-windows, this procedure merely added a single
$\bot$-constraint. Since we may now satisfy requests due to
piggybacking, the new version of \DUZ adds other constraints
and raises the dual variables corresponding to them.

Recall that $\BT$ defines an estimate $\cost(q)$ %
and sets $\lcost(q) = \floor{\log_\la 2\la \cost(q)}$.
After 
defining $\Gamma(v_i, \tau) := \xi$, \DUZ tries to add a constraint to
$\fL^{v_i}$---for this purpose we use the highest index $\ist \leq i$
for which we have previously added a node to the charging forest
$\chF(v_\ist)$ at time $q$. Hence we set $\ist := \min(i,
\lcost(q))$. As explained in \S\ref{sec:charging-forest},
$\chF(v_\ist)$ has a tree rooted at $(v_\ist,q)$, call it $\cT$.  The
algorithm now splits in two cases:

\begin{enumerate}
\item[(i)] Tree $\cT$ is just the singleton
  vertex $(v_\ist,q)$: we add a $\bot$-constraint in
  line~\eqref{l:addp1new} and add $\tau$ to $\cR^s(v_i)$. The intuition is that the tree $\cT$ gives
  us a lower bound for serving the piggybacked requests. So if it has
  no edges, we cannot add a non-$\bot$ constraint. %

\item[(ii)] Tree $\cT$ has more than one vertex: in this case we add a (non-$\bot$) constraint to $\fL^{v_i}$, details of which are described below. 

\end{enumerate}

\begin{algorithm}
  \caption{$\DUZ(v_i, \tau, \xi)$}
   $\Gamma(v_i,\tau) \gets \xi$, add $\tau$ to $\awake(v_i)$. \;
   $i^\star \gets \min(i, \lcost(q))$.  \;
   \eIf{the charging tree $\cT$ in $\chF(v_\ist)$ containing $(v_\ist,q)$ is a singleton}{
     $\cons^{v_i}(\tau) \gets$ $\bot$-constraint for
     $\varphi_{A,f,\btau, v_i}$, where $A = \{v_i\}, f(v_i) =
     (\RL(\tau), \RI(\tau)), \tau_{v_i} = \tau$; add $\tau$ to
     $\cR^s(v_i)$ \label{l:addp1new} \tcp*[f]{solitary timestep for $v_i$} %
   }
   {
     Let $s$ be such that level-$s$ leaves $\cL_s$ in $\cT$ have cost
     at least $c_{v_\ist}/H$. \tcp*[f]{$\cT$ not a singleton} \;
     \ForEach{$a_j=(u_j, q_j) \in \cL_s$}
     {
       Let $(\ell_j, R_j=[b_j, e_j])$ be $\leafreq(a_j)$ as defined in line~\eqref{l:bwleaf} of $\BW(q_j, u_j, v_\ist)$. \label{l:leafduz} \; 
       add $u_j$ to $A$, define $\tau_{u_j} \leftarrow b_i$, $f(u_j) = (\ell_j, R_j)$.  \;
   }
   define $\tau_w \gets \tau$ for each internal node $w$ in $T^A_{v_i}$. \;
   add the constraint $\varphi_{A,f,\btau,v_i}$ (as shown
   in~\eqref{eq:consz}) to $\cons^{v_i}(\tau)$, and set the dual
   variable accordingly; add $\tau$ to $\cR^{ns}(v_i)$
   \tcp*[f]{non-solitary timestep for $v_i$} \label{l:addp2new}
   } 
\end{algorithm}

It remains to describe how to add the local constraints and set the dual variables in case $\cT$ contains more than one node. Recall that $\cT$ is rooted at $(v_\ist,q)$. 
The \BT procedure ensures that the nodes spawned by any node
$w$ cost at least $c_w$; applying this inductively ensures that if
$\cL$ is the set of leaves of this tree $\cT$, we have
$\sum_{a \in \cL} c_a \geq c_{v_\ist}$. (Here we abuse notation by defining
the cost of a tuple $a = (w, q)$ in $\cT$ to equal the cost of the
node $w$.) Hence, there is some level $j$ such that leaves in $\cT$
corresponding to level-$j$ nodes have cost at least $\nf{c_{v_\ist}}H$. Let these
leaves of $\cT$ be denoted
$\cL_s := \{ a_j = (u_j, q_j) \}_{j = 1}^r$.

For each leaf $a_j = (u_j, q_j) \in \cL_s$ of this charging tree, 
let $(\ell_j, R_j = [b_j,e_j])$ denote $\leafreq(a_j)$ (as in
line~\eqref{l:bwleaf} of $\BW(q_j, u_j, v_i)$). 
Define
$A := \{u_j \mid \exists q_j \text{ s.t. } (u_j,q_j) \in \cL_j\}$ to
be the subset of nodes of the original tree $T$ corresponding to the
nodes $\cL_s$ from the charging tree, and define
$f: u_j \mapsto (\ell_j, R_j)$. Recall that $T^A_{v_i}$ denotes the
minimal subtree rooted at $v_i$ and containing $A$ (as
leaves). %
For each node $u_j \in A$, define $\tau_{u_j} = b_j$. Define the
timestep $\tau_w$ for each internal node $w$ in $T^A_{v_i}$ to be
$\tau$. (Note that $\cT$ was rooted at $v_\ist$, but we define $\tau$
for the portion of the backbone from $v_\ist$ up to $v_i$ as well.)
 We will show in~\Cref{cor:alive} that setting
$\tau_{v_\ist}$ to $\tau$ does not violate the monotonicity property,
i.e., $e_j \leq q \leq \tau$ for all request intervals
$R_j$.

Now we add to $\cons^{v_i}(\tau)$ the truncated constraint $\psi_{A,f,\btau,v_i}$, which can be written succinctly as
\begin{align}
\label{eq:consz}
\sum_{u_j \in A} y^{v_i}(u_j, (b_j, \tau])  
  \geq |A| - k_{v_i, \tau_{v_i}} -2\delta(n-n_{v_i}) ,
\end{align}

Observe that the RHS above is positive because $|A| \geq 1$ and $k_{v_i, \tau_{v_i}} < 1-\delta'$. 
Finally, we set the dual variable for this single constraint to
$\xi/(|A| - k_{v_i, \tau_{v_i}} -2\delta(n-n_{v_i}))$, so that the dual objective
increases by exactly $\xi$. We end by declaring timestep
$\tau$ non-solitary, and hence adding it to $\cR^{ns}(v_i)$.

\subsection{The Full Update Procedure}
\label{sec:local-updatenew}

\newcommand{\VI}{v}
\newcommand{\VIM}{u_0}
\newcommand{\yvi}{\yv{\VI}}

The final piece is procedure $\DU(\VI, \tau, \gamma$).  This is
essentially the version in \S\ref{sec:full-local-basic}, with one
change. Previously, if $\sib(\VIM, \tau)$ was not empty, we could have
had very little server movement, in case most of the dual increase was
because of $b^{C_{\VIM}}$. To avoid this, we now force a non-trivial
amount of server movement. When the dual growth reaches $\gamma$, we
stop the dual growth, but if there has been very little server
movement, we transfer servers from active leaves below
$\sib(\VIM, \tau)$ in line~\eqref{l:more}. 

The intuition for this step is as follows:
in the $\DUZ(v_{i}, v_0, \xi)$ procedure for $v_i$ below $v$, we need to
match the dual increase (given by $\xi$) by the amount of server that
actually moves into $v_{i}$. This matching is based on the assumption
that at least $\gamma/\la^h$ transfer happens during the \DU
procedure. By adding this extra step to $\DU$, we ensure that a
roughly comparable amount of transfer always happens.

Finally, let us elaborate on the constraint $C(\VI,\sigma,
\tau)$. This is written as in~\eqref{eq:cons}, using the modified
composition rule for \ksertw from~\Cref{cor:comp}. Since we did not
spell out the details, let us do so now. As before, $\VIM$ is the
principal child of $\VI$ at $\tau$, and
$U := \{\VIM\} \cup \sib(\VIM, \tau)$.  Each of the constraints
$C_u \in \fL^u(\tau_u), u \in U$ has the form
$\varphi_{A(u), f(u), \btau(u), u}$ for some  tuple
$(A(u),f(u),\btau(u))$ for node $u$ ending at $\tau_u
:=\tau(u)_u$. Partition the set $U$ into two sets based on whether
$\cons^{u}(\tau_u)$ is a $\{\bot\}$-constraint (i.e., whether $\tau$
is in $\cR^s(u)$ or in $\cR^{ns}(u)$):
$U' := \{ u \in U: C_u \text{ is a $\bot$-constraint}\}$, and
$U'' := U \setminus U'$. Recall that $I_u$ denotes the interval
$(\tau_u, \tau]$.  For a node $u \in U'$, the $\bot$ constraint is
given by $\varphi_{A(u), f(u), \btau(u), u}$, where $A(u)=\{u\}$, and
let $b_u$ is the starting time of the request interval corresponding
to $f(u)$. Let $I_u'$ denote the interval $(b_u, \tau_u]$.
The new constraint $C(\VI,\sigma, \tau)$ is the composition
$\varphi_{A,f,\btau,\VI}$ of these constraints, and by~\Cref{cor:comp}
implies:
\begin{align}
  \label{eq:consnew}
  \underbrace{\sum_{u \in U'} \yv(u, I_u') + \sum_{u \in U} \left( \yv(u,
  I_u) + a^{C_u} \cdot \yv \right)}_{a^{C(\VI,\sigma,\tau)} \cdot \yv} \geq 
  \underbrace{\sum_{u \in U} \left( \Diff(u, I_u) + b^{C_u} \right)  + 
  (n_{\VI} - \sum_{u \in U} n_u) \delta}_{\leq b^{C(\VI,\sigma, \tau)}}.
\end{align}

Observe that the dual update process itself in $\DU$ remains unchanged
despite these new added variables corresponding to $I_u'$: these
variables $\{\yv(u,\tau)\}_{\tau \in I_{u}', u \in U'}$ do not appear in
line~\eqref{l:draise2new}. Hence all the steps here \emph{exactly
  match} those for the \kser setting, except for
line~\eqref{l:more}. This completes the description of the local
updates, and hence of the algorithm for \ksertw.

\begin{algorithm}[H]
  \caption{$\DU(v, \tau$)}
  \label{algo:dualnew}
  let $h \gets \level(v)-1$ and $u_0 \in \chi_{v}$ be child containing
  the current request $v_0 := \RL(\tau)$. \;
  let $U \gets \{u_0\} \cup \sib(u_0,\tau)$; say $U = \{u_0, u_1,
  \ldots, u_\ell\} $, $L_U \leftarrow $ active leaves below $U \setminus \{u_0\}$. \;
  add timestep $\tau$ to the event set $\cR^{ns}(v)$ and to  $\awake(v)$. \tcp*[f]{``non-solitary''
    timestep for $v$}\;
  \myhl{set timer $s \gets 0$, $\Gamma(\VI,\tau) \gets \gamma$.} \;
  \Repeat{the dual objective corresponding to constraints in $\cons^v(\tau)$ becomes $\gamma$.}{
    \For{$u \in U$}{
      let $\tau_u \gets \prev(u,\tau)$ and $I_u = (\tau_u,\tau]$.
          \label{l:Iunew} \;
      let $C_u$ be a slack constraint in $\cons^u(\tau_u)$. \tcp*[f]{slack constraint exists since $\prev(u,\tau)$ is awake}
      \label{l:chooseCnew}
    }
    let $\sigma \gets (C_{u_0}, C_{u_1}, \ldots, C_{u_\ell})$ be the resulting tuple of
    constraints. \;
    add new constraint $C(v,\sigma, \tau)$ to the constraint set $\cons^v(\tau)$.\;

    \While{all constraints $C_{u_j}$ in $\sigma$ are slack
      \textbf{and} dual objective for $\cons^v(\tau)$ less than $\gamma$}{ \label{l:dualnew}
      increase timer $s$ at uniform rate. \;
      increase $z_{C(v,\sigma, \tau)}$ at the same rate as $s$. \label{l:draise1new}\;
      for all $u \in U$, define $S_u := I_u \cap \left( \cR^{ns}(u) \cup \{\tau_u+\eta\} \right).$ \label{l:sunew} \;
      increase $\yv(u, t)$ for $u \in U, t \in S_u$
      according to 
      $\frac{d \yv(u,t)}{d s} = \frac{\yv(u,t)}{\lambda^{h}} +
        \frac{\gamma}{Mn \cdot \lambda^h}$.     \label{l:draise2new}
      \;
      \underline{\textbf{transfer}} server mass from $ T_u$ into $v_0$ at rate $\frac{d
        \yv(u,I_u)}{d s} + \frac{
          b^{C_u}}{\lambda^h}$ using the leaves in $L_U \cap T_u$, for each $u \in U \setminus \{u_0\}$
     \label{l:dtrnew}\;
    }
    \ForEach{constraint $C_{u_j}$ that is depleted}{
      \label{l:mark2new}
      \lIf{\textit{all} the constraints in $\cons^{u_j}(\tau_{u_j})$
        are depleted}{remove $\tau_{u_j}$ from $\awake(u_j)$.     }
    }
  }
  let $\alpha \leftarrow$ total amount of servers transferred to $v_0$ during this function call. \;    
  \myhl{\If{$\sib(\VIM, \tau) \neq \emptyset$ and  $\alpha < \frac{\gamma}{4H \la^h}$}{
       transfer more servers from $L_U$ to $v_0$ until total transfer equals $\frac{\gamma}{4H \la^h}$. 
       \label{l:more}    
    }}
\end{algorithm}


\section{Analysis for \texorpdfstring{\ksertw}{k-server-TW}}
\label{sec:tw-analysis}

The analysis for \ksertw closely mirrors that for \kser; the principal
difference is due to the additional intervals $I_u'$ on the LHS
of~\eqref{eq:consnew}. If the intervals $I_u'$ are very long, we may
get only a tiny lower bound for the objective value of the LPs:
raising only a few $\yv$ variables variables could satisfy all such
constraints. The crucial argument is that the intervals $I_u'$ are
disjoint for any given vertex $v$ and descendant $u'$: this gives us
approximate dual-feasibility even with these $I_u'$ intervals, and
even with the dual increases performed in the \DUZ procedure. To show
this disjointness, we have to use the properties of the charging
forest.  A final comment: timesteps in $\cR^{ns}(v)$ are now added by
both $\DU$ and $\DUZ$, whereas only $\DUZ$ adds timesteps to
$\cR^s(v)$.

\subsection{Some Preliminary Facts}

\begin{claim}[Facts about $\Gamma$]
  \label{cl:GammaFacts}
  Fix a node $u$ with parent $v$, and timestep  $\tau \in \cR^{ns}(u)$. 
  \begin{OneLiners}
  \item[(i)] \label{cl:matchinggamma} $\frac{\gamma}{\la^H} \leq \Gamma(u, \tau) \leq \gamma$.
  \item[(ii)] \label{cl:gammarange} 
  If $\DU(v, \tau)$ is called, then
    $\Gamma(u,\tau) = \gamma$.
  \item[(iii)] \label{cl:dualsimple} 
   If $\tau$ gets added to $\cR^{ns}(u)$ by $\DUZ$ procedure, then 
    the dual objective value for the sole constraint in
    $\cons^u(\tau)$ is $\Gamma(u, \tau)$.
  \end{OneLiners}
\end{claim}

\begin{proof}
  The first claim follows from the fact that $\Gamma(u,\tau)$ is
  either set to $\gamma$ (in \DU) or
  $\lambda^i \frac{\gamma}{\lambda^{i_0}}$ (in \DUZ), and that
  $1 \leq i \leq i_0 \leq H$ in line~\eqref{l:duzcall}.
  For the second claim, if $\tau$ gets added to $\cR^{ns}(u)$ by $\DU$, then the statement follows immediately. Otherwise it must be the case that 
  $u = v_{i_0}$, and we  call
  $\DUZ(u, \tau, \gamma)$ in line~\eqref{l:duzcall} of ~\Cref{algo:mainnew}, giving
  $\Gamma(u, \tau) = \gamma$ again.

  For the final claim, observe that $\cons^u(\tau)$ contains a single
  constraint $C$ given by~\eqref{eq:consz}, and we set $z(C)$ to be
  $\frac{\Gamma(u, \tau)}{|A| - k_{u, q_u} -2\delta(n-n_u)}$. 
\end{proof}

\begin{claim}[Facts about $Z_q$]
  \label{cl:addi0}
  Suppose the leaf $\ell_q$ becomes critical at time $q$, and $v_i$ is an ancestor
  of $\ell_q$ such that all leaves in $T_{v_i}$ (including $\ell_q$) are inactive at time $q$. Then $v_{i+1}$ gets added to the
  set $Z_q$.
\end{claim}

\begin{proof}
  We claim that $\lcost(q)$ is at least $i+1$. Since there are no
  active leaves in $T_{v_i}$, all the server mass needs to be brought
  into $\ell_q$ from leaves which are outside $T_{v_i}$, and so the
  total cost of this transfer is at least $(1-\delta'-\delta) c_{v_i} =
  (1-\delta'-\delta) \la^i$. It follows that $\lcost(q) = \floor{\log_\la (2 \la \, \cost(q))} \geq i+1$. 
\end{proof}

\newcommand{\istr}[1]{i^\star_{#1}}

\subsection{Congestion of Intervals for \texorpdfstring{$\bot$}{bot}-constraints}

Recall from line~\eqref{l:addp1new} that $\cons^v(\tau)$ is a
$\bot$-constraint (i.e., timestep $\tau \in
\cR^s(v)$) %
exactly when the component $\cT$ of the charging forest
$\chF(v_\ist)$ containing the vertex
$(v_\ist,\floor{\tau})$ is a singleton.

\begin{lemma}[Low Congestion I]
  \label{cl:lowcong1}
  For a vertex $v$, let $Q$ be a set of times such that for each $q \in
  Q,$ there exists a timestep $\tau_q \in
  \cR^s(v)$ satisfying $\floor{\tau_q} = q$.
  Let $\ell_q := \RL(q)$ and $[b_q, q] =
  \RI(q)$ be the request location and interval corresponding to time
  $q$. Then the set of intervals $\{[b_q,q]\}_{q \in
    Q}$ has congestion at most $H$.
\end{lemma}

\begin{proof}
  For brevity, let $J_q := [b_q, q]$; $\istr{q}$ be the value of
  $\ist$ used in the call to $\DUZ$ on $v$ at timestep $\tau_q$ that
  added the $\bot$-constraint.

  \begin{claim}
    \label{cl:monot}
    Suppose there are times $p < q \in Q$ such that $p \in J_q$. Then
    $\istr{p} < \istr{q}$.
  \end{claim}

  \begin{proof}
    Let $v_{\lca}$ be the least common ancestor of leaves $\ell_p$ and
    $\ell_{q}$, and $v_m$ be the higher of $v_{\lca}$ and
    $v_{\istr{p}}$. 
    We first
    give two useful subclaims.

    \begin{subclaim}
      \label{cl:upper1}
      Let $v_i$ be an ancestor of $\ell_{q}$. Suppose $v_i$ gets
      added to the set $Z_{q'}$ for some $q' \in [p,q)$. Then the
      set $S$ returned by $\FL(q',v_i)$ is non-empty.
    \end{subclaim}

    \begin{subproof}
      Since (a)~the request at $\ell_{q}$ starts before $p$ (and hence
      before $q'$), (b)~the node $\ell_q \in T_{v_i}$, and (c)~the
      request at $\ell_q$ is not serviced until time $q$ and hence is
      still active at time $q'$, the set $S$ cannot be empty.
    \end{subproof}

    \begin{subclaim}
      \label{cl:upper2}
      If $\istr{q} < m$, then there exists a $q' \in [p,q)$ such that
      $v_{\istr{q}} \in Z_{q'}$. 
    \end{subclaim}
    \begin{subproof}
      \begin{figure}
        \centering
        \includegraphics[width=1.5in]{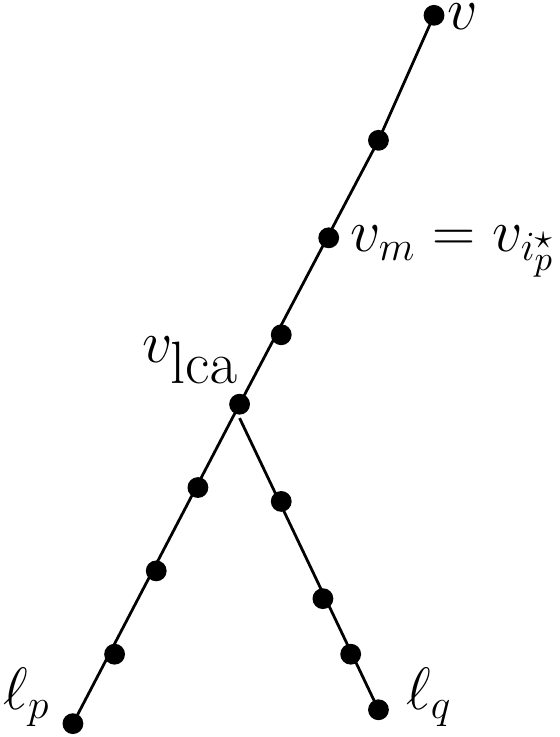}
        \caption{Illustration of proof of~\Cref{cl:upper2}: we consider the interesting case when $v_m = v_{\istr{p}}.$}
        \label{fig:proof-upper2}
    \end{figure}
    First consider the case when $\istr{q} < \lca$.  At time $p$, the
    fact that $\tau_p \in \rs(v)$ implies that no leaf in $T_v$ other
    than $\ell_p$ is active (i.e., has more than $\delta$ amount of
    server mass). Therefore, at time $p$, no leaf below $v_{\istr{q}}$
    is active. We claim that there must have been a time
    $q' \in (p,q)$ at which a request below $v_{\istr{q}}$ became
    critical. Indeed, if not, all leaves below $v_{\istr{q}}$ continue
    to remain inactive until time $q$. But then
    $\cost(q) \geq (1-\delta'-\delta) \la^{\istr{q}}$, and so
    $\lcost(q) > \istr{q}$, a contradiction. So let $q'$ be the first
    time in $(p,q)$ when a request below $v_{\istr{q}}$ became
    critical. Repeating the same argument shows that
    $\istr{q'} \geq \istr{q}$, and so $v_{\istr{q}}$ would be added to
    $Z_{q'}$.
 
    The other case is when $\lca \leq \istr{q} < m$, which means that
    $m > \lca$ and so $m = \istr{p}$. In that case $v_{\istr{q}}$ is
    added to the set $Z_p$ itself.
 \end{subproof}
    
 Now if $\istr{q} < m$, then \Cref{cl:upper2} says that $v_{\istr{q}}$
 is added to some $Z_{q'}$ for $q' \in [p,q)$. By \Cref{cl:upper1} the
 set returned by $\FL(q', v_{\istr{q}})$ is non-empty: this means
 $(v_{\istr{q}},q)$ cannot be a singleton component. This would
 contradict the fact that $\tau_q \in \cR^s(v)$. Similarly, if
 $\istr{q} = m = \istr{p}$, then $v_{\istr{p}} = v_{\istr{q}}$ and the
 argument immediately above also holds for $q' = p$.  Hence, it must
 be that $\istr{q} > \istr{p}$, which proves \Cref{cl:monot}.
  \end{proof}

  \Cref{cl:monot} implies that $p$ belongs to at most $H$ other
  intervals $\{J_{q}\}_{q \in Q}$. Indeed, if $p$ lies in the
  intervals for $q < q' < \cdots$, then $q$ also lies in the
  interval for $q'$, etc. Hence the $\ist$ values for
  $p, q, q', q'', \ldots$ must strictly increase, but then there can be
  only $H$ of them, proving \Cref{cl:lowcong1}.
\end{proof}

\subsection{Relating the Dual Updates to \texorpdfstring{$\ist$}{i*}}
\label{sec:relating-ist-dual}

We first prove a bound on the number of iterations of the {\bf while}
loop in~\Cref{algo:mainnew}: this uses the lower bound on the server
transfer that is ensured by  %
line~\eqref{l:more}).

\begin{claim}
  \label{cl:duiterations}
  Suppose a request at $v_0$ becomes critical at time $q$. The total number of iterations of {\bf while} loop in~\Cref{algo:mainnew} is at most $\frac{8H \cost(q)}{\gamma }$. 
\end{claim}
\begin{proof}
  Let the ancestors of $v_0$ be labeled $v_0, v_1, \ldots, v_H$. 
  If
  the cheapest way of moving the required mass of servers to $v_0$
  at time $q$ moves $\alpha_i$ mass from the active leaves which are
  descendants of siblings of $v_i$, then
  $\cost(q) = \sum_i \alpha_i c_{v_i}$.

  For an ancestor $v_i$ of
  $v_0$, define $t_i$ to be the earliest timestep by which either the algorithm moves at least  $\alpha_i$ server mass from active leaves below the siblings of $v_i$ to $v_0$, or $\sib(v_i, t_i)$ becomes empty.
 Since we transfer at least
  $\nicefrac{\gamma}{4H c_{v_i}}$ amount of server mass from leaves below the siblings of $v_i$  to $v_0$
  during each timestep in $(q, t_i]$,  the number of timesteps in
  $(q, t_i]$ cannot exceed $\nicefrac{4H c_{v_i} \alpha_i}{\gamma}$.
  
  During the algorithm, the set of active siblings of a node $v_i$ may become empty at $t_i$ while leaving up to $\delta$ amount of server mass at a some leaves below the siblings of $v_i$. While calculating $\cost(q)$, we had allowed leaving only $\delta-\gamma$ amount of server at a leaf, and so it is possible that the algorithm may move an additional $\gamma n$ amount of server mass beyond what has been transferred by $\max_i t_i$. Since we move at least $\frac{\gamma}{4 H \Delta}$ amount of server in each call to $\DU$ procedure, it follows the total number of such calls (beyond $\max_i t_i$) would be at most $4H \Delta n.$
  Therefore, the  total number of timesteps before we satisfy the request at
  $v_0$ is  at most
  \begin{equation*}
  4H \Delta n + 
   \max_i \frac{4H c_{v_i} \alpha_i}{\gamma} \leq \frac{4 H \delta'}{\gamma} + \sum_i \frac{4H
     c_{v_i} \alpha_i}{\gamma} = \frac{8H \cost(q)}{\gamma},
 \end{equation*}
 where we have used the fact that $\cost(q) \geq \delta' \geq \gamma n \Delta$. 
\end{proof}

Next, we relate $\ist$ from the \DUZ procedure to the increase in the
dual variables. %
\begin{lemma}
  \label{lem:istar}
  Suppose a request at $v_0$ becomes critical at time $q$.  Let
  $v_0, v_1, \ldots, v_H$ be the path to the root. For indices
  $i' \leq i$, let $S(i,i')$ be the set of timesteps $\tau$ such that
  (a) $\floor{\tau} = q$, and (b) we call $\DUZ(v_i, \tau, \xi_\tau)$
  for some value of $\xi_{\tau}$, and (c) $i^\star = i'$ during this
  function call.  Then
  \begin{equation*}
    \sum_{\tau \in S(i,i')} \xi_\tau \leq 12H c_{v_{i'}}.
  \end{equation*}
\end{lemma}

\begin{proof}
  Suppose that $i' = \ist < i$, then $i'=\lcost(q)$ for all timesteps
  $\tau \in S(i,i')$. Since the parameter $\xi_\tau \leq \ga$ for any
  timestep $\tau \in S(i,i')$  by
  \ref{cl:GammaFacts}(i), \Cref{cl:duiterations} implies that 
  $ \sum_{\tau \in S(i,i')} \xi_\tau \leq |S(i,i')|\; \gamma \leq 8H \cost(q). $
  But $\cost(q) \leq \la^{\lcost(q)} = \la^{i'} = c_{v_{i'}}$, which
  completes the proof of this case.

  The other case is when $i' = i$. We claim that for any timestep
  $\tau \in S(i,i')$, at least $\frac{\xi_\tau}{4H c_{v_{i}}}$ amount
  of server reaches the requested node. Indeed, we know that
  $i_0 \geq i$ at this timestep, so line~\eqref{l:more} of the \DU
  procedure ensures that at least
  $\frac{\gamma}{4H \la^{i_0}} = \frac{\xi_\tau}{4H c_{v_i}}$ amount
  of server reaches $v_0$, where we used that
  $\xi_{\tau} = \la^{i_0 - i} \ga$.  Since at most one unit of server
  reaches $v_0$ when summed over all timesteps corresponding to $q$,
  we get
  \begin{equation*}
   \sum_{\tau \in S(i,i')} \xi_\tau \leq 4H c_{v_i}. \qedhere 
 \end{equation*}
\end{proof}

\subsection{Proving the Invariant Conditions}
\label{sec:prov-invar-cond}

We begin by stating the invariant conditions and show that these are satisfied. 
Invariant~\ref{invar:dual-raised} statement only changes slightly: we
replace $\gamma$ by $\Gamma(v, \tau)$ as given below.

\begin{leftbar}
  \begin{invariant}
    \label{invar:dual-raisednew}
    At the end of each timestep $\tau \in \cR^{ns}(v)$, the objective function value of the dual variables corresponding
    to constraints in $\cons^v(\tau)$ equals $\Gamma(v,\tau)$. I.e., if a
    generic constraint $C$ is given by
    $\langle a^C \cdot \yv \rangle \geq b^C$, then
    \begin{align}
      \label{eq:dualinvariantnew}
      \sum_{C \in \cons^v(\tau)} b^C \cdot z_C = \Gamma(v,\tau) \qquad \forall
      \tau \in \cR^{ns}(v). \tag{I5}
    \end{align}
    Furthermore, $b^C > 0$ for all $C \in \fL^v(\tau)$ and $\tau \in \cR(v)$.
  \end{invariant}
\end{leftbar}

~\Cref{cl:dualsimple} shows that the invariant above is satisfied
whenever $\tau$ gets added to $\cR^{ns}(v)$ by $\DUZ$, and the second
statement follows from the comment after~\eqref{eq:consz}.  As before,
the quantity $\loss(u, \tau)$ is defined by~\eqref{eq:lossdef}
whenever $\DU(v, \tau)$ is called, $v$ being the parent of $u$. The
invariant condition~\ref{invar:movement} is replaced by the following
which also accounts for the extra transfer which happens during
line~\eqref{l:more} in $\DU(v, \tau)$ procedure:

\begin{leftbar}
  \begin{invariant}
    \label{invar:movementnew}
 Consider a node $v$ and timestep $\tau$ such that $\DU(v,\tau)$ is called.  Let $u$ be the $v$'s
    principal child at timestep $\tau$.  The server mass entering subtree
    $T_{u}$ during the procedure $\DU(v, \tau)$ is at most
    \begin{gather}
      \frac{\gamma - \loss(u, \tau)}{\la^h} + \frac{\gamma}{4H \la^h}. \tag{I6}
    \end{gather}
\end{invariant}
\end{leftbar}

We again use the ordering $\prec$ on pairs $(v, \tau)$ and assume that the above two invariant conditions holds for all $(v, \tau) \prec (\vstar, \tst)$. We now outline the main changes needed in the analysis done in~\Cref{sec:proving-invariants}. 
~\Cref{cl:unmarkeddefined} still holds with the same proof. We can
again define $\fll(u, \tau), \tau \in \cR^{ns}(u)$ as
in~\eqref{eq:filldef}. Note that $\tau' \in \fll(u, \tau)$ only if $\DU(v,\tau)$ is called, where $v$ is the parent of $u$.~\Cref{cl:fill} still holds with the same proof. %
 The statement of~\Cref{cl:serverupperbound} changes to the following: 

\begin{claim}
  \label{cl:serverupperboundnew}
  Let $\tau \in \cR(u)$ for some $\tau < \tst$. The server mass
  entering $T_u$ at timestep $\tau$ is at most
   \[ \left( 1 + \frac{1}{\lambda-1} \right) \left(1 + \frac{1}{4H} \right) \frac{\Gamma(u, \tau)}{\lambda^{\level(u)}} -
   \frac{\loss(u, \tau)}{\lambda^{\level(u)}}.  \]
\end{claim}

\begin{proof}
  Consider the iteration of the {\bf while} loop
  of~\Cref{algo:mainnew} corresponding to timestep $\tau$. First
  consider the case when $u$ happens to be $v_i, i \geq i_0$. In this
  case, $\Gamma(u, \tau) = \gamma$. The result follows as in the proof
  of \Cref{cl:serverupperbound}, where the extra term of
  $\frac{\gamma}{4 H \la^h}$ arises because of line~\eqref{l:more} in
  the \DU procedure.

 Now consider the case when $u$ is a vertex of the form $v_i, i <
 i_0$.  Note that $\frac{\gamma}{\la^{h_0}} = \frac{\Gamma(v_i,
   \tau)}{\la^h}$, and so the result follows in this case as well by
 using~\Cref{invar:movementnew}, and the quantity $\loss(u, \tau) = 0$ here. 
\end{proof}

The classification of $g(w, \tau)$ into
$ g^\loc(w, \tau), g^\inh(w, \tau)$ holds as before.  The statement
of~\Cref{lem:rel1} changes as given below, and the proof follows the
same lines. We assume that $\DU(\vstar, \tst)$ is called.

\begin{lemma}
  \label{lem:rel1new}
Let $u$ be a non-principal child of $\vstar$ at timestep $\tst$, and $I := (\tau_1, \tst]$ for some timestep $\tau_1 < \tst$. 
  Let
  $S$ be the timesteps in $\cR^{ns}(u) \cap (\tau_1, \tst]$
  that have been removed from $\awake(u)$  
  by the moment when $\DU(\vstar,\tst)$ is called. 
  Then
  \begin{gather*}
    g^\inh(u, (\tau_1,\tst]) \geq \left( 1 + \frac{1}{2H} \right) \,
    \frac{\Gamma(u,S)}{\lambda^{\level(u)}} - \sum_{\tau \in S} \frac{\loss(u,
      \tau)}{\lambda^{\level(u)}},
  \end{gather*}
where $\Gamma(u,S) = \sum_{\tau \in S} \Gamma(u, \tau)$. 
\end{lemma}

The statement and proof of~\Cref{cor:1} remains unchanged.  The proof
of~\Cref{cl:u0} also remains unchanged, though we now need to
use~\Cref{cl:matchinggamma} (part~(iii)).  We now
restate the analogue of~\Cref{cl:dualtr}:
\begin{claim}[Inductive Step Part I]
  \label{cl:dualtrnew}
  Consider the call $\DU(\vstar, \tst)$, and let $u_0$ be 
  the principal child of $\vstar$ at this timestep. Suppose
  $\sib(u_0, \tst) \neq \emptyset$. Then the dual objective value
  corresponding to the constraints in $\cons^\vstar(\tst)$ equals
  $\gamma$; i.e.,
  \[ \sum_{C \in \cons^\vstar(\tst)} z_C \, b^C = \gamma. \] Moreover, 
  the server mass entering $T_{u_0}$ going to the request node in this
  call is at most
  \[  \frac{\gamma - \loss(u_0, \tst)}{\lambda^h} + \frac{\gamma}{4H \la^h}. \]  
\end{claim}

Since the update rule for the $\yv$ variables in
line~\eqref{l:draise2new} of the \DU procedure does not consider the intervals $I_u'$ (as stated in~\eqref{eq:consnew}), %
the proof proceeds 
 along the same lines as that of \Cref{cl:dualtr}. The extra additive term of $\frac{\gamma}{4H \la^h}$ appears due to line~\eqref{l:more} in \DU procedure. 
 The statement and proof of~\Cref{cl:0} remain unchanged. This shows that the two invariant conditions~\ref{invar:dual-raisednew} and~\ref{invar:movementnew} are satisfied. Finally, we state the analogue of~\Cref{cor:2} which bounds the parameter $M$. 

\begin{corollary}
\label{cor:2new}
For node $u$ and timestep $\tau$, let $\tau_u :=
\unmarked(u,\tau)$. %
There are at most $\frac{5H \la^{2H} k}{2 \gamma}$ timesteps in $(\tau_u, \tau] \cap \cR^{ns}(u)$. So  we can set $M$ to  $\frac{5H \la^{2H} k}{2 \gamma}+1$.
\end{corollary}

\begin{proof}
The proof proceeds along the same lines as that of~\Cref{cor:2}, except that the analogue of~\eqref{eq:g-minus-r} now becomes:
\begin{align*}
    g^\inh(u,I) - r(u,I) \geq 
    \left( \frac{1}{2H} - \frac{1}{\lambda-1}\left(1 + \frac{1}{4H}\right) \right) \frac{\Gamma(u,S)}{\lambda^h} \geq \frac{2}{5H}\cdot \frac{\gamma |S| }{\lambda^{2H}},
  \end{align*}
  where the last inequality uses~\Cref{cl:gammarange}(i). This implies the desired upper bound on $|S|$. 
\end{proof}

This shows that the algorithm~\DU is well defined. Next we give
properties of the charging forest, and then show that the dual
variables in each of the local LPs are near-feasible.

\subsection{Properties of the Charging Forest}
\label{sec:prop-charg-forest}

Fix a vertex $v$ and consider the charging forest $\chF(v)$. %
Recall the notation from \S\ref{sec:charging-forest}: a node in this
forest is a tuple $a_i=(w_i,q_i)$, and has a corresponding leaf
request $\leafreq(w_i,q_i) = (\ell_i, R_i = [b_i, e_i])$. We begin
with a monotonicity property, which is useful to show that
\eqref{eq:consz} is properly defined (and the $\btau$ values are
``monotone'').

\begin{claim}[Monotonicity I]
  \label{cl:alive1}
  Suppose $a_2=(w_2, q_2)$ is the parent of $a_1=(w_1,q_1)$ in the
  forest
  $\chF(v)$. %
Then, %
$e_1 \leq e_2$. 
\end{claim}

\begin{proof}
  By definition of an edge in $\chF(v)$, the request interval $R_2$
  for $a_2$ must contain $q_1$. At time $q_1$, the function
  $\FL(w_2, q_1)$ would have returned $w_1$ as one of the vertices in
  the set $S$ (i.e., $w_2$ would have \emph{spawned} $w_1$ at time
  $q_1$), and $R_2$ was also an request below $w_2$ at time $q_1$, so
  $R_1$ must end no later than $R_2$ does.
\end{proof}

\begin{corollary}[Monotonicity II]
  \label{cor:alive}
  Let $a_i=(w_i,q_i)$ belong to a tree $\cT$ of $\chF(q)$ rooted at
  $(v,q)$.  Then $q_i < e_i \leq
  q$. %
\end{corollary}
\begin{proof}
  The first fact uses that $R_i$ is active at time $q_i$. The second
  fact follows by repeated application of~\Cref{cl:alive1}, and that
  the earliest leaf request at the root $(v,q)$ corresponds to the
  request critical at time $q$, which ends at $q$.
\end{proof}

The next result shows another key low-congestion property of the
charging forest, which we then use to build lower bounds for any
single-server instance.

\begin{lemma}[Low Congestion II]
  \label{lem:lowcongestion}
  Consider $u,v$ such that $u \in T_v$. Let $\cT_1, \ldots, \cT_l$ be
  distinct charging trees in the forest $\chF(v)$, where $\cT_j$ is
  rooted at $(v,q_j)$ and contains a \emph{leaf} vertex
  $a_j = (u, q_j')$, with the corresponding $\leafreq(a_j)$ denoted
  $(\ell_j, R_j=[b_j,e_j])$.  Then (a) the intervals
  $\{(q_j',q_j]\}_{j \in [l]}$ have congestion at most $H$, and (b)
  the set of intervals $\{(b_j,q_j']\}_{j \in [l]}$ are mutually
  disjoint. Therefore, the set of intervals
  $\{(b_j,q_j]\}_{j \in [l]}$ have congestion at most $H+1$.
\end{lemma}
\begin{proof}
  Consider a timestep $t$, and let $S_t \subseteq [l]$ be the set of
  indices $j$ such that $t \in(q_j',q_j]$. For each $j \in S_t$,
  consider the path $P_j$ from $a_j=(u,q_j')$ to the root $(v,q_j)$ in
  $\cT_j$. For sake of concreteness, let this path be
  $(u_j^1,q_j^1)=(u,q_j'), (u_j^2, q_j^2), \ldots, (u_j^{n_j},
  q_j^{n_j}) =(v,q_j)$. Since $t \in (q_j',q_j]$, there is an index
  $i$ such that $t \in (q_j^{i-1}, q_j^{i}]$---call this index
  $i(j)$.

  \begin{subclaim}
    For any two distinct $j,j' \in S_t$,
    $u_{j}^{i(j)} \neq u_{{j'}}^{i(j')}$.
  \end{subclaim}

  \begin{subproof}
    Suppose not. For sake of brevity, let $w$ denote
    $u_{j}^{i(j)} = u_{{j'}}^{i(j')}$, $i$ denote $i(j)$ and $i'$
    denote $i(j')$. Assume wlog that $q_{j'}^{i'} < q_{j}^{i}$. So
    $\cT_j$ and $\cT_{j'}$ have vertices $(w,q_j^i)$ and
    $(w',q_{j'}^{i'})$ respectively. Now consider the child of
    $(u_j^{i-1}, q_j^{i-1})$ of $(w,q_j^i)$. The rule for adding edges
    in $\cT_j$ states that we look at the highest time $q' < q_j^i$
    for which there is a vertex $(w,q')$ in the charging forest, and
    so $q_j^{i-1} = q'$. Also $q' \geq q_{j'}^{i'}$. But then, the
    intervals $(q_j^{i-1}, q_j^i]$ and $(q_{j'}^{i'-1}, q_{j'}^{i'}]$
    are disjoint, which is a contradiction because both of them
    contain $t$.
  \end{subproof}

  Since all these $u_j^{i(j)}$ vertices must lie on the path from $u$
  to $v$ in $T$, there are only $H$ of them. Since they are distinct
  by the above claim, the number of intervals containing $t$ is at
  most $H$, which proves the first statement.

  To prove the second statement, assume w.l.o.g.\ that
  $\{q_j\}_{j \in [l]}$ are arranged in increasing order. It suffices
  to show that for any $j \in [l]$, $(b_j, q_j']$ and
  $(b_{j+1}, q_{j+1}']$ are disjoint.  Suppose not. Since
  $e_{j+1} \geq q_{j+1}'$ (by \Cref{cor:alive}), we have that
  $R_{j+1}$ contains $q_j'$. Since $(u,q_{j+1}')$ has no children
  $\cT_{j+1}$, the construction of the charging forest means
  $\FL(u,q_j')$ should have returned the set $S = \emptyset$.
  Now since $v$ is added to the set $Z_{q_j'}$, all the active
  requests---in particular $R_{j+1}$--- below $u$ at time $q_j'$ would
  be serviced at time $q_j'$ due to line~\eqref{l:serverest1}). This
  contradicts the fact that $R_{j+1}$ is active at time
  $q_{j+1}' > q_j'$.
\end{proof}

\subsection{Dual Feasibility of \texorpdfstring{$\DUZ$}{SimpleUpdate}}
\label{sec:dual-feasibility-duz}

Fix a vertex $v$ and the local LP $\fL^v$, which has variables
$\yv(u, \tau)$ for $u \in T_v$ and timesteps $\tau$. First, consider
only the constraints in $\fL^v$ added by the $\DUZ$ procedure, i.e.,
using~\eqref{eq:consz}.

\begin{theorem}
  \label{thm:duzdual}
  For a variable $\yv(\ust, \tst)$, let $\Sst$ be the set of timesteps
  $\tau$ such that \DUZ is called on $v$, and the (unique) constraint
  $C_\tau$ (of the form~\eqref{eq:consz}) that it adds to
  $\fL^v(\tau)$ contains the variable $\yv(\ust, \tst)$ on the LHS. Then:
  \begin{align}
    \label{eq:dualzmain}
    \sum_{\tau \in \Sst} z(C_\tau) \leq O(H^4) \cdot c_{\ust}. 
  \end{align}
\end{theorem}

\begin{proof}
  In the call to \DUZ, we use the charging tree rooted at a vertex
  $v_\ist$, where $v_\ist$ lies between $v$ and $\ust$. Motivated by
  this, for a node $v'$ on the path between $\ust$ and $v$, let
  $\Sst(v')$ denote the subset of timesteps $\tau$ for which the
  corresponding vertex $v_\ist$ is set to $v'$. We show:

  \begin{subclaim}
    \label{clm:duzdual}
    For any $v'$ on the path between $\ust$ and $v$,
    $$\sum_{\tau \in \Sst(v')} z(C_\tau) \leq O(H^3) \cdot
    c_{\ust}.$$ 
  \end{subclaim}

  \begin{subproof}
    Consider a timestep $\tau \in \Sst(v')$. Let $\cT_q$ be the tree
    in the charging forest $\chF(v')$ containing $(v',q)$ for
    $q = \floor{\tau}$. As described when defining~\eqref{eq:consz},
    let $\cL_{s_\tau}$ be the leaves of $\cT_q$ corresponding to level
    $s_\tau$. Since the variable involving $\ust$ appears in this
    constraint, we have $\ust \in \cL_{s_\tau}$. Therefore, each leaf
    in $\cL_{s_\tau}$ has cost equal to $c_\ust$, By the choice of
    level, the total cost of this set is at least $\frac{c_{v'}}{H}$,
    so $$ |\cL_{s_\tau}| \geq \frac{c_{v'}}{H
      c_{\ust}}.$$ Moreover, the tree $\cT_q$ was not a
      singleton so $c_{\ust} \leq c_{v'}/\la$, and since
      $\la \geq 10H$, we get $|\cL_{s_\tau}| \geq 10$. Recall that
    we set
    $$ z(C_\tau) = \frac{\xi_\tau}{|\cL_{s_\tau}| - k_{v',\tau_{v'}} -
      2\delta(n-n_{v'})} \leq \frac{2\xi_\tau}{|\cL_{s_\tau}|},$$
    where the inequality  uses that all leaves below $v$ (except for the requested
    leaf) have at most $\delta$ servers, and so
    $k_{v',\tau_{v'}} + 2\delta(n-n_{v'}) \leq 1 + 2n\delta \ll |\cL_{s_\tau}|/2$. Combine
    the above two facts, for any time $q$,
    \begin{align}
      \label{eq:singleq}
      \sum_{\tau: \tau \in \Sst(v'), \floor{\tau} = q} z(C_\tau) \leq
      \frac{2H c_\ust}{c_{v'}} \cdot \sum_{\tau: \tau \in \Sst(v'),
      \floor{\tau} = q}  \xi_\tau \leq 24H^2 \; c_{\ust},
    \end{align}
    where the last inequality follows
    from~\Cref{lem:istar}. 
    
    We need to sum over different values of $q$, so consider
    $Q := \{\floor{\tau} \mid \tau \in \Sst(v')\}$.  For each value
    $q_j \in Q$, choose $\tau_j$ to be one representative timestep in
    $\Sst$ (in case there are many). For each $q_j \in Q$, there is a
    vertex of the form $(u^\star_j, q_j')$ in the tree in $\chF(v')$
    rooted at $(v',q_j)$. The constraint $C_{\tau_j}$ involves
    $\yv(\ust,(b_j,\tau_j])$, where $b_j$ is the left end-point of the
    leaf request $\leafreq(\ust, q_j')$. And $\tst$ belongs to all of
    these intervals $(b_j, \tau_j]$. Write $(b_j, \tau_j]$ as
    $(b_j, q_j] \cup (q_j, \tau_j]$. The intervals $(q_j, \tau_j]$ are
    mutually disjoint for all $j$, and the intervals $(b_j, q_j]$ have
    congestion at most $H+1$ by Congestion Lemma II
    (\Cref{lem:lowcongestion}). Therefore the intervals $(b_j,\tau_j]$
    have congestion at most $H+2$, and so $|Q| \leq H+2$. Combining
    this with~\eqref{eq:singleq} completes the proof.
  \end{subproof}
  \noindent \Cref{thm:duzdual} follows from the above claim by summing over all
  $v'$ on the path from $\ust$ to $v$.
\end{proof}

\subsection{Dual Feasibility}
\label{sec:dual-feasibility-du}

In this section, we show approximate dual feasibility of the entire
solution due to both the \DUZ and \DU procedures. The proof is very
similar to the one for \kser (in \S\ref{sec:appr-dual-feas}) except
for two changes: (i) we need to account for the intervals $I_u'$ as
in~\eqref{eq:consnew}, and (ii) we have defined new sets of
constraints in \DUZ procedure, so the result of~\Cref{thm:duzdual}
needs to be combined with the overall dual feasibility result. Observe
that the statements of~\Cref{cl:rec} and~\Cref{cl:uppery} hold without
any changes. We now prove the analogue of~\Cref{lem:dual}. First, a
simple observation.

\begin{claim}
  \label{cl:duflowbound}
  Consider a time $q$ and vertex $v$ such that a critical request at
  time $q$ lies below $v$. Then the total objective value of dual
  variables raised during $\DU$ procedure at $v$ is at most
  $4H c_v/\la$.
\end{claim}

\begin{proof}
  Let the height of $v$ be $h+1$. During each call of $\DU$ at $v$, we
  raise the dual objective by $\gamma$ units (by \Cref{cl:dualtrnew}),
  and transfer at least
  $\frac{\gamma}{4H \la^{h}} = \frac{\gamma\la}{4H c_v}$ server mass
  to the requested leaf node $\rqq$. Since we transfer at most one unit of
  server mass to $\rqq$, the result follows.
\end{proof}

\begin{theorem}[Dual Feasibility for Time-Windows]
  \label{lem:dualnew}
  For a node $v$ at height $h+1$, consider the dual variables $z(C)$
  for constraints added in $\fL^v$ during the \DU and the \DUZ
  procedure. These
  dual variables $z_C$ are $\beta_{h}$-feasible for the dual program
  $\fD^v$, where
  $\beta_h = \left( 1+ \nicefrac{1}{H} \right)^h O(H^4 + H(\ln n + \ln M
  + \ln (k/\gamma))). $
\end{theorem}

\begin{proof}
  We prove the claim by induction on the height of $v$. For a leaf
  node, this follows vacuously, since the primal/dual programs are
  empty. Suppose the claim is true for all nodes of height at most
  $h$.  For a node $v$ at height $h+1 > 0$ with children $\chi_{v}$,
  the variables in $\cons^v$ are of two types: (i)~$\yv(u,t)$ for
  some timestep $t$ and child $u \in \chi_{v}$, and
  (ii)~$\yv(u',t)$ for some timestep $t$ and non-child descendant
  $u' \in T_v \setminus \chi_{v}$.
  
  We consider variables of the first type. Fix a child $u$ and a
  timestep $\tau$, and let $\constr^F$ be the set of constraints in
  $\fL^v$ added in $\DU$ subroutine that contain the variable
  $\yv(u,\tau)$ in the LHS.  We group $\constr^F$ into three classes
  of constraints (and draw on the notation in~\eqref{eq:consnew}):
  \begin{itemize}
  \item[(i)] The timestep $t$ lies in $I_u = (\tau_u, \tau]$, where
    $u$ is a non-principal child of $v$ at the timestep $\tau$ at
    which this constraint is added: call this set of constraints
    $\constr_1(t)$.  The argument here is identical to that in the
    proof of~\Cref{lem:dual}, and so we get
    $$ \sum_{C \in \constr_1} z_C \leq O(\ln n + \ln M + \ln(k/\gamma)) c_u. $$
  \item[(ii)] The timestep $t$ lies in the interval
    $I_u'=(b_u,\tau_u]$, where $u$ is a non-principal child of $v$ at
    timestep $\tau_u$. Denote the set of such constraints by
    $\constr_2 = \{C_1, \ldots, C_s\}$. For sake of concreteness, let
    the interval $I_u'$ in $C_j$ be $I_j' = (b_j, \tau_j],$ and let
    $I_j$ denote the corresponding $I_u$ interval. Observe that
    $\tau_j$ corresponds to a $\bot$-constraint in $\cL^u$, and so always remains in $\awake(u)$.  Let $q_j$ denote $\floor{\tau_j}$. Note that any
    constraint in $\constr_2$ must contain one of the variables
    $\yv(u, q_j+1), j \in [s]$, and each of these variables belongs to  the
    corresponding $I_j$ interval. So if $X$ denotes the set
    $\{q_j: j \in [s]\}$, then
    \begin{align}
      \label{eq:constr2}
      \sum_{C \in \constr_2} z(C) \leq \sum_{t' \in X} \sum_{C \in \constr_1(t'+1)} z(C)  \leq  O(\ln n + \ln M + \ln(k/\gamma)) \cdot |X| \cdot c_u,
    \end{align}
    where the last inequality follows from case~(i) above. It remains
    to bound $|X|.$ We know from the Congestion Lemma~I (\Cref{cl:lowcong1}) that the intervals
    $(b_j, q_j], q_j \in X,$ have congestion at most $H$. Since the
    intervals $(q_j,q_j+1], q_j \in X,$ are mutually disjoint, it
    follows that the intervals $(b_j,q_j+1]$ have congestion at most
    $H+1$. Since all of them contain the timestep $t$, it follows that
    $|X| \leq H+1$. This shows that
    $$ \sum_{C \in \constr_2} z(C) \leq O(H(\ln n + \ln M + \ln(k/\gamma)))  c_u. $$
  \item[(iii)] The timestep $t$ lies in the interval
    $I_u'=(b_u, \tau]$, where $u'$ is the principal child of $u$ at
    timestep $\tau$: call such constraints $\constr_3$. For a time
    $q$, let $\constr_3(q)$ be the subset of constraints in
    $\constr_3$ which were added at timesteps $\tau$ for which
    $\floor{\tau} = q$. \Cref{cl:duflowbound} shows that
    $$ \sum_{C \in \constr_3(q)} z(C) \leq 4H \frac{c_v}{\la} = 4H\,c_u. $$
    Arguing as in case~(ii) above, and again using Congestion Lemma~I (\Cref{cl:lowcong1}), we see
    that there are at most $O(H)$ distinct time $q$ such that the set
    $\constr_3(q)$ is non-empty. Therefore,
    $$ \sum_{C \in \constr^F} z(C) \leq O(H^2) c_u. $$
  \end{itemize}
  
  Let $\constr$ be the set of all constraints containing $\yv(u,t)$.
  Combining the observations above, and using \Cref{thm:duzdual} for
  the constraints in $\constr$ added due to the $\DUZ$ procedure, we
  see that
  $$ \sum_{C \in \constr} z(C) \leq \beta_0\, c_u. $$
  It remains to consider the variables $\yv(u',\tau)$ with
  $u' \in T_u$ and $u \in \chi_{v}$. The argument here follows from
  induction hypothesis, and is identical to the one in the proof
  of~\Cref{lem:dual}.
\end{proof}

\subsection{The Final Analysis}

We can now put the pieces together: this part is also very similar to
\S\ref{sec:analysis}, except for the cost of the piggybacking trees.
Recall that $\lambda \geq 10H$.

\begin{enumerate}
\item \Cref{lem:dualnew} shows that the dual solution for the global
  LP (which is the same as the $\fL^{\rootvtx}$) is
  $O(H^4 + H \log \nf{Mnk}{\gamma})$-feasible. In each iteration of the {\bf while} loop
  in~\Cref{algo:mainnew}, we raise the dual objective corresponding to
  this LP by $\gamma$ units, as \Cref{cl:dualtrnew} shows.

\item 
  The total service cost in each call to the \DU procedure is
  $O(\gamma)$---again by~\Cref{cl:dualtrnew}, the amount of server
  mass transferred during $\DU$ procedure at vertex $v$ at height
  $h+1$ is at most $\frac{\gamma}{\la^h}$, and the cost of moving a
  unit of server mass below $T_v$ is $O(\la^h)$. Therefore, the cost
  to service the critical request in each iteration is $O(\gamma
  H)$. Therefore, the service cost for each critical request is $O(H)$
  times the dual objective value.

\item Now we consider the service cost for piggybacked
  requests. The cost of all the trees is dominated by the cost of
  tree for $v_{\lcost(q)}$, i.e., at most \[ O(H^2 c_{v_{\lcost(q)}}) =
    O(H^2 \; \la^{\lcost(q)}) \leq O(H^2 \la \; \cost(q)). \]
  Since $\cost(q)$ is the least cost to move the required amount of
  server to the request location, the cost of the trees is at most
  $O(H^2 \la)$ times the cost incurred in the previous step.

\end{enumerate}
Hence the competitiveness is
\[ O(H^4 + H \log \nf{Mnk}{\gamma}) \cdot O(H) \cdot O(H^2 \la ). \]
It follows that our fractional algorithm for the \ksertw problem is
$O(H^4\la (H^3 + \log  \nf{Mnk}{\gamma}))$-competitive. %
This proves \Cref{thm:frac-tw}.


\section{Closing Remarks}
\label{sec:open}

Our work suggests several interesting directions for future
research. 
Can our LP 
extend 
to variants and generalizations of \kser
 in the literature? 
One natural candidate is the hard 
version of the $k$-taxi problem.
Another interesting direction is to exploit the 
fact that our LP easily extends to time-windows. The special case of
\ksertw where $k=1$ is known as \emph{online service with 
delay}. While poly-logarithmic competitive ratios are
known for this problem (and also follow from our current work),
no super-constant lower bound on its 
competitive ratio bound is known. On the other hand, a sub-logarithmic 
competitive ratio is not known even for simple metrics like the line.
Can our LP (or a variant) bridge this gap?

More immediate technical questions concern the \ksertw problem itself.
For instance, can the 
competitive ratio of the \ksertw problem be improved from 
$\poly\log(n,\Delta)$ to $\poly\log(k)$? Another direction is to 
extend \ksertw to general delay penalties. 
Often, techniques for time-windows extend to general delay functions 
by reducing the latter to a prize-collecting version of the time-windows
problem. Exploring this direction for \ksertw would be a useful
extension of the results presented in this paper.

\subsection*{Acknowledgments}

AG was supported in part by NSF awards CCF-1907820, CCF-1955785, and CCF-2006953.
DP was supported in part by NSF awards CCF-1750140 (CAREER) and
CCF-1955703, and ARO award W911NF2110230.

{\small
\bibliographystyle{alpha}
\bibliography{paper,server,ref}
}

\appendix
\section{Relating \texorpdfstring{$\fM$}{MainLP} to the Min-cost Flow Formulation for \texorpdfstring{\kser}{k-server}}
\label{sec:min-cost-relation}

We show that the 
constraints \eqref{eq:basic-cons} of the LP relaxation $\fM$ for \kser are implied by the standard
min-cost flow formulation for  \kser on HSTs. 

We first describe the min-cost flow formulation in detail. Consider an instance of the \kser problem consisting of an HST $T$, and a sequence of $N$ request times. 
Recall that the set of timesteps $\cT$ varies from $1$ to $N$ in steps of $\eta$.  
 We construct
a time-expanded graph $G$ with vertices
$V(G) := \{ v_t \mid v \in V(T), t \in \cT \} \cup \{s_G, t_G\}$. 
The edges are of three types (there are no edge-capacities):
\begin{OneLiners}
\item[(i)] cost-$0$ edges $\{(s_G,v_1), (v_N, t_G) \mid v \in V(T)\}$
  connecting the source and sink to the first and last copies of each
  node, 
\item[(ii)] cost-$0$ edges
  $\{(v_t, v_{t+1}) \mid v \in V(T), t \in \cT \}$ between consecutive
  copies of the same vertex, and 
\item[(iii)] edges
  $\{(v_t, p(v)_t) \mid v \in V(T), t \in \cT \}$ of cost $c_v$
  between each node and its parent, and
  $\{(v_t, u_t) \mid v \in V(T), u \in \chi_v, t \in \cT \}$ of cost
  zero between a node and its children. (This captures that moving
  servers up the tree incurs cost, but moving down the tree can be done
  free of charge.)
\end{OneLiners}
The source $s_G$
has $k$ units of supply, and sink $t_G$ has $k$ units of demand (or
equivalently, a supply of $-k$). If the request for time $q$ is at
leaf $\ell$, we require that at least one unit of flow passes through
$\ell_q$. To model this, we assign a supply of $-1$ to $\ell_q$ and
$+1$ to $\ell_{q+1}$. This is consistent with the proof of~\Cref{cl:lp} where we assumed that after servicing $\ell$ at time $q$, the server stays at this leaf till time $q+1$. 

The integrality of the min-cost flow
polytope implies that an optimal solution to this transportation
problem captures the optimal $k$-server solution. Moreover, the
max-flow min-cut theorem says that if $x$ is a solution to this
transportation problem, then for all subsets $S \sse V(G)$,
\begin{align}
    \label{eq:supply}
     x(\partial^+(S)) \geq \suppl(S). 
\end{align}

We now consider special cases of these
constraints. Consider a tuple $(A, \btau)$ corresponding to the LP
constraint~\eqref{eq:basic-cons}: recall that $A$ is a subset of
leaves, and $\btau$ assigns a timestep $\tau_u$ to each $u \in T^A$,
with these timesteps satisfying the ``monotonicity'' constraints stated
before~\eqref{eq:basic-cons}.  
We now define a set $S_{A,\btau}$ as follows: 
for each node $v \in T^A$, we add the nodes $v_t, t > \tau_v$ to $S_{A, \btau}$. 
 Finally, add the sink $t_G$ to $S_{A,\btau}$ as
well. Since each leaf in $A$ contributes $+1$ to the supply of $S_{A, \btau}$, and $t_G$
contributes $-k$, we have $\suppl(S_{A,\btau}) = |A|- k$. Moreover, 
$$ x(\partial^+(S_{A,\btau})) = \sum_{v \neq \rt: v \in T^A} x(v, (q_v, q_{p(v)}]). $$
Thus, constraint~\eqref{eq:supply} for the set $S_{A,\btau}$ is identical
to the covering constraint $\varphi_{A,\btau}$ given by~\eqref{eq:basic-cons}.



\end{document}